\documentclass{llncs}
\usepackage{graphicx}
\usepackage{graphics}
\usepackage{pstricks}
\usepackage{array}
\usepackage{amsmath}
\usepackage{amssymb}
\usepackage{hyperref}
\hypersetup{
  colorlinks   = true,    
  linktoc	   = all,
  linktocpage  = true,
  urlcolor     = blue,    
  linkcolor    = blue,    
  citecolor    = red      
}
\usepackage{array}
\usepackage{comment}
\usepackage{setspace}
\usepackage{algorithm}
\usepackage{epstopdf}       
\usepackage{enumerate}     
\usepackage{tikz}

\usetikzlibrary{shapes.geometric, arrows, snakes}

\def \CB#1#2 {\fill[color=black] (#1,#2) circle(2pt);}
\def \N#1#2#3 {\draw[color=black] (#1,#2) node {$#3$};}
\def \LB#1#2#3#4 {\draw[color=black,line width=0.3mm] (#1,#2) -- (#3,#4);}
\def \LR#1#2#3#4 {\draw[color=red,line width=0.3mm] (#1,#2) -- (#3,#4);}
\def \LRA#1#2#3#4 {\draw[line width=0.3mm,color=red,->] (#1,#2) -- (#3,#4);}
\def \LBA#1#2#3#4 {\draw[line width=0.3mm,color=blue,->] (#1,#2) -- (#3,#4);}
\def \LBLA#1#2#3#4 {\draw[line width=0.3mm,->] (#1,#2) -- (#3,#4);}
\def \NT#1#2#3 {\draw[color=black] (#1,#2) node {#3};}
\def \alert {\textcolor{red} }

\newcolumntype{L}[1]{>{\raggedright\let\newline\\\arraybackslash\hspace{0pt}}m{#1}}
\newcolumntype{C}[1]{>{\centering\let\newline\\\arraybackslash\hspace{0pt}}m{#1}}
\newcolumntype{R}[1]{>{\raggedleft\let\newline\\\arraybackslash\hspace{0pt}}m{#1}}

\def \b {\textcolor{blue} }
\def \u {\textcolor{black} }

\def\NC#1#2{$N^I(#1)\cup\{#1,#2\}$}
\def \K#1#2{$K_{#1,#2}$}
\def \NV#1{$N^I(#1)\cup \{#1\}$}
\def \nin {\notin}

\newtheorem{cl}{Claim}
\usepackage{algcompatible}   
 
\date{}
\title{Hamiltonian Path in Split Graphs- a Dichotomy}
\author{P.Renjith, N.Sadagopan} 
\institute{Indian Institute of Information Technology, Design and Manufacturing, Kancheepuram, India. \\
\email{\{coe14d002,sadagopan\}@iiitdm.ac.in}}
\begin{document}
\maketitle
 \begin{abstract}
In this paper we investigate Hamiltonian path problem in the context of split graphs, and produce a dichotomy result on the complexity of the problem.  Our main result is a deep investigation of the structure of $K_{1,4}$-free split graphs in the context of Hamiltonian path problem, and as a consequence, we obtain a polynomial-time algorithm to the Hamiltonian path problem in $K_{1,4}$-free split graphs.  We close this paper with the hardness result: we show that, unless P=NP, Hamiltonian path problem is NP-complete in $K_{1,5}$-free split graphs by reducing from Hamiltonian cycle problem in $K_{1,5}$-free split graphs.  Thus this paper establishes a ``thin complexity line" separating NP-complete instances and polynomial-time solvable instances. 
 \end{abstract}
\section{Introduction}
%
Hamiltonian path problem is a well studied problem of finding a spanning path in a connected graph.  Hamiltonian path problem has substantial scientific attention in modelling various real life problems, and finds applications in physical sciences, operational research \cite{app_hp1,app_hp2}, etc.  This problem has been studied in various perspectives.  In the initial stages of study, researchers explored the problem on structural perspective.  That is, necessary conditions and sufficient conditions for the existence of Hamiltonian paths in connected graphs.  Further, special graphs with bounded graph parameters such as degree, toughness, connectivity, independence number, etc., have been explored for obtaining Hamiltonian paths \cite{s1}.  Another interesting view on the Hamiltonian problems have been obtained on graphs with forbidden sub graph structures.  For example, Hamiltonian paths in claw-free graphs and its sub classes have been explored \cite{s3}.  Variants of Hamiltonian problems such as Hamiltonian path starting from a specific vertex, Hamiltonian path between a fixed pair of vertices, Hamiltonian connectedness, pancyclicity, etc., have also been explored in the literature.  A detailed survey has been compiled by Broersma and Gould \cite{s1,s2,s3}. \\\\
%
On algorithmic perspective, the problem is NP-complete in general graphs, and in particular, special graph classes such as chordal \cite{bertossi}, bipartite, chordal bipartite \cite{muller}, planar \cite{tarjanplanar}, grid graphs \cite{Gordon}, etc.  On the other hand, polynomial-time results for the problem have been obtained for interval \cite{keil,hungint}, circular arc \cite{hungcir,shih}, proper interval \cite{panda,ibarra}, distance hereditary\cite{hungdis}, cocomparability graphs \cite{HPcocomparability}, complete multipartite graphs \cite{gutin}, etc.  It is important to note that although polynomial-time results are known for special graph classes, we still have a ``thick complexity line" separating NP-complete instances and polynomial-time solvable instances.
For instance, Hamiltonian path problem in chordal graphs is NP-complete and a maximal graph class which is a subclass of chordal graph for which a polynomial-time algorithm is known is the class of interval graphs.   However, the complexity line separating chordal graphs (NP-complete instance) and interval graphs (polynomial-time instance) for Hamiltonian path problem is thick.   It is important to highlight that there are infinitely many non-interval chordal graphs, for example, chordal graphs with \emph{asteroidal triple} as a sub graph, on which the complexity of Hamiltonian path is open.  To make this line thin, one must do a micro level analysis of the NP-complete reduction of the Hamiltonian path problem in chordal graphs.  Further, this asks for a deeper study of the structure of chordal graphs.  \\\\
In this paper we revisit the Hamiltonian path problem in chordal graphs and present a tight hardness result.  We attempt a micro level structural study for Hamiltonian path problem in split graphs and establish that Hamiltonian path problem in \K15-free split graph is NP-complete, which is a popular sub class of chordal graphs.  Further, to make the borderline thin between NP-complete instances and polynomial-time instances, we do a deeper investigation of the structure of \K14-free split graphs, which is a major contribution of this paper.  To the best of our knowledge, this line of investigation has not been looked at in the literature.  The only known results in this context are the study of Hamiltonian cycle in \K15-free and \K14-free split graphs \cite{caldam2017}, and the study of Steiner tree in \K15-free and \K14-free split graphs \cite{dicho1}.
%
As a result of our deep structural study, we show that Hamiltonian path problem is polynomial-time solvable in \K14-free split graphs.  This brings an interesting dichotomy for Hamiltonian path problem in split graphs.   \\\\
%
%
The rest of the paper is organized as follows.  We next present the graph preliminaries.  In Section 2 we present the polynomial-time results of the dichotomy.  The hardness result is presented in Section 3.  The concluding remarks and future work are discussed in Section 4. \\\\
We use standard basic graph-theoretic notations.  Further, we follow \cite{west}.  All the graphs we mention are simple, and unweighted.  Graph $G$ has vertex set $V(G)$ and edge set $E(G)$ which we denote using $V,E$, respectively, once the context is unambiguous.  For \emph{independent set}, \emph{maximal clique}, and \emph{maximum clique} we use the standard definitions.  Split graphs are $C_4,C_5,2K_2$-free graphs and the vertex set of a split graph can be partitioned into a clique $K$ and an independent set $I$.  Such a split graph is denoted as $G(K\cup I,E)$.  For a split graph $G(K\cup I,E)$, we assume $K$ to be a maximum clique.  For $S\subset V(G)$, $N(S)=\{u:u\notin S, v\in S, uv\in E(G)\}$.  If $S=\{v\}$, $N(S)$ is also denoted as $N(v)$.  For a split graph $G(K\cup I,E)$ and $S\subset K$ we define $N^I(S)=N(S)\cap I$.  Accordingly, if $S=\{v\}$, $N^I(v)=N^I(S)$.  $d^I(v)=|N^I(v)|$ and $\Delta^I=\max\{d^I(v) : v\in K\}$.  For $S\subset V(G)$, $G-S$ represents the subgraph of $G$ induced on the vertex set $V(G)\setminus S$.  $c(G)$ represents the number of components in graph $G$.
For a cycle or a path $C=(v_1,\ldots,v_n)$, by $\overrightarrow{C}$, we mean the visit of vertices in order $(v_1,\ldots,v_n)$.  Similarly, by $\overleftarrow{C}$, we mean the visit of vertices in order $(v_n,\ldots,v_1)$.  $u\overrightarrow{C}u$ represents the ordered vertices from $u$ to $v$ in $C$.  For a path $P=(v_1,\ldots,v_{n\ge1})$ of length $n$, for simplicity, we use $P$ to denote the underlying set $V(P)$ and $v_1,v_n$ are end vertices of $P$.

\section{Hamiltonian path problem in split graphs : polynomial-time results}
We organize our results on Hamiltonian path as Hamiltonian path in \K13-free split graphs and Hamiltonian path in \K14-free split graphs.  We present our results on \K14-free split graphs in a systematic way.  That is, we shall present Hamiltonian path in \K14-free split graph with $\Delta^I=1$, $\Delta^I=2$ followed by $\Delta^I=3$.  We make use of the following results from the literature to present our results.
\begin{lemma} [\cite{caldam2017}]\label{lemvileq3}
For a \K13-free split graph $G$,  if $\Delta^I=2$, then $|I|\leq 3$. 
\end{lemma} 
%
%
\begin{lemma} [\cite{caldam2017}]\label{k13hamil}
Let $G$ be a $K_{1,3}$-free split graph.  $G$ contains a Hamiltonian cycle if and only if $G$ is 2-connected.
\end{lemma}
\begin{theorem}[\cite{caldam2017}]\label{deltaleq2}
Let $G$ be a $2$-connected, $K_{1,4}$-free split graph with $\Delta^I=2$.  Then $G$ has a Hamiltonian cycle if and only if there are no short cycles in $G$.  
\end{theorem}
\begin{lemma}[Chvatal\cite{west}]\label{chvatal_path}
Let $G$ be a connected graph.  If $G$ has a Hamiltonian path, then for every $S\subset V(G)$, $c(G-S)\le|S|+1$.
\end{lemma}
\begin{lemma}[\cite{caldam2017}]\label{claimA}
For a connected split graph $G$ with $\Delta^I=3$, let $v\in K, d^I(v)=3$, and $U=N^I(v)$.  If $G$ is $K_{1,4}$-free, then $N(U)=K$.
\end{lemma}
\begin{corollary}[of Lemma \ref{claimA}]\label{cor1}
Let $G$ be a connected $K_{1,4}$-free split graph with $v\in K, d^I(v)=3$.  For every vertex $w\in K\setminus \{v\}$, $N^I(v)\cap N^I(w)\ne\emptyset$. 
\end{corollary}

\subsection{Results on \K13-free split graphs}
\begin{theorem} \label{k13hamilpath}
Let $G$ be a connected $K_{1,3}$-free split graph.  $G$ contains a Hamiltonian path if and only if $G$ has at most $2$ vertices $u,v\in I$ such that $d(u)=1$, and $d(v)=1$.
\end{theorem}
\begin{proof}
If there exists at least three vertices $\{u,v,w\}\subseteq I$ such that $d(u)=d(v)=d(w)=1$, then clearly $G$ has no Hamiltonian path.  
For the sufficiency, we see the following cases.  \\
\emph{Case 1:}  For every $u\in I$, if $d(u)\ge2$, then $G$ is 2-connected, and by Lemma \ref{k13hamil}, $G$ has a Hamiltonian cycle.  Thus $G$ has a Hamiltonian path.\\
\emph{Case 2:} If there exists only one vertex $u\in I$ with $d(u)=1$, then observe that $G-u$ is 2-connected.  By Lemma \ref{k13hamil}, there is a Hamiltonian cycle in $G-u$, which can be easily extended to a Hamiltonian path in $G$.  \\
\emph{Case 3:} There exists two vertices $u,v\in I$ with $d(u)=d(v)=1$.  If $I=\{u,v\}$, it is easy to see that there is a  $(u,v)$-Hamiltonian path in $G$.  If $I=\{u,v,w\}$, then we claim that $N(w)\cap N(u)=\emptyset$ and $N(w)\cap N(v)=\emptyset$.  Suppose $N(w)\cap N(u)\neq\emptyset$, then let $N^I(u')=\{u,w\}$, $u'\in K$.  Clearly, all the vertices $x\in K\setminus \{u'\}$ are adjacent to $w$, otherwise $\{u',u,w,x\}$ induces a \K13.  It follows that $K\cup \{w\}$ is a clique of larger size, contradicting the maximality of $K$.  Similar arguments hold with respect to the vertex $v$, and hence $N(w)\cap N(v)=\emptyset$.  Thus we conclude that $\Delta^{I}=1$.  From Lemma \ref{lemvileq3}, if $|I|>3$, since $G$ is connected, $\Delta^{I}=1$.  Now we produce a Hamiltonian path in $G$ with $|I|\ge3$ as follows.  Let $I=\{u,v,w_1,\ldots, w_k\}$, $k\ge1$ such that for all $w_i$, $1\le i\le k$, $d(w_i)>1$.  Let $x_i,y_i, 1\leq i\leq k$ be any two elements in $N(w_i)$.  Since $\Delta^I=1$, note that for all $s,t\in I$, $N(s)\cap N(t)=\emptyset$.  Let $P_i=(x_i,w_i,y_i)$, $1\le i\le k$, $v'=N(v)$, $u'=N(u)$ and $\{z_1,\ldots,z_l\}=K\setminus\{x_1,\ldots,x_k,y_1,\ldots,y_k,u',v'\}$, then $P=(u,u',x_1,w_1,y_1,\ldots,x_i,w_i,y_i,\ldots,x_k,w_k,y_k,z_1, z_2, \ldots, z_l,v',v)$ is a Hamiltonian path in $G$.  $P$ can also be written as $(u,u',\overrightarrow{P_1},\ldots,\overrightarrow{P_k},z_1, z_2, \ldots, z_l,v',v)$. 
This completes a proof of Theorem \ref{k13hamilpath}.  $\hfill \qed$
\end{proof} 
\subsection{Results on \K14-free split graphs}
\begin{theorem} \label{delta1hamilpath}
Let $G(K\cup I,E)$ be a connected $K_{1,4}$-free split graph with $\Delta^I=1$.  $G$ contains a Hamiltonian path if and only if there exists at most $2$ vertices $u,v\in I$ such that $d(u)=1$, and $d(v)=1$.
\end{theorem}
\begin{proof}
The proof is similar to the proof of Case 3 in Theorem \ref{k13hamilpath}. 
\end{proof}
We shall define some special paths and cycles in a $K_{1,4}$-free split graph $G(K\cup I,E)$.  We define the \emph{restricted bipartite subgraph} $H$ of $G$ as follows.  $V_a=\{u\in I : d(u)\le2\}$, $V_b=N(V_a)$, $V(H)=V_a\cup V_b$ and $E(H)=\{uv : u\in V_a, v\in V_b\}$.  An induced cycle $C$ in $H$ is referred to as \emph{short cycle} in $H$ (as well as $G$) if $V(C)\subset V(G)$.  An $I$-$I$ path is a maximal path in $H$ that starts and ends in $I$.   Similarly $K$-$K$ path and $I$-$K$ path are maximal paths in $H$ with end vertices in $K$ and end vertices in $I$, $K$, respectively.  A maximal $I$-$I$ path $P$ in $H$ is referred to as \emph{Short} $I$-$I$ path if $V(P)\subset V(G)$.  An example is illustrated in Figure \ref{eg:shortIIpath}.
\begin{figure}[h!]
\begin{center}
\includegraphics[scale=1.1]{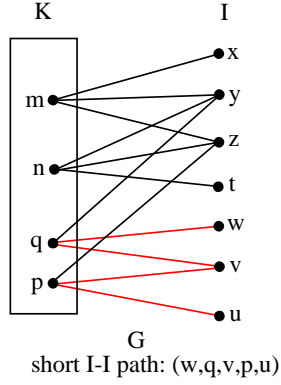}
\caption{Split Graph $G$ having a short $I$-$I$ path}
\label{eg:shortIIpath}
\end{center}
\end{figure} 
%
\begin{theorem} \label{delta2hamilpath}
Let $G(K\cup I,E)$ be a connected $K_{1,4}$-free split graph with $\Delta^I=2$ and $H$ be the restricted bipartite subgraph of $G$.  $G$ contains a Hamiltonian path if and only if the following holds true. \\
1. $H$ has no short $I$-$I$ path.\\
2. The number of $I$-$K$ paths in $H$ is at most 2.
\end{theorem}
\begin{proof}
If there exists a short $I$-$I$ path $P$ in $H$, then note that $c(G-S)>|S|+1$ where $S=P\cap K$, and there is no Hamiltonian path in $G$ as per Lemma \ref{chvatal_path}.  It is easy to see that if the number of $I$-$K$ paths in $H$ is more than 2, then there is no spanning path in $G$ that includes all the vertices in all such $I$-$K$ paths.  For sufficiency, we see the following.  Since $H$ is the restricted bipartite subgraph of $G$, $H$ is a collection of maximal paths and short cycles.  Moreover, any short cycle in $H$ is also a maximal $I$-$K$ path in $H$.  We initialize a set $\mathbb{S}$ with the set of maximal paths in $H$.  It follows that $\mathbb{S}$ has at most two $I$-$K$ paths.  Let $I'=I\setminus\bigcup\limits_{\forall P\in \mathbb{S}}V(P)$ and $K'=K\setminus\bigcup\limits_{\forall P\in \mathbb{S}}V(P)$.  We now outline a procedure to update $\mathbb{S}$ in two stages, using which we construct a Hamiltonian path in $G$.  In the first stage, for every vertex $u\in K'$, which is by definition $P_1$, include $P_1$ in $\mathbb{S}$.  Since $\Delta^I=2$, observe that any vertex $v\in I'$ is not adjacent to any internal vertex of paths in $\mathbb{S}$.  Thus such a vertex $v$ is adjacent to the end vertices of paths in $\mathbb{S}$.  In particular, $v$ may be adjacent to some of the newly added $P_1$ in $\mathbb{S}$ during the first stage.  Further, $d(v)\ge3$ implies that $v$ is adjacent to the end vertices of at least two different paths $Q_i,Q_j\in \mathbb{S}$.  As a part of the second stage, we update $\mathbb{S}$ as follows.  For every vertex $v\in I'$, we find paths $Q_i,Q_j$ such that one of $Q_i,Q_j$ is either a $K$-$K$ path or $P_1$.  The paths $Q_i,Q_j$ are replaced with the path $(\overrightarrow{Q_i},v,\overrightarrow{Q_j})$ in $\mathbb{S}$.  Let $\mathbb{S}_f=\{Q_1,\ldots,Q_k\}$ be the resultant set of paths after completing the second stage.  If there exists two $I$-$K$ paths, then let it be $Q_i,Q_j$, $i<j$ and if there exists only one $I$-$K$ path, then let it be $Q_i$.  Then $(\overrightarrow{Q_i},\overrightarrow{Q_1},\ldots,\overrightarrow{Q}_{i-1},\overrightarrow{Q}_{i+1},\ldots,\overrightarrow{Q}_{j-1},\overrightarrow{Q}_{j+1},\ldots,\overrightarrow{Q_k},\overrightarrow{Q_j})$ is a Hamiltonian path in $G$.    
This completes the sufficiency part and a proof of the theorem.
$\hfill\qed$
\end{proof}
\textbf{Definition:} A connected $K_{1,4}$-free split graph $G$ satisfies Property A if $|K|\ge|I|-1\ge8$, $G$ has no short $I$-$I$ path, and the sum of the number of $I$-$K$ paths and the number of short cycles is at most 2.  In a $K_{1,4}$-free split graph $G$ with $\Delta^I_{G}=3$, we define $V_3=\{v:v\in K,d^I(v)=3\}$. 
\\\\
Consider a $K_{1,4}$-free split graph $G$ with $\Delta^I_{G}=3$.  We shall now show that the number of short cycles in $G$ is at most 1 and the length of short cycle is at most 8.  Subsequently, if $G$ satisfies Property A, then we produce a Hamiltonian path in $G$.  Towards this attempt, we bring in a transformation which will transform an instance of $\Delta^I_G=3$ into $\Delta^I_{G'}=2$ instance $G'$.  Our results are deep and investigates the structure of the restricted bipartite subgraph $H'$ of $G'$ to obtain a Hamiltonian path in $G$. 
%
%
%
\begin{lemma} \label{onecycle}
Let $G$ be a connected $K_{1,4}$-free split graph with $\Delta^I=3$.  Then, the number of short cycles in $G$ is at most one.  Further, if $G$ has a short cycle $C_n$, then $n\le 8$. 
\end{lemma}
\begin{proof}
For a contradiction assume that there are at least two short cycles in $G$.  Let $C,D$ be any two short cycles in $G$.  Since $\Delta^I_{G}=3$, there exists $v\in V_3$.  Clearly, there exists a vertex $v_1\in N^I(v)$ such that $v_1$ is adjacent to all the vertices in $C\cap K$ and $D\cap K$.  It follows that all the vertices in $K\setminus (C\cup D)$ are adjacent to $v_1$.  Note that $K\cup \{v_1\}$ is a clique of larger size, which contradicts the maximality of clique $K$.  For the second part, assume for a contradiction that there exists a short cycle $C_{n\ge10}$.  Consider the cycle $C$ such that $V(C)=\{w_1,\ldots,w_j,x_1,\ldots,x_j\}$, $j\ge5$, $\{w_1,\ldots,w_j\}\subset K$, $\{x_1,\ldots,x_j\}\subset I$, $E(C)=\bigcup\limits_{i=1}^j\{w_ix_i,x_iw_{(i+1)mod~j}\}$.  Since $\Delta^I(G)=3$, there exists $v\in V_3$.  To complete our proof, we identify a vertex $v_1\in N^I(v)$ as follows.  If $v\nin C$, then from Corollary \ref{cor1}, there exists a vertex $v_1\in N^I(v)$ such that $v_1w_1\in E(G)$.  If $v\in C$, then without loss of generality, we assume $w_1=v$.  There exists $v_1\in N^I(v)$ such that $v_1\nin C$.  We claim that the vertices $\{w_3,\ldots,w_{j-1}\}$ are adjacent to $v_1$, otherwise $N^I(w_1)\cup \{w_1,w_i\}$, $3\le i\le j-1$ induces a $K_{1,4}$.  Further $w_2v_1\in E(G)$, otherwise $N^I(w_4)\cup\{w_4,w_2\}$, induces a $K_{1,4}$.  Also $w_jv_1\in E(G)$, otherwise $N^I(w_3)\cup\{w_3,w_j\}$, induces a $K_{1,4}$.   From Corollary \ref{cor1}, it follows that all the vertices in $K\setminus C$ are adjacent to $v_1$.  Suppose there exists $w\in K\setminus C$ such that $wv_1\notin E(G)$, then for any $u\in C\cap K$, $N^I(u)\cup\{u,w\}$ induces a $K_{1,4}$, a contradiction.  Finally, $K\cup \{v_1\}$ is a larger clique, contradicting to the maximality of $K$.  Therefore, no such $C_{n\ge10}$ exists.  This completes a proof of the lemma.  $\hfill\qed$
\end{proof}
%
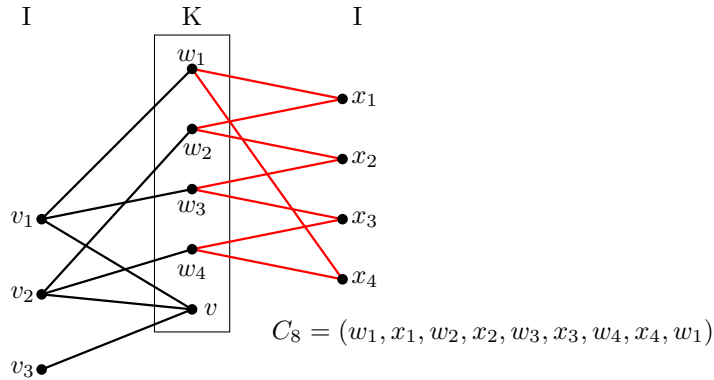
\begin{figure}[h!]
\begin{center}
\begin{tikzpicture}[scale=1] \label{fig:claim_onecycle}
\draw (-0.5,-3.5) rectangle (0.5,0.45); 

\N{0}{0.7}{$K$}  \N{2.2}{0.7}{$I$}  \N{-2.2}{0.7}{$I$}  
\N{4}{-3.5}{C_8=(w_1,x_1,w_2,x_2,w_3,x_3,w_4,x_4,w_1)}
\LR002-0.4  
\LR{0}{-0.8}{2}{-0.4}  
\LR{0}{-0.8}{2}{-1.2}  
\LR{0}{-1.6}{2}{-1.2}  
\LR{0}{-1.6}{2}{-2.0}  
\LR{0}{-2.4}{2}{-2.0}  
\LR{0}{-2.4}{2}{-2.8}  
\LR{0}{0}{2}{-2.8}  
\CB00     \N{-0}{0.15}{w_1}
\CB0-0.8  \N{0.08}{-01.1}{w_2}
\CB0-1.6  \N{-0}{-1.85}{w_3}
\CB0-2.4  \N{-0}{-2.7}{w_4}
\CB0-3.2  \N{0.25}{-3.2}{v}
\CB2-0.4  \N{2.3}{-0.4}{x_1}
\CB2-1.2  \N{2.3}{-1.2}{x_2}
\CB2-2.0  \N{2.3}{-2.0}{x_3}
\CB2-2.8  \N{2.3}{-2.8}{x_4}
\LB{-2}{-2}{0}{0}  
\LB{-2}{-2}{0}{-1.6}  
\LB{-2}{-3}{0}{-0.8}  
\LB{-2}{-3}{0}{-2.4}  
	\fill[color=black] (-2,-2.0) circle(2pt);
	\draw[color=black] (-2.25,-2.0) node {$v_1$};
	\fill[color=black] (-2,-3) circle(2pt);		
	\draw[color=black] (-2.25,-3) node {$v_2$};	
	\fill[color=black] (-2,-4) circle(2pt);		
	\draw[color=black] (-2.25,-4) node {$v_3$};
\LB{-2}{-2}{0}{-3.2}  
\LB{-2}{-3}{0}{-3.2}  
\LB{-2}{-4}{0}{-3.2}  
	\end{tikzpicture}
\end{center}  \vspace{-10pt}
\caption{An example of a \K14-free split graph with $\Delta^I=3$ and a short cycle $C_{8}$ }
\end{figure}
\textbf{Definition:} Let $G$ be a connected $K_{1,4}$-free graph with $\Delta^I_G=2$ satisfying property A and $H$ be the restricted bipartite subgraph of $G$.  By the constructive proof of Theorem \ref{delta2hamilpath}, there exists a collection $\mathbb{S}_f$ of vertex disjoint paths containing all the vertices of $G$.  
%
Such a collection is termed as a \emph{path collection} of $H$.\\\\
Let $G$ be a $K_{1,4}$-free split graph with $\Delta^I=3$, satisfying Property A.  For a vertex $v\in V_3$ let $G'=G-N^I(v)$.  Let $H$, $H'$ be the restricted bipartite subgraphs of $G$, $G'$, respectively and $\mathbb{S}_f$ be a path collection of $H'$.   Clearly, $H'$ is a subgraph of $H$.  Let $\mathbb{P}_k$, $k\ge1$ be the set containing all the maximal paths of length $k$ in $\mathbb{S}_f$.  Thus, $\mathbb{S}_f=\mathbb{P}_1\cup \mathbb{P}_2,\ldots,\cup\mathbb{P}_{k}$, where $\mathbb{P}_j$ is the set of maximal paths of size $j$ where for every $Q\in\mathbb{P}_j$, there does not exist $Q'\in\mathbb{S}_f$ such that $E(Q)\subset E(Q')$.  
Note that $\mathbb{S}_f$ has $I$-$K$ paths (even length paths) and $K$-$K$ paths (odd length paths).  A $K$-$K$ path $P_a\in\mathbb{S}_f$ is defined on the vertex set $V(P_a)=\{w_1,\ldots,w_j,x_1,\ldots,x_{j-1}\}$,  $E(P_a)=\{w_ix_i : 1\leq i\leq j-1\}\cup \{x_{k-1}w_k : 2\leq k\leq j\}$ such that $\{w_1,\ldots,w_j\}\subseteq K$, $\{x_1,\ldots,x_{j-1}\}\subseteq I$.   We denote such a path as $P_a=P(w_1,\ldots,w_j;x_1,\ldots,x_{j-1})$.  Similarly, $P_b=P(w_1,\ldots,w_j;x_1,\ldots,x_{j})$ represents an $I$-$K$ path with $V(P_b)=\{w_1,\ldots,w_j,x_1,\ldots,x_{j}\}$, $\{w_1,\ldots,w_j\}\subseteq K$, $\{x_1,\ldots,x_{j}\}\subseteq I$ and $E(P_b)=\{w_ix_i : 1\leq i\leq j\}\cup \{x_{k-1}w_k : 2\leq k\leq j\}$.

%
\begin{lemma} \label{HPoneSC}
Let $G$ be a connected $K_{1,4}$-free split graph with $\Delta^I=3$, satisfying Property A.  If $G$ has a short cycle $C$, then there exists a Hamiltonian path in $G$.
\end{lemma}
\begin{proof} Let $v\in V_3$, $N^I(v)=\{v_1,v_2,v_3\}$.  Recall that $H'$ is the restricted bipartite subgraph of $G-N^I(v)$ and $\mathbb{S}_f$ is a path collection of $H'$.  From Lemma \ref{onecycle}, there exists exactly one short cycle $C$ in $G$.  Let length of $C$ be $k$ and $P\in \mathbb{P}_k$ is such that $V(P)=V(C)$.   We see the following cases depending on the presence of $v$ in $C$. \\
\emph{Case 1:} $v\nin C$. 
Consider a vertex $w\in C\cap K$.  From Corollary \ref{cor1}, $N^I(v)\cap N^I(w)\ne\emptyset$.  Thus there exists $v_1\in N^I(v)$ such that $v_1w\in E(G)$.  Now we claim that there exists $x\in C\cap K$ such that $v_1x\nin E(G)$.  Suppose for a contradiction assume for every $x\in C\cap K$, $v_1x\in E(G)$.  It follows from Corollary \ref{cor1} that all the vertices in $K\setminus C$ are adjacent to $v_1$ and $K\cup \{v_1\}$ is a larger clique, contradicting the maximality of $K$.  Thus $v_1x\nin E(G)$.  Further, from Corollary \ref{cor1}, either $v_2x\in E(G)$ or $v_3x\in E(G)$.  Without loss of generality, let $v_2x\in E(G)$.
%
%
Using the above vertices $w,x$, we claim that in the collection $\mathbb{S}_f$ of $H'$, for $n\ne k$, $\mathbb{P}_{n\ge 4}=\emptyset$.  Note that $|\mathbb{P}_k|=1$.  Suppose there exists a path $P_a\ne P$, $P_a\in \mathbb{P}_{n\ge 4}$, then there exists a vertex $u\in P_a\cap K$ such that in $P_a$, $d^I(u)=2$.  From Corollary \ref{cor1}, $N^I(u)\cap N^I(v)\ne \emptyset$.  If $uv_1\in E(G)$, then $N^I(x)\cup \{x,u\}$ induces a $K_{1,4}$, otherwise $N^I(w)\{w,u\}$ induces a $K_{1,4}$.  If $d(v_3)=1$, then since $G$ satisfies Property A, there does not exist $z\in I\setminus (C\cup N^I(v))$ such that $d_G(z)=1$.  It follows that $\mathbb{P}_2=\emptyset$.  This is true because $G$ has already one short cycle and apart from that it can have at most one $I$-$K$ path as per Property A.  Since $d(v_3)=1$, no such $z$ exists.  Let $\mathbb{P}_1=\{\{v\},\{w_1\},\ldots,\{w_k\}\}$ and $\mathbb{P}_3=\{(w_{k+1},x_1,w_{k+2}),(w_{k+3},x_2,w_{k+4}),\ldots\}$.  We construct a path $Q=(w_1,w_2,\ldots,w_k,w_{k+1},x_1,w_{k+2},w_{k+3},x_2,w_{k+4},\ldots)$.  As per the premise, $|K|\ge|I|-1\ge8$.  Thus $|I|\ge9$, and $\mathbb{P}_3\ne\emptyset$.  It follows that $Q$ is non-empty, further $|Q|\ge5$. From Corollary \ref{cor1}, all the vertices in $Q\cap K$ are adjacent to both $v_1$ and $v_2$.  Suppose there exists $s\in Q$ such that $v_1s\nin E(G)$ or $v_2s\nin E(G)$, then either \NC{w}{s} or \NC{x}{s} induces a \K14.  Thus for every $s\in Q\cap K$, $v_1s\in E(G)$ and $v_2s\in E(G)$.  Observe that $(v_3,v,v_2,\overrightarrow{Q},v_1,w\overrightarrow{C})$ is a Hamiltonian path in $G$.  If $d(v_3)>1$, then we see the following.  If $\mathbb{P}_2\ne\emptyset$, since $G$ satisfies Property A, $|\mathbb{P}_2|=1$ i.e., $\mathbb{P}_2=\{P_b\}$, $P_b=(y,z)$, $y\in K$.  Further, there exists $w'\in K$ such that $v_3w'\in E(G)$.  Now we claim that $w'\nin C$.  Suppose not, then there exists $w''\in C\cap K$ such that \NC{w''}{w'} induces a \K14.  Similar to the argument with respect to the vertex $s$, for the vertex $w'$, we argue that $v_1w',v_2w'\in E(G)$, and $\{w'\}\in \mathbb{P}_1$.  Let $w'=w_i,i\le k$.  Then we construct a path $Q'=(w_1,w_2,\ldots,w_{i-1},w_{i+1},\ldots,w_k,w_{k+1},x_1,w_{k+2},w_{k+3},x_2,w_{k+4},\ldots)$.   Now we obtain $P_c=(z,y,w',v_3,v,v_2,\overrightarrow{Q'},v_1,w\overrightarrow{C})$ as a Hamiltonian path in $G$.  If $\mathbb{P}_2=\emptyset$, then a $(w'\overrightarrow{P_c})$ is a Hamiltonian path in $G$.\\
\emph{Case 2:} $v\in C$.
Let $v_2,v_3\in N^I(v)\cap C$.  Then note that there exists $v_1\in N^I(v)$ such that $v_1\nin C$.  Clearly, for all $w\in K\setminus C$, $wv_1\in E(G)$.  Since $K$ is a maximal clique, it follows that there exists $x\in C\cap K$ such that $v_1x\nin E(G)$.  We now claim that $\mathbb{P}_{n\ge 4}=\emptyset$, $n\ne k$.  Suppose there exists a path $P_a\ne P$, $P_a\in \mathbb{P}_{n\ge 4}$, then there exists a vertex $u\in P_a\cap K$ such that in $P_a$, $d^I(u)=2$.  We already observed that $v_1u\in E(G)$.  Further, \NC{u}{x} induces a \K14, a contradiction.  Thus such a path $P_a$ does not exist.
%
If $\mathbb{P}_2\ne\emptyset$,  then let $(y,z)\in\mathbb{P}_2$, $y\in K$.  Then $(\overrightarrow{C}v,v_1,\overrightarrow{Q},y,z)$ is a Hamiltonian path in $G$.  If  $\mathbb{P}_2=\emptyset$, then $(\overrightarrow{C}v,v_1,\overrightarrow{Q})$ is a Hamiltonian path in $G$.  This completes the case analysis and a proof of the lemma. $\hfill\qed$
\end{proof} 
%
%
We work on a \K14-free split graph $G$ with $\Delta^I=3$, satisfying Property A, with $G'$, $H'$ and $\mathbb{S}_f$ as defined previously.  If $G$ has a short cycle, then by Lemma \ref{HPoneSC}, $G$ has a Hamiltonian path.  If $G$ has no short cycles, then note that there exists at most 2 $I$-$K$ paths in $\mathbb{S}_f$.  For the following claims, we shall consider such a $G$ with no short cycle.  The structural study of paths in $\mathbb{S}_f$ is deep, which is the highlight of this chapter.  Now we shall present some claims to show the structural observations of paths in $\mathbb{S}_f$.  
\begin{cl}\label{xunivge10}
 If there exists a path $P_a\in \mathbb{P}_k,k\ge10$ such that $P_a=P(w_1,\ldots,w_{j'};$ $x_1,\ldots,x_{j}),$ $j+1\ge j'\ge j\ge5$, then there exists $v_1\in N^I(v)$ such that $v_1w_i\in E(G), {2\le i\le j}$.   
 \end{cl}
\begin{proof}
First we show that for any two vertices $w_i,w_l\in K$, $2\le i,l\le j$, $|i-l|>1$; $v_1w_i, v_1w_l\in E(G)$.  By Corollary \ref{cor1}, clearly there exists $v_1\in N(v)$ such that $v_1w_i\in E(G)$.  If $v_1w_l\nin E(G)$, then by Corollary \ref{cor1}, $v_2w_l\in E(G)$ or $v_3w_l\in E(G)$ is true.  It follows that \NC{w_i}{w_l} induces a \K14.  Since $|\{w_2,\ldots,w_j\}|\ge4$, for every $2\le i\le j$, $v_1w_i\in E(G)$.   $\hfill\qed$ 
\end{proof} 
\begin{cl} \label{noP12}
$\mathbb{P}_{i\ge12}=\emptyset$.
\end{cl}
\begin{proof} 
Assume for a contradiction there exists a path $P_{a}\in \mathbb{P}_i,{i\ge12}$.  Let $P_a=P(w_1,\ldots,w_{j'};x_1,\ldots,x_{j})$, ${j+1\ge j'\ge j\ge6}$.  From Claim \ref{xunivge10} there exists $v_1\in N^I_G(v)$ such that $v_1w_k\in E(G),{2\le k\le j}$.  We now claim that $v_1w_1\in E(G)$.  Suppose not, then from Corollary \ref{cor1}, $v_2w_1\in E(G)$ or $v_3w_1\in E(G)$.  Further, $w_1x_3\in E(G)$, otherwise \NC{w_3}{w_1} or \NC{w_4}{w_1} has an induced \K14.  Similarly, $w_1x_5\in E(G)$.  It follows that \NV{w_1} has an induced \K14, a contradiction.  Thus $v_1w_1\in E(G)$.  If $P_a$ is an odd path, then similar arguments with respect to $w_1$ holds good for the vertex $w_{j'}$, and hence $v_1w_{j'}\in E(G)$.
Since the clique is maximum in $G$, there exists $s\in K$ such that $v_1s\notin E(G)$.  Further, there exists at least three vertices in $x_1,\ldots,x_{j}$ adjacent to $s$, otherwise, for some $2\le r\le j$, $N^I_G(w_r)\cup \{w_r,s\}$ induces a $K_{1,4}$.  Finally, from Corollary \ref{cor1}, either $v_2s\in E(G)$ or $v_3s\in E(G)$.  It follows that $N^I_G(s)\cup \{s\}$ induces a $K_{1,4}$, a contradiction.  Thus such a path $P_a$ does not exist.  This completes a proof of the claim. $\hfill \qed$ 
\end{proof}
%
\begin{cl}\label{xunivtwopathseven}
Let $P_a=P(w_1,\ldots,w_{i'};x_1,\ldots,x_{i}), {i+1\ge i'\ge i\ge2}$, and $P_b=P(s_1,$ $\ldots,s_{j'};t_1,\ldots,t_{j})$, $j+1\ge j'\ge j \ge2$ be arbitrary paths in $\mathbb{S}_f$.  Then there exists $v_1\in N^I(v)$ such that $\forall~{2\le l\le i}$, $v_1w_l\in E(G)$, and $\forall~{2\le m\le j}$,  $v_1s_m\in E(G)$. 
\end{cl}
\begin{proof}
From Corollary \ref{cor1}, there exists $v_1\in N^I(v)$ such that $v_1w_2\in E(G)$.  If $v_1s_m\nin E(G)$, $2\le m\le j$ then by Corollary \ref{cor1}, $v_2s_m\in E(G)$ or $v_3s_m\in E(G)$. It follows that \NC{w_2}{s_m} induces a \K14.  Thus $v_1s_m\in E(G)$.  If path $P_a$ has  size more than 5, then for every $3\le l\le i$, $v_1w_l\in E(G)$.  Suppose not, then by Corollary \ref{cor1}, $v_2w_l\in E(G)$ or $v_3w_l\in E(G)$.  It follows that \NC{s_2}{w_l} induces a \K14.  Therefore, we conclude that for all possible $l,m$; $v_1w_l,v_1s_m\in E(G)$,  and the claim follows. $\hfill\qed$
\end{proof}
\begin{corollary} \label{xunivthreepathseven}
(of Claim \ref{xunivtwopathseven}) If $P_a=P(w_1,\ldots,w_{l'};x_1,\ldots,x_{l}), {l+1\ge l'\ge l\ge2}$, $P_b=P(s_1,\ldots,s_{m'};t_1,\ldots,$ $t_{m})$, $m+1\ge m'\ge m \ge2$ and $P_c=P(q_1,\ldots,q_{n'};$ $r_1,\ldots,r_{n})$, $n+1\ge n'\ge n \ge2$ are arbitrary paths in $\mathbb{S}_f$, then there exists $v_1\in N^I(v)$ such that $\forall~{2\le i\le l}$, $v_1w_i\in E(G)$, $\forall~{2\le j\le m}$,  $v_1s_j\in E(G)$ and $\forall~{2\le k\le n}$,  $v_1q_k\in E(G)$.
\end{corollary}
\begin{cl}\label{P8_no_P4} 
 If there exists $P_a\in \mathbb{P}_{k\ge8}$, then there does not exist $P_b$ such that $|P_b|\ge4$.
\end{cl}
\begin{proof}
Assume for a contradiction that there exists such a path $P_b\in \mathbb{P}_j, j\ge4$.  Let $P_a=(w_1,\ldots,w_{l'};x_1,\ldots,x_l)$, $l+1\ge l'\ge l\ge4$
and $P_b=(s_1,\ldots,s_{r'};t_1,\ldots,t_{r})$,  $r+1\ge r'\ge r\ge2$. 
From Claim \ref{xunivtwopathseven}, $v_1w_i\in E(G),2\le i\le{l}$ and $v_1s_i\in E(G),2\le i\le{r}$. 
Now we claim $v_1w_1\in E(G)$.  Otherwise, by Corollary \ref{cor1}, $v_2w_1$ or $v_3w_1$ is in $E(G)$.  Observe that $w_1x_3\in E(G)$, otherwise $N^I(w_3)\cup \{w_3,w_1\}$ or $N^I(w_4)\cup \{w_4,w_1\}$ induces a $K_{1,4}$.  Further, either $w_1t_1\in E(G)$ or $w_1t_2\in E(G)$, otherwise $N^I(s_2)\cup \{s_2,w_1\}$ induces a $K_{1,4}$.  Now $N^I(w_1)\cup \{w_1\}$ induces a $K_{1,4}$, a contradiction and thus $v_1w_1\in E(G)$.  If $P_a$ is an odd path, then using similar argument, we establish $v_1w_{l'}\in E(G)$.  
Now we claim that $v_1s_1\in E(G)$.  Otherwise, by Corollary \ref{cor1}, $s_1v_2\in E(G)$ or $s_1v_3\in E(G)$.  Further, $s_1x_1\in E(G)$ or $s_1x_2\in E(G)$, otherwise $N^I(w_2)\cup \{w_2,s_1\}$ induces a $K_{1,4}$.  Similarly, $s_1x_3\in E(G)$ or $s_1x_4\in E(G)$.  Now \NV{s_1} induces a $K_{1,4}$, a contradiction.  Therefore, $v_1s_1\in E(G)$.  If $P_b$ is an odd path, then using similar argument, we establish $v_1s_{r'}\in E(G)$.  Since the clique is maximal, there exists a vertex $w'\in K$ such that $v_1w'\notin E(G)$.  By Corollary 1, $w'v_2\in E(G)$ or $w'v_3\in E(G)$; without loss of generality, let $w'v_2\in E(G)$.  Also due to the similar reasoning for $s_1$, $w'$ is adjacent to one among $x_1,x_2$, and $w'$ is adjacent to one among $x_3,x_4$.  Further, either $t_1w'\in E(G)$ or $t_2w'\in E(G)$, otherwise $N^I(s_2)\cup \{s_2,w'\}$ induces a $K_{1,4}$.  Finally, $N^I(w')\cup \{w'\}$ induces a $K_{1,4}$, a contradiction.  Therefore, $P_b$ does not exist.  This completes a proof of Claim \ref{P8_no_P4}.  $\hfill\qed$
\end{proof}
%
\begin{cl}\label{HP_P11} 
 If there exists $P_a\in\mathbb{P}_{11}$, then $G$ has a Hamiltonian path.
\end{cl}
\begin{proof}
Let $P_a=(w_1,\ldots,w_6;x_1,\ldots,x_5)$.  From Claim \ref{xunivge10}, there exists a vertex say $v_1\in N^I_G(v)$, such that $v_1w_i\in E(G)$, $2\le i\le 5$.  From the proof of the previous claim, $v_1w_1,v_1w_6\in E(G)$.  Since the clique is maximal, there exists $w'\in K$, such that $w'v_1\notin E(G)$.  By Corollary \ref{cor1}, $w'v_2\in E(G)$ or $w'v_3\in E(G)$.  Without loss of generality, let $w'v_2\in E(G)$.   We claim $w'x_2\in E(G)$ and $w'x_4\in E(G)$, otherwise for some $2\le i\le 5$, $N^I(w_i)\cup \{w_i,w'\}$ induces a $K_{1,4}$.  One among $v_2,x_2,x_4$ is adjacent to $w_1$, otherwise $N^I(w')\cup \{w_1,w'\}$ induces a $K_{1,4}$.  Similar argument holds good with respect to the vertex $w_6$.  From Claim \ref{P8_no_P4}, $\mathbb{P}_j=\emptyset$, $j\ge4$.  If $\mathbb{P}_2\ne\emptyset$, since $G$ satisfies Property A, note that at most two vertices of $I$ have degree 1.  If $d(v_3)=1$, then $|\mathbb{P}_2|\le1$.  Let $P_b\in\mathbb{P}_2$.  Let $\overrightarrow{Q}$ represents an ordering of paths in $\mathbb{P}_1\cup\mathbb{P}_3$, excluding the paths $\{v\}$, $\{w'\}$.  $P=(v_3,v,v_2,w',x_2\overrightarrow{P_a}w_6,v_1, w_1\overrightarrow{P_a}w_2,\overrightarrow{Q},\overrightarrow{P_b})$ or $(\overrightarrow{P}w_2,\overrightarrow{Q})$ is a Hamiltonian path in $G$.  If $d(v_3)>1$, then $|\mathbb{P}_2|\le2$ and let $P_c\in \mathbb{P}_2$, $P_c\neq P_b$.  Note that there exists $w''\in K$ such that $v_3w''\in E(G)$ and $w''$ is adjacent to at least two vertices in $\{x_1,\ldots,x_5\}$, otherwise for some $2\le i\le 5$, \NC{w_i}{w''} induces a \K14. Thus $\{w''\}\in \mathbb{P}_1$.  Let $\overrightarrow{Q'}$ represents an ordering of paths in $\mathbb{P}_1\cup\mathbb{P}_3$, excluding the paths $\{v\}$, $\{w'\}$, $\{w''\}$.  $P'=(\overleftarrow{P_b},w'',v_3,v,v_2,w',x_2\overrightarrow{P_a}w_6,v_1, w_1\overrightarrow{P_a}w_2,\overrightarrow{Q'},\overrightarrow{P_c})$ is a Hamiltonian path in $G$.  If $|\mathbb{P}_2|<2$, then a observe that $(\overrightarrow{P'}w_2,\overrightarrow{Q'})$ or $(w''\overrightarrow{P'})$ or $(w''\overrightarrow{P'}w_2,\overrightarrow{Q'})$ is a spanning path of $G$. This completes a proof of the claim. $\hfill\qed$
\end{proof}
\begin{cl}\label{HC_P10} 
 If there exists $P_a\in\mathbb{P}_{10}$, then $G$ has a Hamiltonian path.
\end{cl}
\begin{proof}
Let $P_a=(w_1,\ldots,w_5;x_1,\ldots,x_5)$.  Similar to the previous proof, there exists a vertex $v_1\in N^I_G(v)$, such that $v_1w_i\in E(G)$, $1\le i\le 5$.  Similar to Claim \ref{HP_P11} we see that there exists $w'\in K$, such that $w'v_1\notin E(G)$ and $w'v_2,w'x_2,w'x_4\in E(G)$.  One among $v_2,x_2,x_4$ is adjacent to $w_1$, otherwise $N^I(w')\cup \{w_1,w'\}$ induces a $K_{1,4}$.  Recall from Claim \ref{P8_no_P4}, $\mathbb{P}_{j\ge4}=\emptyset$.  Since $G$ satisfies Property A, note that at most two vertices of $I$ have degree 1, and observe that $d(x_5)=1$.  Therefore, if $d(v_3)=1$, then $\mathbb{P}_2=\emptyset$.  
Let $\overrightarrow{Q}$ represents an ordering of paths in $\mathbb{P}_1\cup\mathbb{P}_3$, excluding $\{v\}$, $\{w'\}$.  Note that the first vertex $u$ in $\overrightarrow{Q}$ is adjacent to $v_1$ or $v_2$.   Then $(x_5,w_5,v_1,u\overrightarrow{Q},w_1\overrightarrow{P_a}x_4,w',v_2,v,v_3)$ or $(x_5,w_5,v_1, w_1\overrightarrow{P_a}x_4,w',v_2,u\overrightarrow{Q},v,v_3)$ is a Hamiltonian path in $G$.  If $d(v_3)>1$, then $|\mathbb{P}_2|\le1$ and let $P_b\in \mathbb{P}_2$.  Note that there exists $w''\in K$ such that $v_3w''\in E(G)$ and $w''$ is adjacent to at least two vertices in $\{x_1,\ldots,x_5\}$, otherwise for some $2\le i\le 5$, \NC{w_i}{w''} induces a \K14. Thus $\{w''\}\in \mathbb{P}_1$.  Let $\overrightarrow{Q'}$ represents an ordering of paths in $\mathbb{P}_1\cup\mathbb{P}_3$, excluding the paths $\{v\}$, $\{w'\}$, $\{w''\}$. 
Now depending on the presence of $P_b$, $P'=(x_5,w_5,v_1, w_1\overrightarrow{P_a}x_4,w',v_2,v,v_3,w'',\overrightarrow{Q'},\overrightarrow{P_b})$ or $(x_5\overrightarrow{P'}w'',\overrightarrow{Q'})$ is a Hamiltonian path in $G$.  This completes a proof of the claim. $\hfill\qed$
\end{proof}
From here onwards, for producing Hamiltonian paths in the proof of claims, we obtain a special path, termed as \emph{desired path}, which is a path containing the vertices $\{v_1,v_2,v_3\}$ and all (at most two) the $I$-$K$ paths in $\mathbb{S}_f$ along with some $K$-$K$ paths.  Let $\overrightarrow{Q}$ be an ordering of the paths in $\mathbb{S}_f$ which are not included in the desired path.  Depending on the adjacency of the first vertex of $\overrightarrow{Q}$ to $N^I(v)$, $(\overrightarrow{P_a},v_i,\overrightarrow{Q},\overrightarrow{P_b})$ is a Hamiltonian path in $G$, where $(\overrightarrow{P_a},v_i,\overrightarrow{P_b})$, $i\in \{1,2,3\}$ is the desired path, with $|\overrightarrow{P_a}|\ge0$, $|\overrightarrow{P_b}|\ge0$.
\begin{cl}\label{HP_P9}
If there exists $P_a\in\mathbb{P}_9$, then $G$ has a Hamiltonian path.
\end{cl}
\begin{proof}
Let $P_a=(w_1,\ldots,w_5;x_1,\ldots,x_4)$.  There exists a vertex in $v_1\in N^I(v)$ such that $v_1w_2\in E(G)$.  Note that $v_1w_4\in E(G)$, otherwise \NC{w_2}{w_4} induces a $K_{1,4}$.  Recall from Claim \ref{P8_no_P4}, $\mathbb{P}_{j\ge4}=\emptyset$.  Now we claim that for an arbitrary path $P\in\mathbb{P}_2\cup \mathbb{P}_3$ and $s\in K$ be an end vertex of $P$, then $v_1s\in E(G)$.  Suppose not, then $v_2s\in E(G)$ or $v_3s\in E(G)$.  Further, $s$ is adjacent to either $x_1$ or $x_2$, otherwise \NC{w_2}{s} induces a \K14.  Similarly, $s$ is adjacent to one of $x_3,x_4$, otherwise \NC{w_4}{s} induces a \K14.  Therefore, \NV{s} induces a \K14, a contradiction.  Thus $v_1s\in E(G)$. Note that the above argument is true for any end vertex $s\in K$ of every such paths in $\mathbb{P}_2\cup\mathbb{P}_3$.  Since the clique is maximal, there exists a non-adjacency for $v_1$ in $K$, and based on the non-adjacency, we see the following cases as shown in Tables \ref{tab:HP_P9_1}, \ref{tab:HP_P9_2} and \ref{tab:HP_P9_3}.  This completes a proof of Claim \ref{HP_P9}. $\hfill\qed$
\end{proof}
%
\begin{table}[h!]
\begin{tabular}{ | L{0.18\textwidth} | L{0.8\textwidth} |}
\hline
Case & Arguments\\
\hline
Case 1: $v_1w_1\nin E(G)$ or $v_1w_5\nin E(G)$ &  Without loss of generality, we shall assume $v_1w_1\notin E(G)$.  Note that one of $v_2,v_3$ is adjacent to $w_1$, without loss of generality, let $v_2w_1\in E(G)$.  Further, note that $w_1$ is adjacent to one of $x_3,x_4$, otherwise, \NC{w_4}{w_1} induces a \K14.  \\
\hline
Case 1.1: $d(v_3)=1$ & From Property A, at most 2 vertices of $I$ have degree 1.  Thus, $|\mathbb{P}_2|\le1$.   Let $P_b\in\mathbb{P}_2$.  Further, from Property A, $|I|\ge9$.  Therefore, there exists $P_c\in\mathbb{P}_3$.  Let $\overrightarrow{Q}$ be an ordering of the paths in $\mathbb{P}_1\cup \mathbb{P}_3$ excluding paths $\{v\}$ and $P_c$.  Then $(v_3,v,v_2,w_1\overrightarrow{P_a}w_5,\overrightarrow{Q},\overrightarrow{P_c},v_1,\overrightarrow{P_b})$ is a Hamiltonian path in $G$.  \\
\hline
 Case 1.2: $d(v_3)>1$ & Recall $d^I(w_1)=3$. We claim $v_3$ is not adjacent to $s\in K$ such that $s$ is an end vertex of any path in $\mathbb{P}_2\cup\mathbb{P}_3$.  Recall that $v_1s\in E(G)$.  Suppose $v_3s\in E(G)$, then \NC{w_1}{s} induces a \K14.  Thus $v_3s\nin E(G)$.  It follows that $v_3$ is adjacent to $w_3$ or $w_5$ or a vertex in $\mathbb{P}_1$.  Observe that $|\mathbb{P}_2|\le2$.   To complete the proof, we shall assume that $|\mathbb{P}_2|=2$.   Proof of other two cases are similar.  Let $P_b,P_d\in\mathbb{P}_2$, $y\in P_b\cap K$.    \\
\hline
Case 1.2.1: $v_3w_5\in E(G)$ &   Note $N^I(w_1)\cap N^I(y)\ne\emptyset$.  Thus $y$ is adjacent to one among $\{v_2,x_1,x_3,x_4\}$.  \newline 
If $v_2y\in E(G)$, then $(\overrightarrow{P_b}y,v_2,w_1\overrightarrow{P_a}w_5,v_3,v,v_1,\overrightarrow{P_d})$ is a desired path in $G$.  Let $\overrightarrow{Q}$ be an ordering of the paths in $\mathbb{P}_1\cup \mathbb{P}_3$ excluding $\{v\}$.  Note that $\overrightarrow{Q}$ need to be inserted at appropriate place to get a Hamiltonian path in $G$.  For instance, if the end vertex of $Q$ is adjacent to $v_2\in N^I(v)$, then $(\overrightarrow{P_b}y,v_2,\overrightarrow{Q},w_1\overrightarrow{P_a}w_5,v_3,v,v_1,\overrightarrow{P_d})$ is a Hamiltonian path in $G$.  On a similar way, we obtain Hamiltonian path from all such desired paths. \newline
If $x_1y\in E(G)$, then $(\overrightarrow{P_b}y,x_1,w_1,v_2,v,v_3,w_5\overleftarrow{P_a}w_2,v_1,\overrightarrow{P_d})$ is a desired path in $G$.  \newline
If $x_3y\in E(G)$, then $(\overrightarrow{P_b}y,x_3\overleftarrow{P_a}w_1,v_2,v,v_3,w_5\overleftarrow{P_a}w_4,v_1,\overrightarrow{P_d})$ is a desired path in $G$.  \newline
If $x_4y\in E(G)$, then $(\overrightarrow{P_b}y,x_4,w_5,v_3,v,v_2,w_1\overrightarrow{P_a}w_4,v_1,\overrightarrow{P_d})$ is a desired path in $G$.  
 \\
\hline
Case 1.2.2: $v_3w_3\in E(G)$ & Since $N^I(w_1)\cap N^I(w_4)\neq\emptyset$, $w_1$ is adjacent to either $x_3$ or $x_4$.  If $w_1x_4\in E(G)$, then \NC{w_1}{w_3} induces a \K14.   Thus, $w_1x_3\in E(G)$.  Note $N^I(w_1)\cap N^I(y)\ne\emptyset$ and $N^I(w_3)\cap N^I(y)\ne\emptyset$.  Therefore, $yx_3\in E(G)$.  Further, $w_5$ is adjacent to both $x_3$ and $v_1$, otherwise for some $s\in\{w_1,\ldots,w_4,v,y\}$, \NC{s}{w_5} induces a \K14.  Now $d(x_1)=d(x_2)=d(v_2)=d(v_3)=2$.  \\
\hline
\end{tabular}  \caption{Case analysis for the proof of Claim \ref{HP_P9} }
\label{tab:HP_P9_1}
\end{table}
\begin{table}[h!]
\begin{tabular}{ | L{0.18\textwidth} | L{0.795\textwidth} |}
\hline
Case 1.2.2: $v_3w_3\in E(G)$ & Since $G$ has no short cycles, there exists $w'\in K$ such that $w'$ is an end vertex of a path $P_e\in \mathbb{P}_1\cup\mathbb{P}_3$, and $w'$ is adjacent to at least one vertex in $\{x_1,x_2,v_2,v_3\}$.  Depending on the adjacency of $w'$, we obtain the following desired paths.  
$(\overrightarrow{P_b},\overrightarrow{P_e}w',x_1,w_2,x_2,w_3,v_3,v,v_2,w_1,x_3\overrightarrow{P_a}w_5,v_1,\overrightarrow{P_d})$ 
$(\overrightarrow{P_b},\overrightarrow{P_e}w',x_2,w_2,x_1,w_1,v_2,v,v_3,w_3,x_3\overrightarrow{P_a}w_5,v_1,\overrightarrow{P_d})$  \newline 
$(\overrightarrow{P_b},\overrightarrow{P_e}w',v_2,v,v_3,w_3,x_2,w_2,x_1,w_1,x_3\overrightarrow{P_a}w_5,v_1,\overrightarrow{P_d})$ \newline 
$(\overrightarrow{P_b},\overrightarrow{P_e}w',v_3,v,v_2,w_1\overrightarrow{P_a}w_5,v_1,\overrightarrow{P_d})$ 
 \\
\hline
Case 1.2.3: $v_3z\in E(G)$, $z\nin P_a$ &   
Note that $N^I(w_5)\cap N^I(v)\ne\emptyset$.  We see the following cases.  If $v_1w_5\in E(G)$, then $(\overrightarrow{P_b},z,v_3,v,v_2,w_1\overrightarrow{P_a}w_5,v_1,\overrightarrow{P_d})$ is a desired path in $G$.
If $v_3w_5\in E(G)$, then $(\overrightarrow{P_b},z,v_3,w_5\overleftarrow{P_a}w_1,v_2,v,v_1,\overrightarrow{P_d})$ is a desired path in $G$.  If $v_2w_5\in E(G)$, then we see the following.  Note that $w_5$ is adjacent to one of $N^I(w_2)=\{v_1,x_1,x_2\}$.  If $w_5v_1\in E(G)$, then we obtained a desired path in $G$.  If $w_5x_1\in E(G)$, then $(\overrightarrow{P_b},z,v_3,v,v_2,w_1,x_1,w_5\overleftarrow{P_a}w_2,v_1,\overrightarrow{P_d})$ is a desired path in $G$.  Consider $w_5x_2\in E(G)$.  Recall that $w_1$ is adjacent to one of $x_3,x_4$.   If $w_1x_4\in E(G)$, then $(\overrightarrow{P_b},z,v_3,v,v_2,w_5,x_4,w_1\overrightarrow{P_a}w_4,v_1,\overrightarrow{P_d})$ is a desired path in $G$.  If $w_1x_3\in E(G)$, then we claim that $v_2y\in E(G)$.  Suppose not, then \NC{w_1}{y} or \NC{w_5}{y} induces a \K14.  Now we obtain $(\overrightarrow{P_d},v_1,w_2,x_1,w_1,x_3,w_3,x_2,w_5,x_4,w_4,z,v_3,v,v_2,$ $\overrightarrow{P_b})$ as a desired path in $G$.
  \\ 
\hline
Case 2: $v_1w_3\nin E(G)$ &   In this case we shall assume $v_1w_1,v_1w_5\in E(G)$.   Note that $w_3$ is adjacent to $v_2$ or $v_3$.  Without loss of generality, we shall assume $w_3v_2\in E(G)$.  Recall that all the end vertices of paths in $\mathbb{P}_2\cup \mathbb{P}_3$, in $K$ are adjacent to $v_1$.  Similar to the previous case, if $d(v_3)=1$, then $|\mathbb{P}_2|\le1$, let $P_b\in\mathbb{P}_2$.  Further, there exists $P_c\in\mathbb{P}_3$, let $y\in P_c\cap K$.  Note that $N^I(y)\cap N^I(w_3)\neq\emptyset$.  Depending on the adjacency of $y$ with $\{x_3,x_2,v_2\}$, we obtain desired paths $P_i,P_j,P_k$ as follows. \newline
If $yx_3\in E(G)$, then \newline $P_i=(v_3,v,v_2,w_3,x_2\overleftarrow{P_a}w_1,v_1,w_5\overleftarrow{P_a}x_3,y\overrightarrow{P_c},\overrightarrow{P_b})$.  \newline
If $yx_2\in E(G)$, then \newline $P_j=(v_3,v,v_2,w_3,x_3\overrightarrow{P_a}w_5,v_1,w_1\overrightarrow{P_a}x_2,y\overrightarrow{P_c},\overrightarrow{P_b})$.  \newline
If $yv_2\in E(G)$, then $P_k=(v_3,v,v_2,y\overrightarrow{P_c},w_1\overrightarrow{P_a}w_5,v_1,\overrightarrow{P_b})$.  On the other hand if $d(v_3)>1$, then similar to the previous case, $|\mathbb{P}_2|\le2$.  Let $P_d\in \mathbb{P}_2$, $s\in P_b\cap K$, $s'\in P_d\cap K$.  
 \\ \hline
\end{tabular}  \caption{Case analysis for the proof of Claim \ref{HP_P9} }
\label{tab:HP_P9_2}
\end{table}

\begin{table}[h!]
\begin{tabular}{ | L{0.18\textwidth} | L{0.795\textwidth} |}
\hline
Case 2: $v_1w_3\nin E(G)$ &
Observe that all the vertices in $\{w_1,w_5,s,s'\}$ are adjacent to a vertex in $N^I(w_3)$.    It follows that for every vertex $y\in P_a\cup P_b \cup P_d$ such that $y\in K$, $d^I(y)=3$.  Since $d(v_3)>1$, there exists $z\in K$ such that $z\nin P_a\cup P_b\cup P_d$ and $zv_3\in E(G)$.  From the previous argument $z\in\mathbb{P}_1$.  Depending on the adjacency of $s$ with $\{x_3,x_2,v_2\}$, we obtain a desired path $P$ as follows.   
If $sx_3\in E(G)$, then $P=(\overrightarrow{P_d},z,v_3\overrightarrow{P_i}x_3,s\overrightarrow{P_b})$.
If $sx_2\in E(G)$, then $P=(\overrightarrow{P_d},z,v_3\overrightarrow{P_j}x_2,s\overrightarrow{P_b})$.
If $sv_2\in E(G)$, then $P=(\overrightarrow{P_d},v_1,w_1\overrightarrow{P_a}w_5,z,v_3,v,v_2,s\overrightarrow{P_b})$.
\\
\hline
Case 3: $v_1t\nin E(G)$, $t\nin P_a$ &   In this case we shall assume $v_1w_i\in E(G)$, $1\le i\le5$.   Note that $t$ is adjacent to $v_2$ or $v_3$.  Without loss of generality, we shall assume $tv_2\in E(G)$.  Similar to the previous case, if $d(v_3)=1$, then $|\mathbb{P}_2|\le1$, let $P_b\in\mathbb{P}_2$.  Recall that all the end vertices of paths in $\mathbb{P}_2\cup \mathbb{P}_3$ are adjacent to $v_1$.  We obtain $(v_3,v,v_2,t,\overrightarrow{P_a},v_1,\overrightarrow{P_b})$ as a desired path in $G$.  If $d(v_3)>1$, then we observed that $|\mathbb{P}_2|\le2$, let $P_d\in \mathbb{P}_2$. Note that $t$ is adjacent to at least two vertices in $\{x_1,\ldots,x_4\}$, otherwise for some $w_i$, $2\le i\le 4$, \NC{w_i}{t} induces a \K14.  Since $d(v_3)>1$, there exists $z\in K$ such that $zv_3\in E(G)$.   Let $s\in P_b\cap K$, we observe that $z\ne s$.  Suppose not, then \NC{t}{s} induces a \K14. Similarly, for $s'\in P_d\cap K$, $z\ne s'$.  If $z=w_1$, then $(\overrightarrow{P_b},v_1,w_5\overleftarrow{P_a}w_1,v_3,v,v_2,t,\overrightarrow{P_d})$ is a desired path.  $z=w_5$ is a symmetric case.  If $z\notin P_a$, then from previous arguments, $z\in\mathbb{P}_1$.  We obtain $(\overrightarrow{P_b},z,v_3,v,v_2,t,\overrightarrow{P_a},v_1,\overrightarrow{P_d})$ as a desired path in $G$. \\
\hline
\end{tabular} \caption{Case analysis for the proof of Claim \ref{HP_P9} }
\label{tab:HP_P9_3}
\end{table} 
\begin{cl}\label{HP_P8}
If there exists $P_a\in\mathbb{P}_8$, then $G$ has a Hamiltonian path.
\end{cl}
\begin{proof}
Let $P_a=(w_1,\ldots,w_4;x_1,\ldots,x_4)$.  There exists a vertex in $v_1\in N^I(v)$ such that $v_1w_2\in E(G)$.  Note that $v_1w_4\in E(G)$, otherwise \NC{w_2}{w_4} induces a $K_{1,4}$.  Recall from Claim \ref{P8_no_P4}, $\mathbb{P}_{j\ge4}=\emptyset$.  Now we claim that for an arbitrary path $P\in\mathbb{P}_2\cup \mathbb{P}_3$ and $s\in K$ be an end vertex of $P$, then $v_1s\in E(G)$.  Suppose not, then $v_2s\in E(G)$ or $v_3s\in E(G)$.  Further, $s$ is adjacent to either $x_1$ or $x_2$, otherwise \NC{w_2}{s} induces a \K14.  Similarly, $s$ is adjacent to one of $x_3,x_4$, otherwise \NC{w_4}{s} induces a \K14.  From the above arguments it follow that \NV{s} induces a \K14, a contradiction.  Therefore, $v_1s\in E(G)$. Note that the above argument is true for any end vertex $s\in K$ of every such paths in $\mathbb{P}_2\cup\mathbb{P}_3$.  Since the clique is maximal, there exists a non-adjacency for $v_1$ in $K$, and based on the non-adjacency ($w_1$ or $w_3$ or some element $t\in \mathbb{P}_1$), we see the following cases as shown in Tables \ref{tab:HP_P8_1}, \ref{tab:HP_P8_2} and \ref{tab:HP_P8_3}.    
Let $Q$ be an ordering of paths in $\mathbb{P}_1\cup \mathbb{P}_3$ which are not in a desired path.  Hamiltonian path in $G$ could be obtained from the desired path by augmenting $Q$ appropriately. This completes a proof of Claim \ref{HP_P8}. $\hfill\qed$
\end{proof} 
%
\begin{table}[h!]
\begin{tabular}{ | L{0.18\textwidth} | L{0.795\textwidth} |}
\hline
Case & Arguments\\
\hline
Case 1: $v_1w_1\nin E(G)$  & Note that one of $v_2,v_3$ is adjacent to $w_1$, without loss of generality, let $v_2w_1\in E(G)$.  Further, note that $w_1x_3\in E(G)$, otherwise, \NC{w_4}{w_1} induces a \K14.  \\
\hline
Case 1.1: $d(v_3)=1$ & From Property A, at most 2 vertices of $I$ have degree 1.  Note $d(v_3)=d(x_4)=1$.  It follows that, $\mathbb{P}_2=\emptyset$.  Further, from Property A, $|I|\ge9$.  Therefore, there exists $P_b,P_c\in\mathbb{P}_3$, let $y\in P_b\cap K$.  Let $\overrightarrow{Q}$ be an ordering of the paths in $\mathbb{P}_1\cup \mathbb{P}_3$ excluding paths $P_b,P_c$ and $\{v\}$.  Note that $N^I(w_1)\cap N^I(y)\ne\emptyset$.  Depending on the adjacency of $y$ with $\{v_2,x_1,x_3\}$, we obtain Hamiltonian path $P$ as follows.  
If $yv_2\in E(G)$, then $P=(v_3,v,v_2,y\overrightarrow{P_b},v_1,\overrightarrow{P_c},\overrightarrow{Q},w_1\overrightarrow{P_a}x_4)$.   
\\ \hline
\end{tabular}
\caption{Case analysis for the proof of Claim \ref{HP_P8} } 
\label{tab:HP_P8_1} 
\end{table}
\begin{table}[h!]
\begin{tabular}{ | L{0.18\textwidth} | L{0.795\textwidth} |} 
\hline
Case 1.1: $d(v_3)=1$ &
If $yx_1\in E(G)$, then $P=(v_3,v,v_2,w_1,x_1,y\overrightarrow{P_b},\overrightarrow{Q},\overrightarrow{P_c},v_1,w_2\overrightarrow{P_a}x_4)$.  \newline
If $yx_3\in E(G)$, then $P=(v_3,v,v_2,w_1\overrightarrow{P_a}x_3,y\overrightarrow{P_b},\overrightarrow{Q},\overrightarrow{P_c},v_1,w_4,x_4)$. \\
\hline
 Case 1.2: $d(v_3)>1$ &  Let $U=(\mathbb{P}_2\cup\mathbb{P}_3)\cap K$.  Note that every $s\in U$ is adjacent to a vertex in $N^I(w_1)$.  Further, $d(w_1)=d(w_2)=d(w_4)=3$.  It follows that $v_3$ is adjacent to either $w_3$ or an element $z$ in $\mathbb{P}_1$.  Observe that $|\mathbb{P}_2|\le1$.   To complete the proof, we shall assume that $|\mathbb{P}_2|=1$.   Proof of the other case is similar.  Since $|I|\ge9$, there exists $P_c\in\mathbb{P}_3$.  Let $P_d\in\mathbb{P}_2$.  \\
\hline
Case 1.2.1: $v_3w_3\in E(G)$ &   We claim that for all $s\in U$, $sx_3\in E(G)$.  Suppose not, then \NC{w_1}{s} or \NC{w_3}{s} induces a \K14.  Now $d(x_1)=d(x_2)=d(v_2)=d(v_3)=2$.  Since $G$ has no short cycles, there exists $w'\in \mathbb{P}_1$ such that $w'$ is adjacent to at least one vertex in $\{x_1,x_2,v_2,v_3\}$.  We obtain a desired path $P$ as follows. \newline
If $w'x_1\in E(G)$, then \newline $P=(\overrightarrow{P_d},v_1,\overrightarrow{P_c}, w',x_1,w_2,x_2,w_3,v_3,v,v_2,w_1,x_3,w_4,x_4)$. \newline
If $w'x_2\in E(G)$, then \newline $P=(\overrightarrow{P_d},v_1,\overrightarrow{P_c}, w',x_2,w_2,x_1,w_1,v_2,v,v_3,w_3,x_3,w_4,x_4)$. \newline 
If $w'v_2\in E(G)$, then \newline $P=(\overrightarrow{P_d},v_1,\overrightarrow{P_c}, w',v_2,v,v_3,w_3,x_2,w_2,x_1,w_1,x_3,w_4,x_4)$. \newline
If $w'v_3\in E(G)$, then $P=(\overrightarrow{P_d},v_1,\overrightarrow{P_c},w',v_3,v,v_2,w_1\overrightarrow{P_a}x_4)$. 
 \\
\hline
Case 2: $v_1w_3\nin E(G)$ &   In this case we shall assume $v_1w_1\in E(G)$.   Note that $w_3$ is adjacent to $v_2$ or $v_3$.  Without loss of generality, we shall assume $w_3v_2\in E(G)$.  Recall that all the end vertices of paths in $\mathbb{P}_2\cup \mathbb{P}_3$, in $K$ are adjacent to $v_1$.  We first consider the case when $d(v_3)=1$.  Note that $d(x_4)=1$.  It follows that $\mathbb{P}_2=\emptyset$.  Further, since $|I|\ge9$, there exists $P_b,P_c\in\mathbb{P}_3$, let $y\in P_b\cap K$.  Note that $N^I(y)\cap N^I(w_3)\neq\emptyset$. Depending on the adjacency of $y$ with $\{v_2,x_3,x_2\}$, we obtain desired path $P$ as follows. \newline
If $yv_2\in E(G)$, then $P=(v_3,v,v_2,y\overrightarrow{P_b},\overrightarrow{P_c},v_1,w_1\overrightarrow{P_a}x_4)$. \newline  
If $yx_3\in E(G)$, then \newline $P=(v_3,v,v_2,w_3,x_2\overleftarrow{P_a}w_1,v_1,\overrightarrow{P_c},\overleftarrow{P_b}y,x_3,w_4,x_4)$.  \newline
\\ \hline
\end{tabular}
\caption{Case analysis for the proof of Claim \ref{HP_P8} } 
\label{tab:HP_P8_2} 
\end{table}
\begin{table}[h!]
\begin{tabular}{ | L{0.18\textwidth} | L{0.795\textwidth} |}
\hline & \\
Case 2: $v_1w_3\nin E(G)$ &   
If $yx_2\in E(G)$, then we see the adjacency of $w_1$ in $\{v_2,x_3,x_2\}$. 
If $w_1v_2\in E(G)$, then $(v_3,v,v_2,w_1,x_1,w_2,v_1,\overrightarrow{P_c},\overleftarrow{P_b}y,x_2\overrightarrow{P_a}x_4)$ is a desired path. 
If $w_1x_3\in E(G)$, then \newline $(v_3,v,v_2,w_3,x_2,y\overrightarrow{P_b},\overrightarrow{P_c},v_1,w_2,x_1,w_1,x_3,w_4,x_4)$ is a desired path.  
If $w_1x_2\in E(G)$, then we see the following.  Note that $d(v_2)=d(x_3)=2$.  Since $G$ has no short $I$-$I$ path, there exists a vertex $w'\in K$ such that $w'$ is an end vertex of a path $P_e\in \mathbb{P}_1\cup \mathbb{P}_3$, and $w'$ is adjacent to one of $v_2,x_3$. \newline
If $w'v_2\in E(G)$, then  $(v_3,v,v_2,w'\overrightarrow{P_e},\overrightarrow{P_b},v_1,w_1\overrightarrow{P_a}x_4)$ is a desired path. \newline 
If $w'x_3\in E(G)$, then $(v_3,v,v_2,w_3,x_2\overleftarrow{P_a}w_1,v_1,\overrightarrow{P_b},\overleftarrow{P_e}w',x_3,w_4,x_4)$ is a desired path.  
%
\newline 
Now we consider $d(v_3)>1$.  Observe that $d(x_4)=1$.  It follows that $|\mathbb{P}_2|\le1$.  Let $P_d\in \mathbb{P}_2$, $P_b\in \mathbb{P}_3$, $y\in P_b\cap K$.  Similar to the proof of Case 2 of Claim \ref{HP_P9}, we observe the following.  Since $d(v_3)>1$, there exists $z\in \mathbb{P}_1$ such that $zv_3\in E(G)$.  Note that $y$ is adjacent to a vertex in $\{v_2,x_2,x_3\}$.  Depending on the adjacency of $y$, we obtain a desired path $P$ as follows.   \newline
If $yv_2\in E(G)$, then $P=(\overrightarrow{P_d},z,v_3,v,v_2,y\overrightarrow{P_b},v_1,w_1\overrightarrow{P_a}x_4)$.
\newline
If $yx_2\in E(G)$, then $P=(\overrightarrow{P_d},v_1,w_1\overrightarrow{P_a}w_3,v_2,v,v_3,z,\overleftarrow{P_b}y,x_2\overrightarrow{P_a}x_4)$.
\newline
If $yx_3\in E(G)$, then \newline $P=(\overrightarrow{P_d},z,v_3,v,v_2,w_3\overleftarrow{P_a}w_1,v_1,\overleftarrow{P_b}y,x_3,w_4,x_4)$.
\\
\hline
Case 3: $v_1t\nin E(G)$, $t\in \mathbb{P}_1$ &   In this case we shall assume $v_1w_i\in E(G)$, $1\le i\le4$.   Note that $t$ is adjacent to $v_2$ or $v_3$.  Without loss of generality, we shall assume $tv_2\in E(G)$.  We first consider the case when $d(v_3)=1$.  Note that $d(x_4)=1$.  It follows that $\mathbb{P}_2=\emptyset$.  Further, since $|I|\ge 9$, there exists $P_b,P_c\in \mathbb{P}_3$.  Recall that all the end vertices of paths in $\mathbb{P}_2\cup \mathbb{P}_3$, in $K$ are adjacent to $v_1$.  We obtain $(v_3,v,v_2,t,\overrightarrow{P_b},v_1,w_1\overrightarrow{P_a}x_4)$ as a desired path in $G$.   Now we consider the case when $d(v_3)>1$.  Note that $d(x_4)=1$.  It follows that $|\mathbb{P}_2|\le1$, and let $P_d\in \mathbb{P}_2$.  Note that $t$ is adjacent to at least two vertices in $\{x_1,\ldots,x_4\}$, otherwise for some $w_i$, $2\le i\le 4$, \NC{w_i}{t} induces a \K14.  Observe that for all $s\in (\mathbb{P}_2\cup\mathbb{P}_3)\cap K$, $s$ is adjacent to a vertex in $N^I(t)$.  Since $d(v_3)>1$, there exists $z\in K$ such that $zv_3\in E(G)$.  From the previous arguments, $z\in\mathbb{P}_1$ and $z\ne t$.  We obtain $(\overrightarrow{P_d},z,v_3,v,v_2,t,\overrightarrow{P_b},v_1,w_1\overrightarrow{P_a}x_4)$ as a desired path in $G$. \\
\hline
\end{tabular}  \caption{Case analysis for the proof of Claim \ref{HP_P8} }
\label{tab:HP_P8_3}
\end{table}
\begin{cl}\label{2P6_no_P4} %
If there exists $P_a,P_b\in \mathbb{P}_7\cup\mathbb{P}_6$, then there does not exist $P_c\in\mathbb{S}_f$ such that $P_c\neq P_a$, $P_c\neq P_b$ and $|P_c|\ge 4$.  Further, ${|\mathbb{P}_7|+|\mathbb{P}_6|\le2}$.
\end{cl}
\begin{proof}
Let $P_a=P(w_1,\ldots,w_l;$ $x_1,\ldots,x_3)$ and $P_b=P(s_1,\ldots,s_m;t_1,\ldots,t_{3})$.
%
%
For a contradiction, assume that there exists $P_c\in \mathbb{P}_j$, $j\ge4$ such that $P_c=P(y_1,\ldots,y_{n'};$ $z_1,\ldots,z_n)$, $n+1\ge n'\ge n \ge2$. 
From Corollary \ref{xunivthreepathseven}, there exists $v_1\in N^I(v)$ such that $v_1w_i,v_1s_i,v_1y_j\in E(G)$, $i\in \{2,3\}$, $2\le j\le n$.  Now we claim that $v_1w_1,v_1s_1,v_1y_1\in E(G)$.  Suppose $v_1w_1\notin E(G)$, then by Corollary \ref{cor1}, either $v_2w_1\in E(G)$ or $v_3w_1\in E(G)$.  Note that $w_1t_2\in E(G)$, otherwise, \NC{s_2}{w_1} or \NC{s_3}{w_1} induces $K_{1,4}$.  Similarly, $w_1$ is adjacent to a vertex in $P_d\cap I$.  It follows that \NV{w_1} induces a $K_{1,4}$, a contradiction.  Thus $v_1w_1\in E(G)$.  
Similar arguments hold good for the other edges and $v_1s_1,v_1y_1\in E(G)$.  If the paths $P_a,P_b,P_d$ are odd, then using similar arguments, the other end vertices of those paths are also adjacent to $v_1$.  
Since $K$ is maximal, there exists $w'\in K$ such that $v_1w'\nin E(G)$.  Note that $w'\notin P_a\cup P_b\cup P_d$.  By Corollary \ref{cor1}, $v_2w'\in E(G)$ or $v_3w'\in E(G)$.  Further, similar to the previous arguments, $w'$ is adjacent to a vertex, each in $P_a\cap I$, $P_b\cap I$ and $P_d\cap I$.  Now, \NV{w'} induces a $K_{1,4}$, a final contradiction to the existence of $P_c$.  Observe that $|\mathbb{P}_7|+|\mathbb{P}_6|\le2$ is a direct consequence of the first part of this proof.  This completes a proof of the claim.   $\hfill\qed$
\end{proof}

\begin{cl}\label{HP_2P7_2P6_P7P6}
If there exists $P_a,P_b\in\mathbb{P}_7\cup \mathbb{P}_6$, then $G$ has a Hamiltonian path.
\end{cl}
\begin{proof}
Let $P_a=P(w_1,\ldots,w_{l};x_1,\ldots,x_{3})$ and $P_b=P(s_1,\ldots,s_{m};t_1,\ldots,$ $t_{3})$, $l,m\in \{3,4\}$.  From Claim \ref{xunivtwopathseven}, there exists $v_1\in N^I(v)$ such that for every $w\in \{w_2,w_3,s_2,s_3\}$, $v_1w\in E(G)$.  Now we claim that $v_1w_1,v_1s_1\in E(G)$.  Suppose $v_1w_1\nin E(G)$, then note that $w_1$ is adjacent to either $v_2$ or $v_3$.  Further, $w_1$ is also adjacent to a vertex in $P_b\cap I$, otherwise for some $s\in P_b\cap K$, \NC{s}{w_1} induces a \K14.  Note that $w_1x_2\in E(G)$ or $w_1x_3\in E(G)$, otherwise \NC{w_3}{w_1} induces a \K14.  It follows from the above arguments that \NV{w_1} induces a \K14.  Therefore, $v_1w_1\in E(G)$.  Using similar argument we could establish $v_1s_1\in E(G)$.  Similarly, if the paths $P_a,P_b$ are having odd length, then $v_1w_4,v_1s_4\in E(G)$.  Since $K$ is maximal, there exists a vertex $w'\in K$ such that $v_1w'\nin E(G)$.  From Corollary \ref{cor1}, $v_2w'\in E(G)$ or $v_3w'\in E(G)$.  Without loss of generality, we shall assume $v_2w'\in E(G)$.  Observe that $w'x_2,w't_2\in E(G)$, otherwise for some $w\in \{w_2,w_3,s_2,s_3\}$, \NC{w}{w'} induces a \K14.  From Claim \ref{2P6_no_P4}, there does not exist $P_c$ such that $|P_c|\ge4$.  Further, there are no paths other than $P_a,P_b$ in $\mathbb{P}_6\cup\mathbb{P}_7$.  If $d(v_3)=1$, then we observe the following.  Note that the number of vertices in $I$ with degree 1 is at most two.  Thus there are two possibilities, either $P_a,P_b\in\mathbb{P}_7$, or exactly one of $P_a$ or $P_b$ is in $\mathbb{P}_6$.  We obtain desired path $P$ as follows.  If $\mathbb{P}_6=\emptyset$, then note that $|\mathbb{P}_2|\le1$, let $P_c\in\mathbb{P}_2$.  Then $P=(v_3,v,v_2,w',\overrightarrow{P_a},v_1\overrightarrow{P_b},\overrightarrow{P_c})$.  If exactly one of $P_a,P_b$ is in $\mathbb{P}_6$, then without loss of generality, let $P_b\in \mathbb{P}_6$.  Now $P=(v_3,v,v_2,w',\overrightarrow{P_a},v_1\overrightarrow{P_b})$.  If $d(v_3)>1$, then we observe the following.  Note that there exists $w''\in\mathbb{P}_1$ such that $w''v_3\in E(G)$.  Note that the number of vertices in $I$ with degree 1 is at most two.  Thus there are three possibilities as follows.  
\begin{enumerate}[]
\item \emph{\bf {\em Case} 1:} $|\mathbb{P}_2|=2$. There exists $P_c,P_d\in\mathbb{P}_2$, and $P_a,P_b\in\mathbb{P}_7$.  In this case, $(\overrightarrow{P_c},w'',v_3,v,v_2,w',\overrightarrow{P_a},v_1\overrightarrow{P_b},\overrightarrow{P_d})$ is a desired path.
\item \emph{\bf {\em Case} 2:} $|\mathbb{P}_2|=1$. There exists $P_c\in\mathbb{P}_2$, and exactly one of $P_a$ or $P_b$ is in $\mathbb{P}_6$.  Without loss of generality, let $P_b\in \mathbb{P}_6$.  We obtain $P=(\overrightarrow{P_c},w'',v_3,v,v_2,w',\overrightarrow{P_a},v_1\overrightarrow{P_b})$ as a desired path.  If $P_a,P_b\in \mathbb{P}_7$, then $(\overrightarrow{P},s_4)$ is a desired path.
\item \emph{\bf {\em Case} 3:} $\mathbb{P}_2=\emptyset$.  If $P_a,P_b\in\mathbb{P}_6$, then $R=(t_3,s_3,t_2,w',v_2,v,v_3,w'',s_1,t_1,s_2,v_1,\overrightarrow{P_a})$ is a desired path.  For the other cases, $(s_4,\overrightarrow{R})$ or $(\overrightarrow{R},w_4)$  or $(s_4,\overrightarrow{R},w_4)$ are desired paths.  
\end{enumerate} 
This completes a proof of the claim. $\hfill\qed$
\end{proof}

\begin{cl}\label{P7_P6_no_P5_P4}
If there $|\mathbb{P}_7|+|\mathbb{P}_6|=1$, then $|\mathbb{P}_5|+|\mathbb{P}_4|\le1$.
\end{cl}
\begin{proof}
Let $P_a=P(w_1,\ldots,w_{l};x_1,\ldots,x_{3})$, $l\in\{3,4\}$, $P_b=P(s_1,\ldots,s_{m};t_1,\ldots,t_{2})$, $m\in \{2,3\}$.  Assume for a contradiction that there exists a path $P_c=P(q_1,\ldots,q_{n};$ $r_1,\ldots,r_{2})$, $n\in \{2,3\}$.  From Corollary \ref{xunivthreepathseven}, there exists $v_1\in N^I(v)$ such that $v_1w_2,v_1w_3,v_1s_2,v_1q_2\in E(G)$.  We claim that $v_1w_1,v_1s_1,v_1q_1\in E(G)$.  Suppose $v_1w_1\nin E(G)$, then $w_1$ is adjacent to either $v_2$ or $v_3$, and one each from $P_b\cap I$ and $P_c\cap I$.  Therefore, \NV{w_1} induces a \K14, a contradiction.  Similar argument holds good for the other edges.  If $P_a,P_b,P_c$ are odd paths, then similar to the previous argument, the end vertices $w_l,s_m,q_n$ are adjacent to $v_1$.  Since the clique is maximal, there exists $w'\in K$ such that $v_1w'\nin E(G)$.  From the previous argument, $w'\nin \{w_1,\ldots,w_l,s_1,\ldots,s_m,q_1,\ldots,q_n\}$. By Corollary \ref{cor1}, either $v_2w'\in E(G)$ or $v_3w'\in E(G)$.  Further, we argue that $w'x_2\in E(G)$, otherwise \NC{w_2}{w'} or \NC{w_3}{w'} induces a \K14.  Observe that $w'$ is adjacent to one of $\{t_1,t_2\}$, otherwise \NC{s_2}{w'} induces a \K14.  Similarly, $w'$ is adjacent to one of $\{r_1,r_2\}$.  Now, \NV{w'} induces a $K_{1,4}$, which is a final contradiction to the existence of such a path $P_c$.  This completes a proof of the claim.   $\hfill\qed$
\end{proof}

\begin{cl}\label{HP_P7_P6_P5_P4}
If there exists $P_a\in\mathbb{P}_7\cup \mathbb{P}_6$, $P_b\in\mathbb{P}_5\cup \mathbb{P}_4$, then $G$ has a Hamiltonian path.  
\end{cl}
\begin{proof}
Let $P_a=P(w_1,\ldots,w_{l};x_1,\ldots,x_{3})$, $l\in\{3,4\}$, $P_b=P(s_1,\ldots,s_{m};t_1,\ldots,t_{2})$, $m\in \{2,3\}$.  From Corollary \ref{xunivthreepathseven}, there exists $v_1\in N^I(v)$ such that $v_1w_2,v_1w_3,v_1s_2\in E(G)$.  We argue that $v_1w_1\in E(G)$.  Suppose not, then from Corollary \ref{cor1}, $w_1$ is adjacent to either $v_2$ or $v_3$.  Further, $w_1$ is adjacent to one of $\{x_2,x_3\}$, otherwise \NC{w_3}{w_1} induces a \K14. $w_1$ is also adjacent to one of $\{t_1,t_2\}$, otherwise \NC{s_2}{w_1} induces a \K14.  It follows that \NV{w_1} induces a \K14, a contradiction.  Thus $v_1w_1\in E(G)$.  If $P_a$ is an odd path, then similar arguments holds good with respect to the other end vertex of $P_a$ and $v_1w_l\in E(G)$.  From Claim \ref{P7_P6_no_P5_P4}, there are no paths of size 4 or more other than $P_a,P_b$ in $\mathbb{S}_f$.  Now we claim that for an arbitrary path $P\in\mathbb{P}_2\cup \mathbb{P}_3$ and $s\in K$ be an end vertex of $P$, then $v_1s\in E(G)$.  Suppose not, then $v_2s\in E(G)$ or $v_3s\in E(G)$.  Further, $s$ is adjacent to $x_2$, otherwise \NC{w_2}{s} or \NC{w_3}{s} induces a \K14.  Similarly, $s$ is adjacent to one of $t_1,t_2$, otherwise \NC{s_2}{s} induces a \K14.  From the above arguments it follow that \NV{s} induces a \K14, a contradiction.  Therefore, $v_1s\in E(G)$. Note that the above argument is true for any end vertex $s\in K$ of every such paths in $\mathbb{P}_2\cup\mathbb{P}_3$.  Since the clique is maximal, there exists $y\in K$ such that $v_1y\nin E(G)$.  Now we see the following cases depending on the length of paths $P_a,P_b$, and in each case, we also see the possibility of adjacency of $y$. 
\begin{enumerate}[]
\item \emph{\bf {\em Case} 1:} $|P_a|=7,|P_b|=5$ and $v_1y\nin E(G)$.  
\begin{enumerate}[]
\item \emph{\bf {\em Case} 1.1:}  $y\in \{s_1,s_3\}$, without loss of generality, let $s_1v_1\nin E(G)$.  From Corollary \ref{cor1}, $s_1v_2\in E(G)$ or $s_1v_3\in E(G)$.  Without loss of generality, let $s_1v_2\in E(G)$.  Note that $s_1x_2\in E(G)$, otherwise \NC{w_2}{s_1} or \NC{w_3}{s_1} induces a \K14.  Consider $d(v_3)=1$.  Since there are at most 2 vertices in $I$ of degree 1, $|\mathbb{P}_2|\le1$, let $P_c\in\mathbb{P}_2$.  We obtain $(v_3,v,v_2,s_1\overrightarrow{P_b},\overrightarrow{P_a},v_1,\overrightarrow{P_c})$ as a desired path in $G$.  If $d(v_3)>1$, then we observe the following.  Since there are at most 2 vertices in $I$ of degree 1, $|\mathbb{P}_2|\le2$, let $P_c,P_d\in\mathbb{P}_2$.  Observe that $w_1,w_4$ are adjacent to a vertex in $N^I(s_1)$.  Thus $d^I(w_1)=d^I(w_4)=3$.  We now claim that there does not exist a vertex $s\in K$ such that $s$ is an end vertex of a path in $\mathbb{P}_2\cup\mathbb{P}_3$ and $v_3s\in E(G)$.  Suppose, for such a vertex $s$, let $v_3s\in E(G)$, then \NC{s_1}{s} induces a \K14.  From the above observations we conclude that $v_3$ is adjacent to either $s_3$ or a vertex in $\mathbb{P}_1$.  If $v_3s_3\in E(G)$, then note that $s_3x_2\in E(G)$, otherwise for some $w\in\{w_2,w_3,s_1\}$, \NC{w}{s_3} induces a \K14.  Further, $w_1x_2,w_4x_2\in E(G)$.  Note that $d(v_2)=d(v_3)=d(t_1)=d(t_2)=2$.  Since $G$ has no short cycles, there exists  $w'\in\mathbb{P}_1$ such that $w'$ is adjacent to one among $\{v_2,v_3,t_1,t_2\}$.  In each cases, we obtain a desired path $P$ as follows.  \newline
If $w'v_2\in E(G)$, then $P=(\overrightarrow{P_c},w',v_2,v,v_3,s_3\overleftarrow{P_b}s_1,\overrightarrow{P_a},v_1,\overrightarrow{P_d})$.  \newline
If $w'v_3\in E(G)$, then $P=(\overrightarrow{P_c},w',v_3,s_3\overleftarrow{P_b}s_1,v_2,v,v_1,\overrightarrow{P_a},\overrightarrow{P_d})$.  \newline
If $w't_1\in E(G)$, then $P=(\overrightarrow{P_c},w',t_1,s_1,v_2,v,v_3,s_3,t_2,s_2,v_1,\overrightarrow{P_a},\overrightarrow{P_d})$.  \newline
If $w't_2\in E(G)$, then $P=(\overrightarrow{P_c},w',t_2,s_3,v_3,v,v_2,s_1,t_1,s_2,v_1,\overrightarrow{P_a},\overrightarrow{P_d})$.  \newline
For a vertex $z\in \mathbb{P}_1$, if $v_3z\in E(G)$, then $P=(\overrightarrow{P_c},z,v_3,v,v_2,s_1\overrightarrow{P_b}s_3,\overrightarrow{P_a},v_1,\overrightarrow{P_d})$.  
\item \emph{\bf {\em Case} 1.2:} $y\in\mathbb{P}_1$.  In this case we shall assume that $v_1s_1,v_1s_3\in E(G)$.  From Corollary \ref{cor1}, $yv_2\in E(G)$ or $yv_3\in E(G)$.  Without loss of generality, let $yv_2\in E(G)$.
 Consider $d(v_3)=1$.  Since there are at most 2 vertices in $I$ of degree 1, $|\mathbb{P}_2|\le1$, let $P_c\in\mathbb{P}_2$.  We obtain $(v_3,v,v_2,y,\overrightarrow{P_b},\overrightarrow{P_a},v_1,\overrightarrow{P_c})$ as a desired path in $G$. On the other hand, if $d(v_3)>1$, then we observe the following.  Since there are at most 2 vertices in $I$ of degree 1, $|\mathbb{P}_2|\le2$, let $P_c,P_d\in\mathbb{P}_2$.  If $v_3s_1\in E(G)$, then $(\overrightarrow{P_c},s_3\overleftarrow{P_b}s_1,v_3,v,v_2,y,\overrightarrow{P_a},v_1,\overrightarrow{P_d})$ is a desired path in $G$.
If $v_3s_3\in E(G)$, then $(\overrightarrow{P_c},s_1\overrightarrow{P_b}s_3,v_3,v,v_2,y,\overrightarrow{P_a},v_1,\overrightarrow{P_d})$ is a desired path in $G$.
If $v_3z\in E(G)$ for some $z\in\mathbb{P}_1$, then we see the following.  Observe that $yx_2\in E(G)$ and there exists a vertex in $P_b\cap I$ adjacent to $y$.  Therefore, $d^I(y)=3$ and $z\ne y$.  We obtain $(\overrightarrow{P_c},z,v_3,v,v_2,y,\overrightarrow{P_b},\overrightarrow{P_a},v_1,\overrightarrow{P_d})$ as a desired path in $G$.
\end{enumerate}
\item \emph{\bf {\em Case} 2:}$|P_a|=7,|P_b|=4$ and $v_1y\nin E(G)$. 
\begin{enumerate}[]
\item \emph{\bf {\em Case} 2.1:} $y=s_1$.  From Corollary \ref{cor1}, $s_1v_2\in E(G)$ or $s_1v_3\in E(G)$.  Without loss of generality, let $s_1v_2\in E(G)$.  Note that $s_1x_2\in E(G)$, otherwise \NC{w_2}{s_1} or \NC{w_3}{s_1} induces a \K14.  Consider $d(v_3)=1$.  Since there are at most 2 vertices in $I$ of degree 1 and $d(t_2)=1$, $\mathbb{P}_2=\emptyset$.  Now we see that $d(t_1)=d(v_2)=2$.  Since $G$ has no short $I$-$I$ path, $d(t_1)\ge3$ or $d(v_2)\ge3$.  Depending on the adjacency of $t_1,v_2$, we see the following cases, and obtain a desired path $P$ in all the possibilities.  Note that $w_1$ and $w_4$ are adjacent to a vertex in $N^I(s_1)$. \newline
If $w_1t_1\in E(G)$, then $P=(v_3,v,v_2,s_1,t_1,\overrightarrow{P_a},v_1,s_2,t_2)$.  \newline
If $w_1v_2\in E(G)$, then $P=(v_3,v,v_2,w_1,x_1,w_2,v_1,w_4,x_3,w_3,x_2,s_1\overrightarrow{P_b})$.  \newline
If $w_4t_1\in E(G)$ or $w_4v_2\in E(G)$, then we obtain a similar path.  
If $w_1x_2,w_4x_2\in E(G)$, then there exists $P_e\in \mathbb{P}_1\cup \mathbb{P}_3$, $z$ is an end vertex in $P_e$ and $z$ is adjacent to either $t_1$ or $v_2$.  \newline
If $zt_1\in E(G)$, then $P=(v_3,v,v_2,s_1,t_1,z\overrightarrow{P_e},\overrightarrow{P_a},
v_1,s_2,t_2)$.  \newline
If $zv_2\in E(G)$, then $P=(v_3,v,v_2,z\overrightarrow{P_e},w_1,x_1,w_2,v_1,w_4,x_3,w_3,x_2,s_1\overrightarrow{P_b})$.  \newline
Now we consider the case where $d(v_3)>1$.  Since there are at most 2 vertices in $I$ of degree 1 and $d(t_2)=1$, $|\mathbb{P}_2|\le1$, let $P_c\in\mathbb{P}_2$, let $P_c\in\mathbb{P}_2$.  Observe that there exists $t\in\mathbb{P}_1$ such that $tv_3\in E(G)$.  We obtain $(\overrightarrow{P_c},v_1,\overrightarrow{P_a},t,v_3,v,v_2,\overrightarrow{P_b})$ as a desired path in $G$.
\item \emph{\bf {\em Case} 2.2:} $y\in\mathbb{P}_1$.  We shall assume that $v_1s_1\in E(G)$. Note that $yx_2\in E(G)$, otherwise \NC{w_2}{y} or \NC{w_3}{y} induces a \K14.  Further, $yt_1\in E(G)$, otherwise \NC{s_2}{y} induces a \K14.  Consider $d(v_3)=1$.  Since there are at most 2 vertices in $I$ of degree 1 and $d(t_2)=1$, $\mathbb{P}_2=\emptyset$.  We obtain a desired path $P$ as follows.  $P=(v_3,v,v_2,y,\overrightarrow{P_a},v_1,\overrightarrow{P_b})$.  If $d(v_3)>1$, then note that either $v_3s_1\in E(G)$ or for some $z\in\mathbb{P}_1$, $zv_3\in E(G)$.  Since there are at most 2 vertices in $I$ of degree 1 and $d(t_2)=1$, $|\mathbb{P}_2|\le1$, let $P_c\in\mathbb{P}_2$.   If $v_3s_1\in E(G)$, then $P=(\overrightarrow{P_c},v_1,\overrightarrow{P_a},y,v_2,v,v_3,\overrightarrow{P_b})$.  If $zv_3\in E(G)$, then $P=(\overrightarrow{P_c},v_1,\overrightarrow{P_a},z,v_3,v,v_2,y,\overrightarrow{P_b})$.
\end{enumerate}

%
\item \emph{\bf {\em Case} 3:} $|P_a|=6,|P_b|=5$ and $v_1y\nin E(G)$. \\
\begin{enumerate}[  ] 
\item \emph{\bf {\em Case} 3.1:} $y\in\{s_1,s_3\}$.  Without loss of generality, let $y=s_1$.  From Corollary \ref{cor1}, $s_1v_2\in E(G)$ or $s_1v_3\in E(G)$.  Without loss of generality, let $s_1v_2\in E(G)$.  Note that $s_1x_2\in E(G)$, otherwise \NC{w_2}{s_1} or \NC{w_3}{s_1} induces a \K14.  We see the following cases and obtain a desired path $P$ as follows.  We first consider $d(v_3)=1$.  Note that $\mathbb{P}_2=\emptyset$.  Since $|I|\ge9$, there exists $P_d\in\mathbb{P}_3$.  Then $P=(v_3,v,v_2,\overrightarrow{P_b},\overrightarrow{P_d},v_1,\overrightarrow{P_a})$. If $d(v_3)>1$, then we see the following.  Note that $|\mathbb{P}_2|\le1$.  We consider the case $|\mathbb{P}_2|=1$, let $P_c\in\mathbb{P}_2$.  Consider $Q\in\mathbb{P}_2\cup\mathbb{P}_3$.  For any vertex $w'\in Q\cap K$, we claim that $v_3w'\nin E(G)$.  Suppose $v_3$ is adjacent to such a vertex $w'\in Q\cap K$, then we know that $w'v_1\in E(G)$.  Further, $w'$ is also adjacent to a vertex in $N^I(s_1)$.  Observe that \NV{w'} induces a \K14, a contradiction.  Therefore, $v_3w'\nin E(G)$.  This is true for any such vertex $w'$.  It follows that, $v_3$ is adjacent to $s_3$ or a vertex in $\mathbb{P}_1$.  Depending on the adjacency of $v_3$, we see the following sub cases. 
\begin{enumerate}[]
\item \emph{\bf {\em Case} 3.1.1:} $v_3s_3\in E(G)$.  We claim that $s_3x_2\in E(G)$.  Suppose not, then there exists $w\in\{w_2,w_3,s_1\}$ such that \NC{w}{s_3} induces a \K14.  We also observe that $w_1x_2\in E(G)$, otherwise \NC{s_1}{w_1} or \NC{s_3}{w_1} induces a \K14.  Now observe that $d(t_1)=d(t_2)=d(v_2)=d(v_3)=2$.  Since we consider $G$ with no short cycles, there exists a vertex $w''\in K$ such that $w''$ is adjacent to a vertex in $\{t_1,t_2,v_2,v_3\}$.  Now we observe that $w''$ is not an end vertex of a path in $\mathbb{P}_2\cup\mathbb{P}_3$.  Suppose not, then recall that $w''$ is adjacent to $v_1$.  Further, \NC{s_1}{w''} or \NC{s_2}{w''} induces a \K14, a contradiction.  Thus we conclude $w''\in\mathbb{P}_1$.
We obtain the following desired path depending on the adjacency of $w''$ with $\{t_1,t_2,v_2,v_3\}$.\\
If $w''t_1\in E(G)$, then $P=(\overrightarrow{P_c},w'',t_1,s_1,v_2,v,v_3,s_3,t_2,s_2,v_1,\overrightarrow{P_a})$\\
If $w''v_2\in E(G)$, then, $P=(\overrightarrow{P_c},w'',v_2,\overrightarrow{P_b},v_3,v,v_1,\overrightarrow{P_a})$\\
If $w''t_2\in E(G)$, then, $P=(\overrightarrow{P_c},w'',t_2,s_3,v_3,v,v_2,s_1,t_1,s_2,v_1,\overrightarrow{P_a})$\\
If $w''v_3\in E(G)$, then, $P=(\overrightarrow{P_c},w'',v_3,s_3\overleftarrow{P_b},v_2,v,v_1,\overrightarrow{P_a})$
\item \emph{\bf {\em Case} 3.1.2:} For $z\in\mathbb{P}_1$, $v_3z\in E(G)$.  Note that $|\mathbb{P}_2|\le1$.  If $|\mathbb{P}_2|=1$, then let $P_c\in\mathbb{P}_2$.  If $\mathbb{P}_2=\emptyset$, then there exists $P_c\in\mathbb{P}_3$ as  $|I|\ge9$.
Note that $s_3$ is adjacent to a vertex in $N^I(v)$, and depending on the adjacency, we obtain a desired path $P$ as follows.\\
If $v_1s_3\in E(G)$, then $P=(\overrightarrow{P_c},z,v_3,v,v_2,\overrightarrow{P_b},v_1,\overrightarrow{P_a})$. \\
If $v_2s_3\in E(G)$, then $P=(\overrightarrow{P_c},z,v_3,v,v_2,s_3\overleftarrow{P_b},x_2\overleftarrow{P_a}w_1,v_1,w_3,x_3)$. \\
If $v_3s_3\in E(G)$, then $P=(\overrightarrow{P_c},z,v_3,s_3\overleftarrow{P_b},v_2,v,v_1,\overrightarrow{P_a})$.
\end{enumerate}
\item \emph{\bf {\em Case} 3.2:} $y\in\mathbb{P}_1$.  In this case we shall assume that $v_1s_1,v_1s_2\in E(G)$.  From Corollary \ref{cor1}, $yv_2\in E(G)$ or $yv_3\in E(G)$.  Without loss of generality, let $yv_2\in E(G)$.  Note that $yx_2\in E(G)$, otherwise \NC{w_2}{y} or \NC{w_3}{y} induces a \K14.  We see the following cases and obtain a desired path $P$ as follows.  We first consider $d(v_3)=1$.  Note that $\mathbb{P}_2=\emptyset$.  Then $P=(v_3,v,v_2,y,\overrightarrow{P_b},v_1,\overrightarrow{P_a})$. If $d(v_3)>1$, then we see the following.  Note that $|\mathbb{P}_2|\le1$, let $P_c\in\mathbb{P}_2$.  Clearly, $v_3$ is adjacent to $s_1$ or $s_3$ or a vertex $z\in\mathbb{P}_1$.  If $v_3s_1\in E(G)$, then we obtain $(\overrightarrow{P_c},y,v_2,v,v_3,\overrightarrow{P_b},v_1,\overrightarrow{P_a})$ as a desired path.  $v_3s_3\in E(G)$ is a symmetric case.  Note that $y$ is adjacent to either $t_1$ or $t_2$, otherwise, \NC{s_2}{y} induces a \K14.  If $v_3z\in E(G)$ where $z\in\mathbb{P}_1$, $z\neq y$, then  $(\overrightarrow{P_c},y,v_2,v,v_3,z,\overrightarrow{P_b},v_1,\overrightarrow{P_a})$ is a desired path.
 
\end{enumerate}

\item \emph{\bf {\em Case} 4:} $|P_a|=6,|P_b|=4$ and $v_1y\nin E(G)$. 
\begin{enumerate}[]
\item \emph{\bf {\em Case} 4.1:} $y=s_1$.   From Corollary \ref{cor1}, $s_1v_2\in E(G)$ or $s_1v_3\in E(G)$.  Without loss of generality, let $s_1v_2\in E(G)$.  Note that $s_1x_2\in E(G)$, otherwise \NC{w_2}{s_1} or \NC{w_3}{s_1} induces a \K14.  Note that there exists $z\in\mathbb{P}_1$ such that $zv_3\in E(G)$.  Further, since $|I|\ge9$, there exists $P_d\in\mathbb{P}_3$.  We obtain $(\overleftarrow{P_a},v_1,\overrightarrow{P_d},z,v_3,v,v_2,\overrightarrow{P_b})$ as a desired path in $G$.  
\item \emph{\bf {\em Case} 4.2:} $y\in \mathbb{P}_1$.  Note that $yx_2\in E(G)$ otherwise \NC{w_2}{y} or \NC{w_3}{y} induces a \K14.  Further, $yt_1\in E(G)$, otherwise, \NC{s_2}{y} induces a \K14.  Since $d(v_3)>1$, there exists $z\in \mathbb{P}_1$, $z\ne y$ and $v_3z\in E(G)$.  Now we obtain $(\overleftarrow{P_a},v_1,\overrightarrow{P_d},z,v_3,v,v_2,y,\overrightarrow{P_b})$ as a desired path in $G$.
\end{enumerate}

\end{enumerate}
This completes the case analysis and a proof of Claim \ref{HP_P7_P6_P5_P4}. $\hfill\qed$
\end{proof}
\begin{cl}\label{HP_P7_P3}
If there exists $P_a\in\mathbb{P}_7$ and there does not exist $P\in\mathbb{S}_f$ such that $P\neq P_a$ and $|P|\ge 4$, then $G$ has a Hamiltonian path.  
\end{cl}
\begin{proof}
Let $P_a=(w_1,\ldots,w_4;,x_1,\ldots,x_3)$.  From Corollary \ref{cor1}, the vertices $w_2,w_3$ are adjacent to at least one of the vertices in $N^I(v)$.  Depending on this adjacency, we see the following two cases.
\begin{table}[h!]
\begin{tabular}{ | L{0.1\textwidth} | L{0.89\textwidth} |}
\hline
\textbf{Case} & \textbf{Arguments}\\
\hline
Case 1 &  There exists $v_1\in N^I(v)$ such that $v_1w_2,v_1w_3\in E(G)$. Consider a path $Q\in\mathbb{P}_2\cup\mathbb{P}_3$, and $s\in Q\cap K$.  We first claim that $s$ is adjacent to either $v_1$ or $x_2$.  Suppose that $v_1s,x_2s\nin E(G)$.   Note that $v_2s\in E(G)$ or $v_3s\in E(G)$.  If $sx_1\nin E(G)$, then \NC{w_2}{s} induces a \K14.  Therefore,  $sx_1\in E(G)$ and similarly, $sx_3\in E(G)$, otherwise \NC{w_3}{s} induces a \K14.  It follows that \NV{s} induces a \K14.   \\ 
%
\hline
 & Therefore, all the end vertices of all such paths are adjacent to either $v_1$ or $x_2$.  Since the clique is maximal, there exists $v'\in K$ such that $v_1v'\nin E(G)$.  We further classify based on the possibilities of $v'$ as follows. \\ \hline
Case 1.1  &  $v'\in P_a$; i.e., $w_1$ or $w_4$ is non-adjacent to $v_1$.  Without loss of generality, let $v_1w_4\nin E(G)$.  Note that $w_4$ is adjacent to either $v_2$ or $v_3$. Without loss of generality, let $v_2w_4\in E(G)$.   \\ \hline
Case 1.1.1 & $d(v_3)=1$.  Since $|I|\ge9$, there exists $P_d,P_e\in\mathbb{P}_3$; $P_d=P(s_1,s_2;t_1)$, and $P_e=P(q_1,q_2;r_1)$.  Further, note that $|\mathbb{P}_2|\le1$.  If $|\mathbb{P}_2|=1$, then let $P_b\in\mathbb{P}_2$.  Recall that the vertices $s_1,s_2,q_1,q_2$ are adjacent to at least one of $v_1,x_2$.  We obtain a desired path $P$ as follows. \newline
 If $s_1v_1,q_1v_1\in E(G)$, then $P=(v_3,v,v_2,\overleftarrow{P_a},\overleftarrow{P_d},v_1,\overrightarrow{P_e},\overrightarrow{P_b})$.  \newline
 If $s_1x_2,q_1x_2\in E(G)$, then $P=(v_3,v,v_2,w_4,x_3,w_3,v_1,w_2,x_1,w_1,\overleftarrow{P_d},x_2,\overrightarrow{P_e},\overrightarrow{P_b})$.   \newline
 If $s_1v_1,q_1x_2\in E(G)$, then $P=(v_3,v,v_2,w_4,x_3,w_3,v_1,\overrightarrow{P_d},\overleftarrow{P_e},x_2\overleftarrow{P_a},\overrightarrow{P_b})$.  \newline
  If $s_1x_2,q_1v_1\in E(G)$, then $P=(v_3,v,v_2,w_4,x_3,w_3,v_1,\overrightarrow{P_e},\overleftarrow{P_d},x_2\overleftarrow{P_a},\overrightarrow{P_b})$.  
\\ \hline
Case 1.1.2 & $d(v_3)>1$.  Note that $w_4$ is adjacent to a vertex in $N^I(w_2)$.  In particular, either $w_4x_1\in E(G)$ or $w_4x_2\in E(G)$.  Thus the only vertex to which $v_3$ is adjacent in $P_a$ is $w_1$.  Consider $v_3w_1\in E(G)$.  Since $|I|\ge9$, there exists $P_d\in\mathbb{P}_3$; $P_d=P(s_1,s_2;t_1)$.  Note that $|\mathbb{P}_2|\le2$.  If $|\mathbb{P}_2|=2$, then let $P_b,P_c\in\mathbb{P}_2$. Let $y_1=P_b\cap K$, $y_2=P_c\cap K$.  Recall that $y_1,y_2,s_1,s_2$ are adjacent to either $v_1$ or $x_2$.  Let $C=(w_2,x_1,w_1,v_3,v,v_2,w_4,x_3,w_3,x_2,w_2)$ and $C'=(w_2,x_1,w_1,v_3,v,v_2,w_4,x_3,w_3,v_1,w_2)$.  We obtain desired path $P$ as follows.   If $s_1v_1\in E(G)$ then we observe the following.  If $y_1x_2\in E(G)$, then $P=(\overrightarrow{P_b},x_2,w_2\overrightarrow{C'}v_1,\overrightarrow{P_d},\overrightarrow{P_c})$.  Note $y_2x_2\in E(G)$ is a symmetric case.  If $y_1v_1,y_2v_1\in E(G)$, then note that $y_1$ is adjacent to a vertex $s$ in $N^I(w_4)$.  Note $N^I(w_4)\subset C$.  We obtain $P=(\overrightarrow{P_b},s\overrightarrow{C},\overrightarrow{P_d},v_1,\overrightarrow{P_c})$.  If $s_1x_2\in E(G)$ and  $y_1v_1\in E(G)$, then $P=(\overrightarrow{P_b},v_1,w_2\overrightarrow{C}x_2,\overrightarrow{P_d},\overrightarrow{P_c})$.   Note that $s_1x_2,y_2v_1\in E(G)$ is a symmetric case.   \\
\hline
\end{tabular}  
\end{table}
\begin{table}[h!]
\begin{tabular}{ | L{0.1\textwidth} | L{0.89\textwidth} |}
\hline
Case 1.1.2 &
If $s_1x_2,y_1x_2,y_2x_2\in E(G)$, then note that $y_1$ is adjacent to a vertex $s$ in $N^I(v)$.  Note $N^I(v)\subset C'$.  We obtain $P=(\overrightarrow{P_b},s\overrightarrow{C'},\overrightarrow{P_d},x_2,\overrightarrow{P_c})$. 
Now we consider the case in which $v_3$ is adjacent to a vertex in $\mathbb{P}_2\cup\mathbb{P}_3$.  If $v_3y_1\in E(G)$, then we observe the following.  \u{Observe that either $y_1v_1\in E(G)$ or $y_1x_2\in E(G)$.  If $w_4x_1\in E(G)$, then \NC{y_1}{w_4} induces a \K14.  Therefore, $w_4x_2\in E(G)$.}  Further, all the end vertices of paths in $\mathbb{P}_2\cup\mathbb{P}_3$ which are in $K$ are adjacent to $x_2$.  We obtain $(\overrightarrow{P_b},v_3,v,v_2,w_4,x_3,w_3,v_1,w_2,x_1,w_1,\overrightarrow{P_d},x_2,\overrightarrow{P_c})$ as a desired path.  If $v_3s_1\in E(G)$, then similar to the arguments for with respect to the vertex $y_1$, all the end vertices of paths in $\mathbb{P}_2\cup\mathbb{P}_3$ which are in $K$ are adjacent to $x_2$.  We obtain $(\overrightarrow{P_b},x_2,s_2,t_1,s_1,v_3,v,v_2,w_4,x_3,w_3,v_1,w_2,x_1,w_1,\overrightarrow{P_c})$ as a desired path. Now we shall consider the case in which $v_3$ is adjacent to a vertex in $\mathbb{P}_1$; for $w'\in\mathbb{P}_1$, let $v_3w'\in E(G)$.  We obtain desired path $P$ as follows.  
 If $y_1x_2,s_1x_2\in E(G)$, then \newline $P=(\overrightarrow{P_b},x_2,\overrightarrow{P_d},w',v_3,v,v_2,w_4,x_3,w_3,v_1,w_2,x_1,w_1,\overrightarrow{P_c})$.
\newline
 If $y_1v_1,s_1v_1\in E(G)$, then $P=(\overrightarrow{P_b},v_1,\overrightarrow{P_d},w',v_3,v,v_2,w_4\overleftarrow{P_a}w_1,\overrightarrow{P_c})$.  
\newline
 If $y_1x_2,s_1v_1\in E(G)$, then  \newline $P=(\overrightarrow{P_b},x_2\overleftarrow{P_a}w_1,\overleftarrow{P_d},v_1,w_3,x_3,w_4,v_2,v,v_3,w',\overrightarrow{P_c})$.  Note $y_1v_1,s_1x_2\in E(G)$ is a symmetric case.
\\ \hline 
 Case 1.2 & $v'\in\mathbb{P}_2$.  In this case we shall assume that $v_1w_1,v_1w_4\in E(G)$.  Let $P_b\in\mathbb{P}_2$, $P_b=P(y_1;z_1)$ and $v'=y_1$.     Note that $y_1$ is adjacent to either $v_2$ or $v_3$. Without loss of generality, let $v_2y_1\in E(G)$.  Observe that $y_1x_2\in E(G)$.  We first see the case in which $d(v_3)=1$.  Note that the vertices $w_1,w_4$ are both adjacent to either $v_2$ or $x_2$.  We obtain the following desired paths.  \newline
 If $v_2w_4\in E(G)$, then $P=(v_3,v,v_2,w_4,x_3,w_3,v_1,w_1\overrightarrow{P_a}x_2,\overrightarrow{P_b})$. 
 \newline
 If $v_2w_1\in E(G)$, then $P=(v_3,v,v_2,w_1,x_1,w_2,v_1,w_4\overleftarrow{P_a}x_2,\overrightarrow{P_b})$. 
 \newline
 If $x_2w_1,x_2w_4\in E(G)$, then observe that there exists a path $P_f\in \mathbb{P}_3\cup\mathbb{P}_1$ such that an end vertex of $P_f$ is adjacent to $v_2$, otherwise $\{v_3,v,v_2,y_1,z_1\}$ induces a short $I$-$I$ path.  We obtain \u{ $P=(v_3,v,v_2,\overrightarrow{P_f},w_1,x_1,w_2,v_1,w_4\overleftarrow{P_a}x_2,\overrightarrow{P_b})$. }
Now we see the case in which $d(v_3)>1$.  Similar to Case 1.1.2, there exists paths $P_d\in\mathbb{P}_3$ and $P_c\in\mathbb{P}_2$, $P_c=P(y_2,z_2)$.  Note that $y_1x_2\in E(G)$.  \u{Observe that $w_1x_2,w_4x_2\in E(G)$.  Further, if $v_3y_2\in E(G)$, then $y_2x_2\in E(G)$. } Note that there exists a path $P_f\in\mathbb{P}_3\cup\mathbb{P}_1$ such that an end vertex of $P_f\cap K$ is adjacent to either $v_2$ or $v_3$, otherwise, $\{z_2,y_2,v_3,v,v_2,y_1,z_1\}$ induces a short $I$-$I$ path.  Depending on the adjacency of the path $P_f$, We obtain $(\overrightarrow{P_b},v_2,\overrightarrow{P_f},\overrightarrow{P_a},v_1,v,v_3,\overrightarrow{P_c})$ or $(\overrightarrow{P_b},v_2,v,v_3,\overrightarrow{P_f},w_1,x_1,w_2,v_1,w_4\overleftarrow{P_a}x_2,\overrightarrow{P_c})$ as a desired path. If $v_3s_1\in E(G)$, then $w_1x_2,w_4x_2\in E(G)$.  Further, $y_1x_2,y_2x_2\in E(G)$. \\
\hline
\end{tabular}
\end{table}
\begin{table}[h!]
\begin{tabular}{ | L{0.1\textwidth} | L{0.89\textwidth} |}
\hline \\[-9pt]
 Case 1.2 &
 We obtain $(\overrightarrow{P_b},v_2,v,v_3,\overrightarrow{P_d},w_1,x_1,w_2,v_1,w_4\overleftarrow{P_a}x_2,\overrightarrow{P_c})$ as a desired path.  \u{ For some $w'\in \mathbb{P}_1$, if $v_3w'\in E(G)$, then we observe the following.  Recall that $y_2v_1\in E(G)$ or  $y_2x_2\in E(G)$.  Depending on the adjacency of $y_2$, we obtain  $(\overrightarrow{P_b},v_2,v,v_3,w',w_1,x_1,w_2,v_1,w_4\overleftarrow{P_a}x_2,\overrightarrow{P_c})$ or $(\overrightarrow{P_b},v_2,v,v_3,w',\overrightarrow{P_a},v_1,\overrightarrow{P_c})$ as a desired path. }
\\ \hline
Case 1.3 & $v'\in\mathbb{P}_3$.  In this case we shall assume that $v_1w_1,v_1w_4\in E(G)$.  Let $P_d\in\mathbb{P}_3$, $P_d=P(s_1,s_2;y_1)$ and $v'=s_1$.   Note that $s_1$ is adjacent to either $v_2$ or $v_3$.  Without loss of generality, let $s_1v_2\in E(G)$.  If $d(v_3)=1$, then $|\mathbb{P}_2|\le1$.  If $|\mathbb{P}_2|=1$, then let $P_b\in\mathbb{P}_2$.  We obtain  $(v_3,v,v_2,\overrightarrow{P_d},w_1,x_1,w_2,v_1,w_4\overleftarrow{P_a}x_2,\overrightarrow{P_b})$ or $(v_3,v,v_2,\overrightarrow{P_d},w_1\overrightarrow{P_a}w_4,v_1,\overrightarrow{P_b})$ as a desired path.  If $d(v_3)>1$, then similar to Case 1.2, there exists desired path in all possibilities.
\\ \hline
Case 1.4 & $v'\in\mathbb{P}_1$. Let $P_d\in\mathbb{P}_3$, $P_d=P(s_1,s_2;y_1)$.  If $d(v_3)=1$, then $|\mathbb{P}_2|\le1$.  If $|\mathbb{P}_2|=1$, then let $P_b\in\mathbb{P}_2$.  Now $(v_3,v,v_2,v',w_1,x_1,w_2,v_1,w_4\overleftarrow{P_a}x_2,\overrightarrow{P_b})$ or $(v_3,v,v_2,v',w_1\overrightarrow{P_a}w_4,v_1,\overrightarrow{P_b})$ is a desired path.  If $d(v_3)>1$, then note that $|\mathbb{P}_2|\le2$.  Further let $P_b,P_c\in \mathbb{P}_2$.  If $v_3$ is adjacent to an end vertex of a path in $\mathbb{P}_2\cup\mathbb{P}_3$, then we obtain one of the following as a desired path.  $Q_1=(\overrightarrow{P_b},v_3,v,v_2,v',w_1,x_1,w_2,v_1,w_4\overleftarrow{P_a}x_2,\overrightarrow{P_c})$ \newline $Q_2=(\overrightarrow{P_b},v_3,v,v_2,v',w_1\overrightarrow{P_a}w_4,v_1,\overrightarrow{P_c})$ \newline
  $Q_3=(\overrightarrow{P_b},\overrightarrow{P_d},v_3,v,v_2,v',w_1,x_1,w_2,v_1,w_4\overleftarrow{P_a}x_2,\overrightarrow{P_c})$  \newline
  $Q_4=(\overrightarrow{P_b},\overrightarrow{P_d},v_3,v,v_2,v',w_1\overrightarrow{P_a}w_4,v_1,\overrightarrow{P_c})$ 
If $v_3v'\in E(G)$, then note that there exists a path $P_f\in\mathbb{P}_1\cup\mathbb{P}_2\cup\mathbb{P}_3$ \u{such that} a vertex in $P_f\cup K$ is adjacent to $v_2$ or $v_3$, otherwise $\{v_2,v,v_3,v'\}$ induces a short cycle.  If $P_f\in\mathbb{P}_2$, then let $P_f=P_b$.  Let $C=(v_2,v,v_3,v',v_2)$.  Now $(\overrightarrow{P_f},\overrightarrow{C},w_1,x_1,w_2,v_1,w_4\overleftarrow{P_a}x_2,\overrightarrow{P_c})$ or $(\overrightarrow{P_f},\overrightarrow{C},w_1\overrightarrow{P_a}w_4,v_1,\overrightarrow{P_c})$ is a desired path.  If $v_3$ is adjacent to a vertex $w''\in\mathbb{P}_1$ such that $w'\neq w''$, then   $(\overrightarrow{P_b},w',v_2,v,v_3,w'',w_1,x_1,w_2,v_1,w_4\overleftarrow{P_a}x_2,\overrightarrow{P_c})$ or $(\overrightarrow{P_b},w',v_2,v,v_3,w'',w_1\overrightarrow{P_a}w_4,v_1,\overrightarrow{P_c})$ is a desired path.
\\ \hline
\end{tabular}  
\end{table}

\noindent
Case 2 is detailed in Tables \ref{tab:HP_P7_P3_2.1}, \ref{tab:HP_P7_P3_2.2}, \ref{tab:HP_P7_P3_2.3}, \ref{tab:HP_P7_P3_2.4}, \ref{tab:HP_P7_P3_2.5}, \ref{tab:HP_P7_P3_2.6}, \ref{tab:HP_P7_P3_2.7}.
This completes a proof of the claim. $\hfill\qed$
\end{proof}
\begin{table}[h!]
\begin{tabular}{ | L{0.1\textwidth} | L{0.89\textwidth} |}
\hline
\textbf{Case} & \textbf{Arguments}\\
\hline
 Case 2 &  There exists $v_1,v_2\in N^I(v)$ such that $v_1w_2,v_2w_3\in E(G)$. We obtain a  desired path $P$ in the following sub cases. \\ \hline
Case 2.1  &  $d(v_3)=1$.  Note $|\mathbb{P}_2|\le1$.  Let $P_b\in \mathbb{P}_2$, $y_1\in P_b\cap K$ and $P_d,P_e\in\mathbb{P}_3$.  Note $y_1$ is adjacent to $v_1$ or $v_2$.  Without loss of generality, let $y_1v_1\in E(G)$.  Further $w_1,w_4$ are adjacent to either $v_1$ or $v_2$.  \\ \hline
Case 2.1.1  & $v_1w_1,v_2w_4\in E(G)$.  $P=(v_3,v,v_2,\overleftarrow{P_a},v_1,y_1\overrightarrow{P_b})$ \\ \hline
Case 2.1.2  & $v_1w_4,v_2w_1\in E(G)$.  $P=(v_3,v,v_2,\overrightarrow{P_a},v_1,y_1\overrightarrow{P_b})$ \\ \hline
Case 2.1.3  & $v_1w_1,v_1w_4\in E(G)$.   Note that for $s_1\in P_d\cap K$, $s_1v_2\in E(G)$ or $s_1x_2\in E(G)$ or $s_1x_3\in E(G)$. \newline 
If $s_1v_2\in E(G)$, then $P=(v_3,v,v_2,\overrightarrow{P_d},\overrightarrow{P_a},v_1,\overrightarrow{P_b})$. \newline
 If $s_1x_2\in E(G)$, then $P=(v_3,v,v_2,w_3,x_3,w_4,v_1,w_1\overrightarrow{P_a}x_2,\overrightarrow{P_d},\overrightarrow{P_b})$. \newline
 If $s_1x_3\in E(G)$, then $P=(v_3,v,v_2,w_3\overleftarrow{P_a},v_1,w_4,x_3,\overrightarrow{P_d},\overrightarrow{P_b})$.
 \\ \hline 
Case 2.1.4  & $v_2w_1,v_2w_4\in E(G)$.   Note that for $s_1\in P_d\cap K$, $s_1v_1\in E(G)$ or $s_1x_1\in E(G)$ or $s_1x_2\in E(G)$. \newline
If $s_1v_1\in E(G)$, then $P=(v_3,v,v_2,\overrightarrow{P_a},\overleftarrow{P_d},v_1,\overrightarrow{P_b})$. \newline
 If $s_1x_1\in E(G)$, then $P=(v_3,v,v_1,w_2\overrightarrow{P_a},v_2,w_1,x_1,\overrightarrow{P_d},\overrightarrow{P_b})$. \newline
 If $s_1x_2\in E(G)$, then $P=(v_3,v,v_1,w_2,x_1,w_1,v_2,w_4\overleftarrow{P_a}x_2,\overrightarrow{P_d},\overrightarrow{P_b})$.
\\ \hline
Case 2.2  & $d(v_3)>1$.  Note $|\mathbb{P}_2|\le2$.  If $|\mathbb{P}_2|=2$, then let $P_b,P_c\in \mathbb{P}_2$.  Since $|I|\ge9$, there exists $P_d\in\mathbb{P}_3$.  Let $P_b=P(y_1;z_1)$, $P_c=P(y_2;z_2)$ and $P_d=P(s_1,s_2;t_1)$.  There exists a vertex $w*$ in $K$ such that $v_3w*\in E(G)$.  \\ \hline
Case 2.2.1  &  $w*\in P_a$.  That is, $v_3w_1\in E(G)$ or $v_3w_4\in E(G)$.  Without loss of generality, let $v_3w_1\in E(G)$.  Note that $w_1$ is adjacent to a vertex in $N^I(w_3)$ and $s_1$ is adjacent to a vertex in $N^I(w_1)$.  We see the following possibilities.\\ \hline
Case 2.2.1.A  &  $w_1v_2,s_1v_1\in E(G)$.  Note that $s_1v_2\in E(G)$ and $w_4$ is adjacent to either $v_1$ or $v_2$. \newline
If $w_4v_1,y_1v_1\in E(G)$, then $P=(\overrightarrow{P_b},v_1,\overleftarrow{P_a},v_3,v,v_2,\overrightarrow{P_d},\overrightarrow{P_c})$. \newline
If $w_4v_2,y_1v_1\in E(G)$, then $P=(\overrightarrow{P_b},v_1,v,v_3,\overrightarrow{P_a},v_2,\overrightarrow{P_d},\overrightarrow{P_c})$. \newline
If $w_4v_1,y_1v_2\in E(G)$, then $P=(\overrightarrow{P_b},v_2,v,v_3,\overrightarrow{P_a},v_1,\overrightarrow{P_d},\overrightarrow{P_c})$. \newline
If $w_4v_2,y_1v_2\in E(G)$, then $P=(\overrightarrow{P_b},v_2,\overleftarrow{P_a},v_3,v,v_1,\overrightarrow{P_d},\overrightarrow{P_c})$. \newline
If $w_4v_1,y_1v_3\in E(G)$, then $P=(\overrightarrow{P_b},v_3,\overrightarrow{P_a},v_1,v,v_2,\overrightarrow{P_d},\overrightarrow{P_c})$. \newline
If $w_4v_2,y_1v_3\in E(G)$, then $P=(\overrightarrow{P_b},v_3,\overrightarrow{P_a},v_2,v,v_1,\overrightarrow{P_d},\overrightarrow{P_c})$. 
\\ \hline
Case 2.2.1.B  &  $w_1v_2,s_1x_1\in E(G)$.  Note that $s_1v_2\in E(G)$ and $w_4$ is adjacent to either $v_2$ or $x_1$. \\ \hline
\end{tabular}  \caption{Case analysis for the proof of Claim \ref{HP_P7_P3} }
\label{tab:HP_P7_P3_2.1}
\end{table}
\begin{table}[h!]
\begin{tabular}{ | L{0.115\textwidth} | L{0.89\textwidth} |}
\hline
\textbf{Case} & \textbf{Arguments}\\
\hline
Case 2.2.1.B  &  
If $w_4x_1,y_1v_1\in E(G)$, then $P=(\overrightarrow{P_b},v_1,w_2\overrightarrow{P_a}w_4,x_1,w_1,v_3,v,v_2,\overrightarrow{P_d},\overrightarrow{P_c})$. \newline
If $w_4v_2,y_1v_1\in E(G)$, then $P=(\overrightarrow{P_b},v_1,w_2\overrightarrow{P_a}w_4,v_2,v,v_3,w_1,x_1,\overrightarrow{P_d},\overrightarrow{P_c})$. \newline
If $w_4x_1,y_1v_2\in E(G)$, then $P=(\overrightarrow{P_b},v_2,w_1,v_3,v,v_1,w_2\overrightarrow{P_a}w_4,x_1,\overrightarrow{P_d},\overrightarrow{P_c})$. \newline
If $w_4v_2,y_1v_2\in E(G)$, then $P=(\overrightarrow{P_b},v_2,w_4\overleftarrow{P_a}w_2,v_1,v,v_3,w_1,x_1,\overrightarrow{P_d},\overrightarrow{P_c})$. \newline
If $w_4x_1,y_1v_3\in E(G)$, then $P=(\overrightarrow{P_b},v_3,v,v_1,w_2\overrightarrow{P_a}w_4,x_1,w_1,v_2,\overrightarrow{P_d},\overrightarrow{P_c})$. \newline
If $w_4v_2,y_1v_3\in E(G)$, then $P=(\overrightarrow{P_b},v_3,v,v_1,w_2\overrightarrow{P_a}w_4,v_2,w_1,x_1,\overrightarrow{P_d},\overrightarrow{P_c})$. 
\\ \hline
Case 2.2.1.C  &  $w_1v_2,s_1x_2\in E(G)$.  Note that $s_1v_2\in E(G)$ or $s_1v_3\in E(G)$.
\\ \hline
Case 2.2.1.C.1  &  $s_1v_2\in E(G)$.  Note that $w_4$ is adjacent to either $v_2$ or $x_2$.  \newline
If $w_4v_2,y_1v_1\in E(G)$, then $P=(\overrightarrow{P_b},v_1,v,v_3,\overrightarrow{P_a},v_2,\overrightarrow{P_d},\overrightarrow{P_c})$. \newline
If $w_4x_2,y_1v_1\in E(G)$, then $P=(\overrightarrow{P_b},v_1,v,v_3,\overrightarrow{P_a}x_2,w_4,x_3,w_3,v_2,\overrightarrow{P_d},\overrightarrow{P_c})$. \newline   
 If $y_1v_2\in E(G)$, then we observe the following.  Let $C=(x_1,w_1,v_3,v,v_1,w_2,x_1)$.  Since we consider the case with no short cycles in $G$, there exists a vertex $m\in K\setminus C$ such that $m$ is adjacent to a vertex in $C\cap I$. \newline
 Consider $m=w_4$.  If $w_4v_2\in E(G)$, then $P=(\overrightarrow{P_b},v_2,C,w_4\overleftarrow{P_a}x_2,\overrightarrow{P_d},\overrightarrow{P_c})$. \newline
 If $w_4x_2\in E(G)$, then $P=(\overrightarrow{P_b},v_2,C,w_3,x_3,w_4,x_2,\overrightarrow{P_d},\overrightarrow{P_c})$.  
\newline Consider $m=s_2$.  Now $P=(\overrightarrow{P_b},v_2,C,s_2,t_1,s_1,x_2\overrightarrow{P_a}w_4,\overrightarrow{P_c})$.   \newline
 If $m$ is the end vertex of a path in $\mathbb{P}_2$; without loss of generality, let $m=y_1$. \newline
 If $w_4v_2\in E(G)$, then $P=(\overrightarrow{P_b},C,v_2,w_4\overleftarrow{P_a}x_2,\overrightarrow{P_d},\overrightarrow{P_c})$. \newline
 If $w_4x_2\in E(G)$, then $P=(\overrightarrow{P_b},C,v_2,w_3,x_3,w_4,x_2,\overrightarrow{P_d},\overrightarrow{P_c})$.  \newline
If $m$ is an end vertex of a path $P_f\in\mathbb{P}_3\cup\mathbb{P}_1$ other than $P_a,P_b,P_c,P_d$, then $P=(\overrightarrow{P_b},v_2,C,\overrightarrow{P_f},\overleftarrow{P_d},x_2\overrightarrow{P_a}w_4,\overrightarrow{P_c})$. 
\newline
If $w_4v_2,y_1v_3\in E(G)$, then $P=(\overrightarrow{P_b},v_3,v,v_1,w_2,x_1,w_1,v_2,w_4\overleftarrow{P_a}x_2,\overrightarrow{P_d},\overrightarrow{P_c})$. \newline
If $w_4x_2,y_1v_3\in E(G)$, then \newline $P=(\overrightarrow{P_b},v_3,v,v_1,w_2,x_1,w_1,v_2,w_3,x_3,w_4,x_2,\overrightarrow{P_d},\overrightarrow{P_c})$. 
\\ \hline 
\end{tabular}\\[10pt]   \caption{Case analysis for the proof of Claim \ref{HP_P7_P3} }
\label{tab:HP_P7_P3_2.2}
\end{table}
\begin{table}[h!]
\begin{tabular}{ | L{0.115\textwidth} | L{0.89\textwidth} |}
\hline
\textbf{Case} & \textbf{Arguments}\\
\hline
Case 2.2.1.C.2  &  $s_1v_3\in E(G)$.  Note that $w_4$ is adjacent to $v_3$ or $x_2$. \newline
If $w_4v_3,y_1v_1\in E(G)$, then $P=(\overrightarrow{P_b},v_1,v,v_2,\overrightarrow{P_a},v_3,\overrightarrow{P_d},\overrightarrow{P_c})$. \newline
If $w_4x_2,y_1v_1\in E(G)$, then \newline $P=(\overrightarrow{P_b},v_1,w_2,x_1,w_1,v_3,v,v_2,w_3,x_3,w_4,x_2,\overrightarrow{P_d},\overrightarrow{P_c})$. \newline
If $w_4v_3,y_1v_2\in E(G)$, then $P=(\overrightarrow{P_b},v_2,w_1,x_1,w_2,v_1,v,v_3,w_4\overleftarrow{P_a}x_2,\overrightarrow{P_d},\overrightarrow{P_c})$. \newline
If $w_4x_2,y_1v_2\in E(G)$, then we see the adjacency of $s_2$.  If $s_2$ is adjacent to $v_1$ or $x_1$, then Case 2.2.1.A or 2.2.1.B could be applied.  Thus we shall consider $s_2x_2\in E(G)$.  Now $P=(\overrightarrow{P_b},v_2,v,v_1,w_2,x_1,w_1,v_3,s_1,t_1,s_2,x_2\overrightarrow{P_a},\overrightarrow{P_c})$.  \newline
If $w_4v_3,y_1v_3\in E(G)$, then note that $w_4$ is adjacent to a vertex in $N^I(w_2)$.\newline
If $w_4x_2\in E(G)$, then $P=(\overrightarrow{P_b},v_3,v,v_1,w_2,x_1,w_1,v_2,w_3,x_3,w_4,x_2,\overrightarrow{P_d},\overrightarrow{P_c})$.  \newline
If $w_4x_1\in E(G)$, then $P=(\overrightarrow{P_b},v_3,w_1,v_2,v,v_1,w_2,x_1,w_4\overleftarrow{P_a}x_2,\overrightarrow{P_d},\overrightarrow{P_c})$.  \newline
If $w_4v_1\in E(G)$, then $P=(\overrightarrow{P_b},v_3,v,v_2,w_1,x_1,w_2,v_1,w_4,x_3,w_3,x_2,\overrightarrow{P_d},\overrightarrow{P_c})$.  \newline
If $w_4x_2,y_1v_3\in E(G)$, then \newline $P=(\overrightarrow{P_b},v_3,v,v_1,w_2,x_1,w_1,v_2,w_3,x_3,w_4,x_2,\overrightarrow{P_d},\overrightarrow{P_c})$. 
\\ \hline
Case 2.2.1.D  &  $w_1x_2,s_1v_1\in E(G)$.  Note that $w_4$ is adjacent to $v_1$ or $x_2$. \newline
If $w_4v_1,y_1v_1\in E(G)$, then  $P=(\overrightarrow{P_b},v_1,w_4,x_3,w_3,v_2,v,v_3,w_1\overrightarrow{P_a}x_2,\overrightarrow{P_d},\overrightarrow{P_c})$. \newline
If $w_4x_2,y_1v_1\in E(G)$, then \newline $P=(\overrightarrow{P_b},v_1,w_2,x_1,w_1,v_3,v,v_2,w_3,x_3,w_4,x_2,\overrightarrow{P_d},\overrightarrow{P_c})$. \newline
If $w_4v_1,y_1v_2\in E(G)$, then $P=(\overrightarrow{P_b},v_2,w_3,x_3,w_4,v_1,v,v_3,w_1\overrightarrow{P_a}x_2,\overrightarrow{P_d},\overrightarrow{P_c})$. \newline
If $w_4x_2,y_1v_2\in E(G)$, then $P=(\overrightarrow{P_b},v_2,w_3,x_3,w_4,x_2\overleftarrow{P_a}w_1, v_3,v,v_1,\overrightarrow{P_d},\overrightarrow{P_c})$. \newline
If $w_4v_1,y_1v_3\in E(G)$, then \newline $P=(\overrightarrow{P_b},v_3,v,v_2,w_3,x_3,w_4,v_1,w_2,x_1,w_1,x_2,\overrightarrow{P_d},\overrightarrow{P_c})$. \newline
If $w_4x_2,y_1v_3\in E(G)$, then \newline $P=(\overrightarrow{P_b},v_3,v,v_2,w_3,x_3,w_4,x_2,w_1,x_1,w_2,v_1,\overrightarrow{P_d},\overrightarrow{P_c})$. 
\\ \hline
Case 2.2.1.E  &  $w_1x_2,s_1x_1\in E(G)$.  Note that $w_4$ is adjacent to $x_1$ or $v_2$. \newline
If $w_4x_1,y_1v_1\in E(G)$, then $P=(\overrightarrow{P_b},v_1,w_2\overrightarrow{P_a}w_4,x_1,w_1,v_3,v,v_2,\overrightarrow{P_d},\overrightarrow{P_c})$. \newline
If $w_4v_2,y_1v_1\in E(G)$, then $P=(\overrightarrow{P_b},v_1,v,v_3,\overrightarrow{P_a},v_2,\overrightarrow{P_d},\overrightarrow{P_c})$. 
\\ \hline
\end{tabular}\\[10pt]  \caption{Case analysis for the proof of Claim \ref{HP_P7_P3} }
\label{tab:HP_P7_P3_2.3}
\end{table}
\begin{table}[h!]
\begin{tabular}{ | L{0.11\textwidth} | L{0.94\textwidth} |}
\hline
\textbf{Case} & \textbf{Arguments}\\
\hline
Case 2.2.1.E  & If $w_4x_1,y_1v_2\in E(G)$, then we observe the following three cases depending on the adjacency of $s_1,s_2$ with $N^I(v)$. 
If $v_1s_1\in E(G)$ or $v_1s_2\in E(G)$, then \newline $P=(\overrightarrow{P_b},v_2,w_3,x_3,w_4,x_1,w_2,x_2,w_1,v_3,v,v_1,\overrightarrow{P_d},\overrightarrow{P_c})$. 
If $v_3s_1\in E(G)$ or $v_3s_2\in E(G)$, then \newline $P=(\overrightarrow{P_b},v_2,w_3,x_3,w_4,x_1,w_1,x_2,w_2,v_1,v,v_3,\overrightarrow{P_d},\overrightarrow{P_c})$.  \newline
If $v_2s_1\in E(G)$ and $v_2s_2\in E(G)$, then  \newline $P=(\overrightarrow{P_b},v_2,s_2,t_1,s_1,x_1,w_1,v_3,v,v_1,x_2\overrightarrow{P_a},\overrightarrow{P_c})$.  \newline
If $w_4v_2,y_1v_2\in E(G)$, then $P=(\overrightarrow{P_b},v_2,w_4\overleftarrow{P_a}w_2,v_1,v,v_3,w_1,x_1,\overrightarrow{P_d},\overrightarrow{P_c})$. \newline
If $w_4x_1,y_1v_3\in E(G)$, then \newline $P=(\overrightarrow{P_b},v_3,v,v_1,w_2,x_2,w_1,x_1,w_4,x_3,w_3,v_2,\overrightarrow{P_d},\overrightarrow{P_c})$. \newline
If $w_4v_2,y_1v_3\in E(G)$, then $P=(\overrightarrow{P_b},v_3,v,v_2,w_2,x_1,w_1,x_2\overrightarrow{P_a}w_4,v_2,\overrightarrow{P_d},\overrightarrow{P_c})$. 
\\ \hline
Case 2.2.1.F  &  $w_1x_2,s_1x_2\in E(G)$.  Note that $s_1v_1\in E(G)$ or $s_1v_2\in E(G)$ or $s_1v_3\in E(G)$.
\\ \hline
Case 2.2.1.F.1  &  $s_1v_1\in E(G)$.  Note that $w_4$ is adjacent to $v_1$ or $x_2$. \newline
If $w_4v_1,y_1v_1\in E(G)$, then $P=(\overrightarrow{P_b},v_1,w_4,x_3,w_3,v_2,v,v_3,w_1\overrightarrow{P_a}x_2,\overrightarrow{P_d},\overrightarrow{P_c})$. \newline
If $w_4x_2,y_1v_1\in E(G)$, then \newline $P=(\overrightarrow{P_b},v_1,w_2,x_1,w_1,v_3,v,v_2,w_3,x_3,w_4,x_2,\overrightarrow{P_d},\overrightarrow{P_c})$. \newline
If $w_4v_1,y_1v_2\in E(G)$, then  $P=(\overrightarrow{P_b},v_2,v,v_3,\overrightarrow{P_a},v_1,\overrightarrow{P_d},\overrightarrow{P_c})$. \newline
If $w_4x_2,y_1v_2\in E(G)$, then $P=(\overrightarrow{P_b},v_2,w_3,x_3,w_4,x_2\overleftarrow{P_a}w_1,v_3,v,v_1,\overrightarrow{P_d},\overrightarrow{P_c})$. \newline
If $w_4v_1,y_1v_3\in E(G)$, then \newline $P=(\overrightarrow{P_b},v_3,v,v_2,w_3,x_3,w_4,v_1,w_2,x_1,w_1,x_2,\overrightarrow{P_d},\overrightarrow{P_c})$. \newline
If $w_4x_2,y_1v_3\in E(G)$, then \newline $P=(\overrightarrow{P_b},v_3,v,v_2,w_3,x_3,w_4,x_2,w_1,x_1,w_2,v_1,\overrightarrow{P_d},\overrightarrow{P_c})$. 
\\ \hline
Case 2.2.1.F.2  &  $s_1v_2\in E(G)$.  Note that $w_4$ is adjacent to $v_2$ or $x_2$.  \newline Let $C=(x_1,w_1,v_3,v,v_1,w_2,x_1)$. \newline
If $w_4v_2,y_1v_1\in E(G)$, then $P=(\overrightarrow{P_b},v_1,\overrightarrow{C}v,v_2,w_4\overleftarrow{P_a}x_2,\overrightarrow{P_d},\overrightarrow{P_c})$. \newline
If $w_4x_2,y_1v_1\in E(G)$, then $P=(\overrightarrow{P_b},v_1,\overrightarrow{C}v,w_3,x_3,w_4,x_2,\overrightarrow{P_d},\overrightarrow{P_c})$. \newline
 If $y_1v_2\in E(G)$, then we observe the following.  Since we consider the case with no short cycles in $G$, there exists a vertex $m\in K\setminus C$ such that $m$ is adjacent to a vertex in $C\cap I$. \newline
 Consider $m=w_4$.  If $w_4v_2\in E(G)$, then $P=(\overrightarrow{P_b},v_2,x_4\overleftarrow{P_a}x_2,C,\overrightarrow{P_d},\overrightarrow{P_c})$. \newline
 If $w_4x_2\in E(G)$, then $P=(\overrightarrow{P_b},v_2,w_3,x_3,w_4,x_2,C,\overrightarrow{P_d},\overrightarrow{P_c})$.  
Consider $m=s_2$.  Now $P=(\overrightarrow{P_b},v_2,s_1,t_1,s_2,C,x_2\overrightarrow{P_a}w_4,\overrightarrow{P_c})$.   \newline
If $m$ is the end vertex of a path in $\mathbb{P}_2$; without loss of generality, let $m=y_1$. 
\\ \hline
\end{tabular}\\[10pt]  \caption{Case analysis for the proof of Claim \ref{HP_P7_P3} }
\label{tab:HP_P7_P3_2.4}
\end{table}
\begin{table}[h!]
\begin{tabular}{ | L{0.11\textwidth} | L{0.94\textwidth} |}
\hline
\textbf{Case} & \textbf{Arguments}\\
\hline
Case 2.2.1.F.2  &  
 If $w_4v_2\in E(G)$, then $P=(\overrightarrow{P_b},C,x_2\overrightarrow{P_a}w_4,v_2,\overrightarrow{P_d},\overrightarrow{P_c})$. \newline
 If $w_4x_2\in E(G)$, then $P=(\overrightarrow{P_b},C,x_2,w_4,x_3,w_3,v_2,\overrightarrow{P_d},\overrightarrow{P_c})$.  \newline
If $m$ is an end vertex of a path $P_f\in\mathbb{P}_3\cup\mathbb{P}_1$ other than $P_d$, then we obtain the following desired paths. If $w_4v_2\in E(G)$, then  $P=(\overrightarrow{P_b},v_2,w_4\overleftarrow{P_a}x_2,C,\overrightarrow{P_f},\overleftarrow{P_d},\overrightarrow{P_c})$. \newline
 If $w_4x_2\in E(G)$, then $P=(\overrightarrow{P_b},v_2,w_3,x_3,w_4,x_2,C,\overrightarrow{P_f},\overleftarrow{P_d},\overrightarrow{P_c})$. 
\newline
If $w_4v_2,y_1v_3\in E(G)$, then $P=(\overrightarrow{P_b},v_3,\overleftarrow{C}v,v_2,w_4\overleftarrow{P_a}x_2,\overrightarrow{P_d},\overrightarrow{P_c})$. \newline
If $w_4x_2,y_1v_3\in E(G)$, then $P=(\overrightarrow{P_b},v_3,\overleftarrow{C}v,v_2,w_3,x_3,w_4,x_2,\overrightarrow{P_d},\overrightarrow{P_c})$. 
\\ \hline
Case 2.2.1.F.3  &  $s_1v_3\in E(G)$.  Note that $w_4$ is adjacent to $v_3$ or $x_2$. \newline
If $w_4v_3,y_1v_1\in E(G)$, then $P=(\overrightarrow{P_b},v_1,v,v_2,w_3,x_3,w_4,v_3,w_1\overrightarrow{P_a}x_2,\overrightarrow{P_d},\overrightarrow{P_c})$. \newline
If $w_4x_2,y_1v_1\in E(G)$, then $P=(\overrightarrow{P_b},v_1,w_2,x_1,w_1,v_3,v,v_2,w_3,x_3,w_4,x_2,\overrightarrow{P_d},\overrightarrow{P_c})$. \newline
If $w_4v_3,y_1v_2\in E(G)$, then $P=(\overrightarrow{P_b},v_2,v,v_1,w_2,x_1,w_1,v_3,w_4\overleftarrow{P_a}x_2,\overrightarrow{P_d},\overrightarrow{P_c})$. \newline
If $w_4x_2,y_1v_2\in E(G)$, then $P=(\overrightarrow{P_b},v_2,w_3,x_3,w_4,x_2,w_1,x_1,w_2,v_1,v,v_3,\overrightarrow{P_d},\overrightarrow{P_c})$. \newline
If $w_4v_3,y_1v_3\in E(G)$, then $P=(\overrightarrow{P_b},v_3,w_4,x_3,w_3,v_2,v,v_1,w_2,x_1,w_1,x_2,\overrightarrow{P_d},\overrightarrow{P_c})$. \newline
If $w_4x_2,y_1v_3\in E(G)$, then $P=(\overrightarrow{P_b},v_3,w_1,x_1,w_2,v_1,v,v_2,w_3,x_3,w_4,x_2\overrightarrow{P_d},\overrightarrow{P_c})$. 
\\ \hline
Case 2.2.1.G  &  $w_1x_3,s_1v_1\in E(G)$.  Note that $s_1x_3\in E(G)$ and $s_2$ is adjacent to a vertex in $N^I(v)$.  Further, we claim that $s_2v_1\in E(G)$, otherwise for some $w\in\{w_1,w_2,w_3,v,s_1\}$, \NC{w}{s_2} or \NC{w}{y_1} induces a \K14.  We also observe that for all paths $Q\in \mathbb{P}_2\cup\mathbb{P}_3$, all the vertices $u\in Q\cap K$ are adjacent to both $v_1$ and $x_3$.  Let $C=(w_1,x_1,w_2,x_2,w_3,v_2,v,v_3,w_1)$.  Since we consider the case with no short cycles in $G$, there exists a vertex $w'\in \mathbb{P}_1$ such that $w'$ is adjacent to a vertex in $C\cap I$.  We obtain $P=(\overrightarrow{P_b},w',C,v_1,s_1,t_1,s_2,x_3,w_4,\overrightarrow{P_c})$.  
\\ \hline
Case 2.2.1.H  &  $w_1x_3,s_1x_1\in E(G)$.  Note that $s_1v_2\in E(G)$.  Similar to the previous case, for all paths $Q\in \mathbb{P}_2\cup\mathbb{P}_3$, all the vertices $u\in Q\cap K$ are adjacent to both $v_2$ and $x_1$.  We obtain $P=(\overrightarrow{P_b},v_2,s_1,t_1,s_2,x_1,w_1,v_3,v,v_1,w_2\overrightarrow{P_a}w_4,\overrightarrow{P_c})$.
\\ \hline
Case 2.2.1.I  &  $w_1x_3,s_1x_2\in E(G)$.  Note that $s_1v_3\in E(G)$.  Similar to the previous cases, for all paths $Q\in \mathbb{P}_2\cup\mathbb{P}_3$, all the vertices $u\in Q\cap K$ are adjacent to both $v_3$ and $x_2$.  We obtain $P=(\overrightarrow{P_b},v_3,s_1,t_1,s_2,x_2,w_3,v_2,v,v_1,w_2,x_1,w_1,x_3,w_4,\overrightarrow{P_c})$. 
\\ \hline
Case 2.2.2 &  $w*\in\mathbb{P}_2$.  Note $|\mathbb{P}_2|\le2$.  Let $P_b,P_c\in\mathbb{P}_2$, $P_b=P(y_1;z_1)$ and $w*=y_1$.  That is, $v_3y_1\in E(G)$.   Note $y_1x_2\in E(G)$.  Further, observe that for all paths $Q\in \mathbb{P}_2\cup\mathbb{P}_3$, all the vertices $u\in Q\cap K$ are adjacent to $x_2$.  Since $|I|\ge9$, there exists $P_d\in\mathbb{P}_3$.  Let $P_d=P(s_1,s_2;t_1)$. \\ \hline
\end{tabular}\\  \caption{Case analysis for the proof of Claim \ref{HP_P7_P3}} 
\label{tab:HP_P7_P3_2.5}
\end{table}
\begin{table}[h!]
\begin{tabular}{ | L{0.085\textwidth} | L{0.935\textwidth} |}
\hline
\textbf{Case} & \textbf{Arguments}\\
\hline
Case 2.2.2 & We see the adjacency of $s_1$ and obtain a desired path $P$ in all possibilities as follows. \newline
 If $s_1v_1\in E(G)$, then $P=(\overrightarrow{P_c},x_2,s_2,t_1,s_1,v_1,w_2,x_1,w_1,w_4,x_3,w_3,v_2,v,v_3,\overrightarrow{P_b})$.  \newline
 If $s_1v_2\in E(G)$, then $P=(\overrightarrow{P_c},x_2,s_2,t_1,s_1,v_2,w_3,x_3,w_4,w_1,x_1,w_2,v_1,v,v_3,\overrightarrow{P_b})$.  \newline
 If $s_1v_3\in E(G)$, then we see adjacency of $w_4$.
 If $w_4v_1\in E(G)$, then  $P=(\overrightarrow{P_c},x_2,s_2,t_1,s_1,v_3,v,v_2,w_3,x_3,w_4,v_1,w_2,x_1,w_1,\overrightarrow{P_b})$.  \newline
 If $w_4v_2\in E(G)$, then note that $x_2w_4\in E(G)$ and $P=(\overrightarrow{P_c},v_3,s_1,t_1,s_2,x_2,w_4,x_3,w_3,v_2,v,v_1,w_2,x_1,w_1,\overrightarrow{P_b})$.  \newline
  If $w_4v_3\in E(G)$, then  $P=(\overrightarrow{P_c},x_2,s_2,t_1,s_1,v_3,w_4,x_3,w_3,v_2,v,v_1,w_2,x_1,w_1,\overrightarrow{P_b})$.   
\\ \hline
Case 2.2.3 &  $w*\in\mathbb{P}_3$.  Let $P_d=P(s_1,s_2;t_1)$ and $w*=s_1$.  That is, $v_3s_1\in E(G)$.  Note $|\mathbb{P}_2|\le2$.  Let $P_b,P_c\in\mathbb{P}_2$, $P_b=P(y_1;z_1)$ and $P_c=P(y_2;z_2)$.  Similar to the previous case observe that $y_1x_2,y_2x_2\in E(G)$.  We see the adjacency of $s_2$ and obtain a desired path $P$ in all possibilities as follows. 
If $s_2v_1\in E(G)$, then we see adjacency of $y_2$.  \newline
If $y_2v_1\in E(G)$, then $P=(\overrightarrow{P_c},v_1,s_2,t_1,s_1,v_3,v,v_2,w_3,x_3,w_4,w_1,x_1,w_2,x_2,\overrightarrow{P_b})$.  \newline
If $y_2v_2\in E(G)$, then $P=(\overrightarrow{P_c},v_2,v,v_3,s_1,t_1,s_2,v_1,w_2,x_1,w_1,w_4,x_3,w_3,x_2,\overrightarrow{P_b})$.  \newline
If $y_2v_3\in E(G)$, then $P=(\overrightarrow{P_c},v_3,s_1,t_1,s_2,v_1,v,v_2,w_3,x_3,w_4,w_1,x_1,w_2,x_2,\overrightarrow{P_b})$.  \newline
If $s_2v_2\in E(G)$, then we see adjacency of $y_2$.  \newline
If $y_2v_1\in E(G)$, then $P=(\overrightarrow{P_c},v_1,v,v_3,s_1,t_1,s_2,v_2,w_3,x_3,w_4,w_1,x_1,w_2,x_2,\overrightarrow{P_b})$.  \newline
If $y_2v_2\in E(G)$, then $P=(\overrightarrow{P_c},v_2,s_2,t_1,s_1,v_3,v,v_1,w_2,x_1,w_1,w_4,x_3,w_3,x_2,\overrightarrow{P_b})$.  \newline
If $y_2v_3\in E(G)$, then $P=(\overrightarrow{P_c},v_3,s_1,t_1,s_2,v_2,v,v_1,w_2,x_1,w_1,w_4,x_3,w_3,x_2,\overrightarrow{P_b})$.  \newline
If $s_2v_3\in E(G)$, then we see adjacency of $w_4$.  \newline
If $w_4v_1\in E(G)$, then $P=(\overrightarrow{P_c},x_2,s_1,t_1,s_2,v_3,v,v_2,w_3,x_3,w_4,v_1,w_2,x_1,w_1,\overrightarrow{P_b})$.  \newline
If $w_4v_2\in E(G)$, then note that $w_4x_2\in E(G)$ and we see the adjacency of $y_2$.  \newline
If $y_2v_1\in E(G)$, then  $P=(\overrightarrow{P_b},x_2,s_1,t_1,s_2,v_3,v,v_2,w_3,x_3,w_4,w_1,x_1,w_2,v_1,\overrightarrow{P_c})$.  \\
\hline
\end{tabular}\\  \caption{Case analysis for the proof of Claim \ref{HP_P7_P3}} 
\label{tab:HP_P7_P3_2.6}
\end{table}
\begin{table}[h!]
\begin{tabular}{ | L{0.085\textwidth} | L{0.935\textwidth} |}
\hline
\textbf{Case} & \textbf{Arguments}\\ \hline 
Case 2.2.3 &
If $y_2v_2\in E(G)$, then  $P=(\overrightarrow{P_b},x_2,s_1,t_1,s_2,v_3,v,v_1,w_2,x_1,w_1,w_4,x_3,w_3,v_2,\overrightarrow{P_c})$.  \newline
If $y_2v_3\in E(G)$, then  $P=(\overrightarrow{P_b},x_2,w_4,x_3,w_3,v_2,v,v_1,w_2,x_1,w_1,s_1,t_1,s_2,v_3,\overrightarrow{P_c})$.  \newline
If $w_4v_3\in E(G)$, then $P=(\overrightarrow{P_c},x_2,s_1,t_1,s_2,v_3,w_4,x_3,w_3,v_2,v,v_1,w_2,x_1,w_1,\overrightarrow{P_b})$.  
\\ \hline
Case 2.2.4 &  $w*\in\mathbb{P}_1$.  Note that $|\mathbb{P}_2|\le2$.  Let $P_b,P_c\in\mathbb{P}_2$, $P_b=P(y_1;z_1)$ and $P_c=P(y_2;z_2)$.  Since $|I|\ge9$, there exists $P_d\in\mathbb{P}_3$.  Let $P_d=P(s_1,s_2;t_1)$.  In this case we shall assume that none of the vertices $y_1,y_2,s_1,s_2$ are  adjacent to $v_3$.  Thus vertices $s_1,y_1,y_2$ are adjacent to either $v_1$ or $v_2$.  It follows that there exists a vertex $u\in \{v_1,v_2\}$ such that $us_1,uy_1\in E(G)$ or $uy_1,uy_2\in E(G)$ or $us_1,uy_2\in E(G)$.  Without loss of generality, let $u=v_1$.  If $v_1s_1,v_1y_1\in E(G)$, then we observe the adjacency of $y_2$.  Note that $y_2$ is adjacent to a vertex in $N^I(w_3)$.    \newline
If $y_2v_2\in E(G)$, then $P=(\overrightarrow{P_b},v_1,\overrightarrow{P_d},\overrightarrow{P_a},w*,v_3,v,v_2,\overrightarrow{P_c})$. \newline
If $y_2x_2\in E(G)$, then $P=(\overrightarrow{P_b},v_1,\overrightarrow{P_d},w*,v_3,v,v_2,w_3,x_3,w_4,w_1\overrightarrow{P_a}x_2,\overrightarrow{P_c})$. \newline
If $y_2x_3\in E(G)$, then $P=(\overrightarrow{P_b},v_1,\overrightarrow{P_d},w*,v_3,v,v_2,w_3\overleftarrow{P_a}w_1,w_4,x_3,\overrightarrow{P_c})$. \newline
When $v_1s_1,v_1y_2\in E(G)$, then the case is symmetric and the desired paths could be obtained similar as above by interchanging the paths $P_c$ and $P_d$. 
Now we consider the case in which $v_1y_1,v_1y_2\in E(G)$.  We shall consider $s_1v_2\in E(G)$.  \newline
If $w_4v_1\in E(G)$, then $P=(\overrightarrow{P_b},v_1,w_4\overleftarrow{P_a},s_2,t_1,s_1,v_2,v,v_3,w*,\overrightarrow{P_c})$. \newline
If $w_4v_2\in E(G)$, then $P=(\overrightarrow{P_b},v_1,v,v_3,w*,w_1\overrightarrow{P_a}w_4,v_2,s_1,t_1,s_2,\overrightarrow{P_c})$. \newline
If $w_4v_3\in E(G)$, then
 $P=(\overrightarrow{P_b},v_1,v,v_2,s_1,t_1,s_2,w_1\overrightarrow{P_a}w_4,v_3,w*,\overrightarrow{P_c})$. \newline 
\\ \hline
\end{tabular}\\  \caption{Case analysis for the proof of Claim \ref{HP_P7_P3}} 
\label{tab:HP_P7_P3_2.7}
\end{table}

\begin{cl}\label{HP_P6_P3}
If there exists $P_a\in\mathbb{P}_6$ and there does not exist $P\in\mathbb{S}_f$ such that $P\neq P_a$ and $|P|\ge 4$, then $G$ has a Hamiltonian path.  
\end{cl}
\begin{proof}
Let $P_a=(w_1,\ldots,w_3;,x_1,\ldots,x_3)$.  
From Corollary \ref{cor1}, the vertices $w_2,w_3$ are adjacent to at least one of the vertices in $N^I(v)$.  Depending on this adjacency, we see the following two cases as shown in Tables \ref{tab:HP_P6_P3_1}, \ref{tab:HP_P6_P3_2}, \ref{tab:HP_P6_P3_3}, \ref{tab:HP_P6_P3_4}.
\begin{table}[h!]
\begin{tabular}{ | L{0.1\textwidth} | L{0.89\textwidth} |}
\hline
\textbf{Case} & \textbf{Arguments}\\
\hline
Case 1 &  There exists $v_1\in N^I(v)$ such that $v_1w_2,v_1w_3\in E(G)$.  Note that all the vertices in $K\setminus P_a$ are adjacent to either $v_1$ or $x_2$.  Since the clique is maximal, there exists $v'\in K$ such that $v_1v'\nin E(G)$.  We further classify based on the possibilities of $v'$ as follows. \\ \hline
Case 1.1  &  $v'\in P_a$; i.e., $w_1v_1\nin E(G)$.  Note that $w_1$ is adjacent to either $v_2$ or $v_3$.  Without loss of generality, let $v_2w_1\in E(G)$.  Clearly, $w_1x_2\in E(G)$.  \\ \hline
Case 1.1.1 & $d(v_3)=1$.  Since $|I|\ge9$, there exists $P_d,P_e\in\mathbb{P}_3$; $P_d=P(s_1,s_2;t_1)$ and $P_e=P(q_1,q_2;r_1)$.  Further, note that $\mathbb{P}_2=\emptyset$.  Recall that the vertices $s_1,s_2,q_1,q_2$ are adjacent to at least one of $v_1,x_2$.  We obtain a desired path as follows. 
 If $s_1v_1,q_1x_2\in E(G)$, then desired path $P_i=(v_3,v,v_2,w_1,x_1,w_2,v_1,\overrightarrow{P_d},\overleftarrow{P_e},x_2,w_3,x_3)$.  
 If $s_1x_2,q_1v_1\in E(G)$, then desired path $P_j=(v_3,v,v_2,w_1,x_1,w_2,v_1,\overrightarrow{P_e},\overleftarrow{P_d},x_2,w_3,x_3)$.  \newline
 If $s_1v_1,q_1v_1\in E(G)$, then note that $s_1$ is adjacent to a vertex in $N^I(w_1)$. \newline
 If $s_1x_2\in E(G)$, then $P_j$ is a desired path. \newline
 If $s_1x_1\in E(G)$, then $P=(v_3,v,v_2,w_1,x_1,\overrightarrow{P_d},\overleftarrow{P_e},v_1,w_2\overrightarrow{P_a})$.  \newline
 If $s_1v_2\in E(G)$, then $P=(v_3,v,v_2,\overrightarrow{P_d},\overleftarrow{P_e},v_1,w_2,x_1,w_1,x_2,w_3,x_3)$.  \newline
 If $s_1x_2,q_1x_2\in E(G)$, then we observe the following.   Note that $s_1v_1\in E(G)$ or $s_1v_2\in E(G)$.  \newline
 If $s_1v_1\in E(G)$, then $P_i$ is a desired path. \newline
  If $s_1v_2\in E(G)$, then $P=(v_3,v,v_1,w_2,x_1,w_1,v_2,\overrightarrow{P_d},\overleftarrow{P_e},x_2,w_3,x_3)$. 
\\ \hline
Case 1.1.2 & $d(v_3)>1$.  Since $|I|\ge9$, there exists $P_d,P_e\in\mathbb{P}_3$; $P_d=P(s_1,s_2;t_1)$ and $P_e=P(q_1,q_2;r_1)$.  Note that $|\mathbb{P}_2|\le1$.  If $|\mathbb{P}_2|=1$, then let $P_b\in\mathbb{P}_2$. Let $P_b=P(y_1;z_1)$.  Recall that the vertices $y_1,s_1,s_2,q_1,q_2$ are adjacent to either $v_1$ or $x_2$.  Since $d(v_3)>1$, there exists a path $Q$ in $\mathbb{P}_3\cup\mathbb{P}_2\cup\mathbb{P}_1$ such that $v_3$ is adjacent to an end vertex of $Q$.  We see the following possibilities.  If $Q\in\mathbb{P}_3$, then without loss of generality we shall assume $Q=P_d$ and $v_3s_1\in E(G)$.  Observe that all the vertices $y_1,q_1,q_2$ are adjacent to $x_2$.  We obtain $P=(\overrightarrow{P_b},x_2,\overrightarrow{P_e},\overleftarrow{P_d},v_3,v,v_2,w_1,x_1,w_2,v_1,w_3,x_3)$ as a desired path.  If $Q\in\mathbb{P}_2$, then $Q=P_b$ and $v_3y_1\in E(G)$.  Observe that all the vertices $s_1,s_2,q_1,q_2$ are adjacent to $x_2$. We obtain the following desired paths depending on the adjacency of $s_1$ with $N^I(v)$.  \newline
If $s_1v_3\in E(G)$, then $P=(\overrightarrow{P_b},x_2,\overleftarrow{P_d},v_3,v,v_2,w_1,x_1,w_2,v_1,w_3,x_3)$.   \newline
If $s_1v_2\in E(G)$, then $P=(\overrightarrow{P_b},v_3,v,v_2,\overrightarrow{P_d},x_2,w_1,x_1,w_2,v_1,w_3,x_3)$.   \newline
If $s_1v_1\in E(G)$, then  $P=(\overrightarrow{P_b},v_3,v,v_2,w_1,x_1,w_2,v_1,\overrightarrow{P_d},x_2,w_3,x_3)$.   \newline
$Q\in\mathbb{P}_1$. For $w'\in\mathbb{P}_1$, $v_3w'\in E(G)$.   
\\ \hline 
\end{tabular}\\[10pt]  \caption{Case analysis for the proof of Claim \ref{HP_P6_P3} }
\label{tab:HP_P6_P3_1}
\end{table}
\begin{table}[h!]
\begin{tabular}{ | L{0.1\textwidth} | L{0.89\textwidth} |}
\hline
\textbf{Case} & \textbf{Arguments}\\
\hline
Case 1.1.2 &
If $s_1v_1,q_1v_1\in E(G)$, then $P=(\overrightarrow{P_b},\overleftarrow{P_d},v_1,\overrightarrow{P_e},w',v_3,v,v_2,\overrightarrow{P_a})$.  \newline
 If $s_1x_2,q_1x_2\in E(G)$, then \newline $P=(\overrightarrow{P_b},\overleftarrow{P_d},x_2,\overrightarrow{P_e},w',v_3,v,v_2,w_1,x_1,w_2,v_1,w_3,x_3)$.   \newline
 If $s_1v_1,q_1x_2\in E(G)$, then \newline $P=(\overrightarrow{P_b},w',v_3,v,v_2,w_1,x_1,w_2,v_1,\overrightarrow{P_d},\overleftarrow{P_e},x_2,w_3,x_3)$.  \newline
  If $s_1x_2,q_1v_1\in E(G)$, then \newline $P=(\overrightarrow{P_b},w',v_3,v,v_2,w_1,x_1,w_2,v_1,\overrightarrow{P_e},\overleftarrow{P_d},x_2,w_3,x_3)$. \\
\hline  
Case 1.2 & $v'\in\mathbb{P}_2$.  In this case we shall assume that $w_1v_1\in E(G)$.  Let $P_b\in\mathbb{P}_2$, $P_b=P(y_1;z_1)$ and  $v_1y_1\nin E(G)$.  We know that either $y_1v_2\in E(G)$ or $y_1v_3\in E(G)$.  Without loss of generality, let $y_1v_2\in E(G)$.  Note that $y_1x_2\in E(G)$.  Further, $d(v_3)>1$.  Observe that $w_1$ is adjacent to a either $v_2$ or $x_2$.  Clearly, $v_3$ is adjacent to an end vertex of a path in $\mathbb{P}_1\cup\mathbb{P}_3$.  Since $|I|\ge9$, there exists $P_d\in\mathbb{P}_3$; $P_d=P(s_1,s_2;t_1)$.  If $v_3$ is adjacent to an end vertex of a path in $\mathbb{P}_3$, then without loss of generality, let $v_3s_1\in E(G)$. 
 Recall that $s_2$ is adjacent to either $v_1$ or $x_2$.  We obtain a desired path $P$ as follows.  \newline
If $s_2v_1\in E(G)$, then $P=(\overrightarrow{P_b},v_2,v,v_3,s_1,t_1,s_2,v_1,\overrightarrow{P_a})$.  \newline
If $s_2x_2\in E(G)$, then $P=(\overrightarrow{P_b},v_2,v,v_3,s_1,t_1,s_2,x_2,w_2,x_1,w_1,v_1,w_3,x_3)$.  \newline
If $v_3$ is adjacent to a vertex $z\in\mathbb{P}_1$, then we obtain the following desired paths.  \newline
If $s_2v_1\in E(G)$, then $P=(\overrightarrow{P_b},v_2,v,v_3,z,s_1,t_1,s_2,v_1,\overrightarrow{P_a})$.  \newline
If $s_2x_2\in E(G)$, then $P=(\overrightarrow{P_b},v_2,v,v_3,z,s_1,t_1,s_2,x_2,w_2,x_1,w_1,v_1,w_3,x_3)$.  
\\ \hline
Case 1.3 & $v'\in\mathbb{P}_3$.  In this case we shall assume that $w_1v_1\in E(G)$.  Let $P_d\in\mathbb{P}_3$, $P_d=P(s_1,s_2;t_1)$ and $v_1s_1\nin E(G)$.  We know that either $s_1v_2\in E(G)$ or $s_1v_3\in E(G)$.  Without loss of generality, let $s_1v_2\in E(G)$.  Note $|\mathbb{P}_2
|\le1$.  If $|\mathbb{P}_2|=1$, then let $P_b\in\mathbb{P}_2$, $P_b=P(y_1;z_1)$.  Since $|I|\ge9$, there exists $P_e\in\mathbb{P}_3$; $P_e=P(q_1,q_2;r_1)$.  If $d(v_3)=1$, then we obtain the following.  Recall that $s_2$ is adjacent to either $v_1$ or $x_2$. \newline
 If $s_2v_1\in E(G)$, then $P=(v_3,v,v_2,s_1,t_1,s_2,v_1,\overrightarrow{P_a})$. \newline
 If $s_2x_2\in E(G)$, then $P=(v_3,v,v_2,s_1,t_1,s_2,x_2,w_2,x_1,w_1,v_1,w_3,x_3)$. \newline
 Now we shall see the case in which $d(v_3)>1$.  Note that $s_1x_2\in E(G)$.  Further, $w_1$ is adjacent to a vertex in $N^I(s_1)$.  Thus $v_3$ is adjacent to $s_2$ or a clique vertex in $\mathbb{P}_1\cup\mathbb{P}_2\cup\mathbb{P}_3$.  \\ \hline
 Case 1.3.1 & 
 If $v_3s_2\in E(G)$, then we obtain the following.  \newline
 If $w_1v_2,y_1v_1\in E(G)$, then \newline $P=(\overrightarrow{P_b},v_1,w_2,x_1,w_1,v_2,v,v_3,s_2,t_1,s_1,x_2,w_3,x_3)$  \newline

\\ \hline 
\end{tabular}  \caption{Case analysis for the proof of Claim \ref{HP_P6_P3} }
\label{tab:HP_P6_P3_2}
\end{table}
\begin{table}[h!]
\begin{tabular}{ | L{0.1\textwidth} | L{0.89\textwidth} |}
\hline
\textbf{Case} & \textbf{Arguments}\\
\hline
Case 1.3.1 & 
  If $w_1v_2,y_1x_2\in E(G)$, then \newline $P=(\overrightarrow{P_b},x_2,s_1,t_1,s_2,v_3,v,v_2,w_1,x_1,w_2,v_1,w_3,x_3)$ \newline
 If $w_1t_1,y_1v_1\in E(G)$, then \newline $P=(\overrightarrow{P_b},v_1,w_2,x_1,w_1,t_1,s_2,v_3,v,v_2,s_1,x_2,w_3,x_3)$  \newline
  If $w_1t_1,y_1x_2\in E(G)$, then \newline $P=(\overrightarrow{P_b},x_2,s_1,v_2,v,v_3,s_2,t_1,w_1,x_1,w_2,v_1,w_3,x_3)$ \newline
If $w_1x_2\in E(G)$, then we observe the following.  Note that $d(v_2)=d(v_3)=d(t_1)=2$.  Let $C=(v,v_2,s_1,t_1,s_2,v_3,v)$.  We consider $G$ with no short cycles.  Therefore, there exists a vertex $z$ in $K\setminus\{w_1,w_2,w_3,v,s_1,s_2\}$ adjacent to a vertex in $C\cap I$.  If $z\in \mathbb{P}_2$; that is, $z=y_1$, then we obtain the following desired paths.  \newline
If $q_2v_1\in E(G)$, then $P=(\overrightarrow{P_b},\overrightarrow{C},\overrightarrow{P_e},v_1,\overrightarrow{P_a})$. \newline
If $q_2x_2\in E(G)$, then $P=(\overrightarrow{P_b},\overrightarrow{C},\overrightarrow{P_e},x_2,w_2,x_1,w_1,v_1,w_3,x_3)$. \newline
If $z\in \mathbb{P}_3$, then without loss of generality, let $z=s_1$, then we obtain the following desired paths.  
If $q_2v_1\in E(G)$, then $P=(\overrightarrow{P_b},\overleftarrow{P_d},\overrightarrow{C},\overrightarrow{P_e},v_1,\overrightarrow{P_a})$. \newline
If $q_2x_2\in E(G)$, then $P=(\overrightarrow{P_b},\overleftarrow{P_d},\overrightarrow{C},\overrightarrow{P_e},x_2,w_2,x_1,w_1,v_1,w_3,x_3)$. 
If $z\in \mathbb{P}_1$, then we see the following.  \newline
If $q_2v_1\in E(G)$, then $P=(\overrightarrow{P_b},z,\overrightarrow{C},\overrightarrow{P_e},v_1,\overrightarrow{P_a})$. \newline
If $q_2x_2\in E(G)$, then $P=(\overrightarrow{P_b},z,\overrightarrow{C},\overrightarrow{P_e},x_2,w_2,x_1,w_1,v_1,w_3,x_3)$. 
\\\hline
Case 1.3.2 & 
 If $v_3$ is adjacent to a vertex in $\mathbb{P}_2$;  that is, $v_3y_1\in E(G)$.  Recall that either $q_2v_1\in E(G)$ or $q_2x_2\in E(G)$.  We obtain the following desired paths.  
 If $q_2v_1\in E(G)$, then $P=(\overrightarrow{P_b},v_3,v,v_2,\overrightarrow{P_d},\overrightarrow{P_e},v_1,\overrightarrow{P_a})$. \newline
 If $q_2x_2\in E(G)$, then $P=(\overrightarrow{P_b},v_3,v,v_2,\overrightarrow{P_d},\overrightarrow{P_e},x_2,w_2,x_1,w_1,v_1,w_3,x_3)$. 
 \\\hline
Case 1.3.3 & 
 If $v_3$ is adjacent to a vertex in $\mathbb{P}_3$.  Without loss of generality, $v_3q_1\in E(G)$.  Recall that either $s_2v_1\in E(G)$ or $s_2x_2\in E(G)$.  \newline
 If $s_2v_1,y_1v_1\in E(G)$, then $P=(\overrightarrow{P_b},v_1,s_2,t_1,s_1,v_2,v,v_3,q_1,r_1,q_2,\overrightarrow{P_a})$. \newline
If $s_2x_2,y_1x_2\in E(G)$, then \newline $P=(\overrightarrow{P_b},x_2,s_2,t_1,s_1,v_2,v,v_3,q_1,r_1,q_2,w_1,x_1,w_2,v_1,w_3,x_3)$. \newline
If $s_2v_1,y_1x_2\in E(G)$, then \newline $P=(\overrightarrow{P_b},x_2\overleftarrow{P_a},v_1,s_2,t_1,s_1,v_2,v,v_3,q_1,r_1,q_2,w_3,x_3)$. \newline
If $s_2x_2,y_1v_1\in E(G)$, then \newline $P=(\overrightarrow{P_b},v_1,w_1\overrightarrow{P_a}x_2,s_2,t_1,s_1,v_2,v,v_3,q_1,r_1,q_2,w_3,x_3)$.   
\\\hline
Case 1.3.4 & 
 If $v_3$ is adjacent to a vertex $z$ in $\mathbb{P}_1$.  
 Recall that either $q_2v_1\in E(G)$ or $q_2x_2\in E(G)$.  We obtain the following desired paths.  \newline
 If $q_2v_1\in E(G)$, then $P=(\overrightarrow{P_b},z,v_3,v,v_2,\overrightarrow{P_d},\overrightarrow{P_e},v_1,\overrightarrow{P_a})$. 
 If $q_2x_2\in E(G)$, then $P=(\overrightarrow{P_b},z,v_3,v,v_2,\overrightarrow{P_d},\overrightarrow{P_e},x_2,w_2,x_1,w_1,v_1,w_3,x_3)$. \\\hline
 \end{tabular}  \caption{Case analysis for the proof of Claim \ref{HP_P6_P3} }
\label{tab:HP_P6_P3_3}
\end{table}
\begin{table}[h!]
\begin{tabular}{ | L{0.1\textwidth} | L{0.89\textwidth} |}
\hline 
\textbf{Case} & \textbf{Arguments}\\
\hline
Case 1.4 & $v'\in\mathbb{P}_1$.   In this case we shall assume that $w_1v_1\in E(G)$.  Note that $v'x_2\in E(G)$.  If $d(v_3)=1$, then $\mathbb{P}_2=\emptyset$.  Note $P=(v_3,v,v_2,v',x_2,w_2,x_1,w_1,v_1,w_3,x_3)$ is a desired path.  If $d(v_3)>1$, then we observe the following.   Note $|\mathbb{P}_2|\le1$.  If $|\mathbb{P}_2|=1$, then let $P_b\in\mathbb{P}_2$, $P_b=P(y_1;z_1)$.  Since $|I|\ge9$, there exists $P_d,P_e\in\mathbb{P}_3$, $P_d=P(s_1,s_2;t_1)$, $P_e=P(q_1,q_2;r_1)$.  Note that $v_3$ is adjacent to $w_1$ or $v'$ or a clique vertex in $\mathbb{P}_1\cup\mathbb{P}_2\cup\mathbb{P}_3$.   If $v_3w_1\in E(G)$, then we obtain the following desired paths. \newline
If $y_1v_1\in E(G)$, then $P_i=(\overrightarrow{P_b},v_1,w_2,x_1,w_1,v_3,v,v_2,v',x_2,w_3,x_3)$. \newline 
If $y_1x_2,s_1x_2\in E(G)$, then \newline $P_j=(\overrightarrow{P_b},x_2,s_1,t_1,s_2,v',v_2,v,v_3,w_1,x_1,w_2,v_1,w_3,x_3)$. \newline 
If $y_1x_2,s_1v_1\in E(G)$, then \newline $P_k=(\overrightarrow{P_b},x_2\overleftarrow{P_a}w_1,v_3,v,v_2,v',s_2,t_1,s_1,v_1,w_3,x_3)$. \newline 
If $v_3v'\in E(G)$, then we see the adjacency of $w_1$.  If $w_1v_3$, then one of $P_i,P_j,P_k$ is  a desired path.  If $w_1v_2\in E(G)$, then $P=(\overrightarrow{P_b},v_1,w_2,x_1,w_1,v_2,v,v_3,v',x_2,w_3,x_3)$ or $P=(\overrightarrow{P_b},x_2\overleftarrow{P_a}w_1,v_2,v',v_3,v,v_1,w_3,x_3)$ is a desired path.  If $w_1x_2\in E(G)$, then note that $d(v_2)=d(v_3)=2$.   Let $C=(v_2,v,v_3,v',v_2)$.  We consider $G$ with no short cycles.  Therefore, there exists a vertex $z$ in $K\setminus\{w_1,w_2,w_3,v,v'\}$ adjacent to a vertex in $C\cap I$.  If $z\in \mathbb{P}_2$; that is, $z=y_1$, then we obtain the following desired paths.  
If $y_1v_2\in E(G)$, then $P=(\overrightarrow{P_b},v_2,v',v_3,v,v_1,\overrightarrow{P_a})$. \newline
If $y_1v_3\in E(G)$, then $P=(\overrightarrow{P_b},v_3,v',v_2,v,v_1,\overrightarrow{P_a})$. \newline
  If $z\in \mathbb{P}_3$, then without loss of generality, let $z=s_2$.  Now we observe the following.  \newline
If $s_2v_2\in E(G)$, then $P=(\overrightarrow{P_b},\overrightarrow{P_d},v_2,v',v_3,v,v_1,\overrightarrow{P_a})$. \newline
If $s_2v_3\in E(G)$, then $P=(\overrightarrow{P_b},\overrightarrow{P_d},v_3,v',v_2,v,v_1,\overrightarrow{P_a})$. \newline
If $z\in \mathbb{P}_1$, then we obtain $P=(\overrightarrow{P_b},z,v_2,v',v_3,v,v_1,\overrightarrow{P_a})$ or
\newline $P=(\overrightarrow{P_b},z,v_3,v',v_2,v,v_1,\overrightarrow{P_a})$ as a desired path.  \newline
Now we consider the case where $v_3$ is adjacent to a clique vertex in $\mathbb{P}_2$; that is, $v_3y_1\in E(G)$.  We obtain $(\overrightarrow{P_b},v_3,v,v_2,v',x_2,w_2,x_1,w_1,v_1,w_3,x_3)$ as a desired path.  If $v_3$ is adjacent to a clique vertex in $\mathbb{P}_3$, then without loss of generality, $v_3s_1\in E(G)$.  We obtain $(\overrightarrow{P_b},\overleftarrow{P_d},v_3,v,v_2,v',x_2,w_2,x_1,w_1,v_1,w_3,x_3)$ as a desired path.  If $v_3$ is adjacent to $z\in\mathbb{P}_1$, then  $(\overrightarrow{P_b},z,v_3,v,v_2,v',x_2,w_2,x_1,w_1,v_1,w_3,x_3)$ is a desired path.
\\\hline
Case 2 &  There exists $v_1,v_2\in N^I(v)$ such that $v_1w_2,v_2w_3\in E(G)$.  A case analysis similar to the previous one could be obtained. 
%
%
%
%
%
\\\hline
\end{tabular}  \caption{Case analysis for the proof of Claim \ref{HP_P6_P3} }
\label{tab:HP_P6_P3_4}
\end{table}
\end{proof}

\begin{cl}\label{NO_3P5}
 If $\mathbb{P}_j=\emptyset,j\ge6$, and $\mathbb{P}_5\cup\mathbb{P}_4\neq\emptyset$, then $|\mathbb{P}_5|+|\mathbb{P}_4|\le2$. 
\end{cl}
\begin{proof}
Let $P_a=P(w_1,\ldots,w_{l};x_1,x_{2})$, $l\in\{2,3\}$, $P_b=P(s_1,\ldots,s_{m};t_1,t_{2})$, $m\in \{2,3\}$.  Assume for a contradiction that there exists a path $P_c=P(q_1,\ldots,q_{n};r_1,r_{2})$, $n\in \{2,3\}$.  From Corollary \ref{xunivthreepathseven}, there exists $v_1\in N^I(v)$ such that $v_1w_2,v_1s_2,v_1q_2\in E(G)$.  We claim that $v_1w_1,v_1s_1,v_1q_1\in E(G)$.  Suppose $v_1w_1\nin E(G)$, then $w_1$ is adjacent to one among $\{v_2,v_3\}$, and one each from $P_b\cap I$ and $P_c\cap I$.  Therefore, \NV{w_1} induces a \K14, a contradiction.  Similar argument holds good for the other edges and thus $v_1s_1,v_1q_1\in E(G)$.  If $P_a,P_b,P_c$ are odd paths, then similar to the previous argument, the end vertices $w_3,s_3,q_3$ are adjacent to $v_1$.  Since the clique is maximal, there exists $w'\in K$ such that $v_1w'\nin E(G)$.  From the previous arguments, $w'\nin \{w_1,\ldots,w_l,s_1,\ldots,s_m,q_1,\ldots,q_n\}$. By Corollary \ref{cor1}, either $v_2w'\in E(G)$ or $v_3w'\in E(G)$.  Further, we argue that $w'$ is adjacent to one of $\{x_1,x_2\}$, one of $\{t_1,t_2\}$ and one of $\{r_1,r_2\}$.  Now, \NV{w'} induces a $K_{1,4}$, which is a final contradiction to the existence of such a path $P_c$.  This completes a proof of the claim.   $\hfill\qed$
\end{proof}
\begin{cl}\label{HP_2P5_P6_P3}
 If there exists $P_a,P_b\in\mathbb{P}_5\cup\mathbb{P}_4$, then $G$ has a Hamiltonian path.
\end{cl}
\begin{proof}
Case analysis is similar to Claim \ref{HP_2P7_2P6_P7P6}.  $\hfill \qed$
\end{proof}
\newpage
\begin{cl}\label{HP_2P5_P6_P3}
 If there exists $P_a\in\mathbb{P}_5$, then $G$ has a Hamiltonian path.
\end{cl}
\begin{cl}\label{HP_2P5_P6_P3}
 If $\mathbb{P}_{j\ge4}=\emptyset$, then $G$ has a Hamiltonian path.
\end{cl}
For the above claims, a case analysis similar to Claim \ref{HP_P7_P3} could be obtained.  
%

\begin{theorem}
Let $G$ be a $K_{1,4}$-free split graph with $|K|\ge|I|-1\ge8$.  Then $G$ has a Hamiltonian path if and only if $G$ has no short $I$-$I$ path, and the sum of the number of $I$-$K$ paths and the number of short cycles is at most 2.  Further, finding such a path is polynomial-time solvable.
\end{theorem}
\begin{proof}
Necessity is trivial.  Sufficiency follows from the previous claims.   $\hfill \qed$
\end{proof}

\begin{figure}[h!]
\begin{center}
\hspace*{-30pt}
\begin{tikzpicture}[scale=1, transform shape]
\footnotesize
 
\NT{5}{0.2}{$K_{1,4}$-free split graph $G$, no short $I$-$I$ path, \# $I$-$K$ path + \# short cycles $\le2$  } 
\draw [line width=0.3mm](0,0) -- (0,-1.51);
\LBLA{0}{-0.5}{1}{-0.5}
\LBLA{0}{-1}{1}{-1}
\LBLA{0}{-1.5}{1}{-1.5}

\NT{5.5}{-0.5}{$K_{1,3}$-free split graph has Hamiltonian path (Theorem \ref{k13hamilpath})}
\NT{4.65}{-1}{$\Delta^I=2$ has Hamiltonian path (Theorem \ref{delta2hamilpath})}
\NT{6.8}{-1.5}{$\Delta^I=3$ (at most one short cycle has Hamiltonian path; Lemmas \ref{onecycle}, \ref{HPoneSC})}

\draw[line width=0.3mm] (1.2,-1.7) -- (1.2,-3.11);
\LBLA{1.2}{-2.1}{2.2}{-2.1}
\LBLA{1.2}{-2.6}{2.2}{-2.6}
\LBLA{1.2}{-3.1}{2.2}{-3.1}
\NT{4.3}{-2.1}{Obtain Transformed graph}
\NT{4.5}{-2.6}{Restricted bipartite subgraph}
\NT{4.05}{-3.1}{Collection of paths $\mathbb{S}_f$}

\NT{3.8}{-3.5}{Path combinations }
\NT{3.75}{-4}{{not possible} in $\mathbb{S}_f$}

\draw [line width=0.3mm,color=red](2.5,-4.1) -- (2.5,-6.01);
\LRA{2.5}{-4.5}{3.5}{-4.5}
\LRA{2.5}{-5}{3.5}{-5}
\LRA{2.5}{-5.5}{3.5}{-5.5}
\LRA{2.5}{-6}{3.5}{-6}

\NT{4}{-3.5-1}{$P_{n\ge 12}$ }
\NT{4.5}{-4-1}{$P_{u}$ and $ P_{j\ge 4}$  }
\NT{4.8}{-4.5-1}{$P_w$, $P_x$  and $P_{j\ge 4}$  }
\NT{4.8}{-5-1}{$P_y$, $P_z$ and $P_{j\ge 4}$  }
\NT{4.8}{-5.5-1}{$u\in\{11,10,9,8\}$ }
\NT{4.5}{-6-1}{$w,x\in \{7,6\}$ }
\NT{4.5}{-7.5}{$y,z\in \{5,4\}$ }

\draw [line width=0.3mm,color=blue] (6.5,-3.65) -- (6.5,-7.51);
\LBA{6.5}{-4}{7.5}{-4}
\LBA{6.5}{-4.5}{7.5}{-4.5}
\LBA{6.5}{-5}{7.5}{-5}
\LBA{6.5}{-5.5}{7.5}{-5.5}
\LBA{6.5}{-6}{7.5}{-6}
\LBA{6.5}{-6.5}{7.5}{-6.5}
\LBA{6.5}{-7}{7.5}{-7}
\LBA{6.5}{-7.5}{7.5}{-7.5}

\NT{8.5}{-3.5}{Other Path combinations in $\mathbb{S}_f$ }
\NT{8.5+0.4}{-4}{$P_{u}$, $P_3$\alert{$^*$}, $P_2$\b{$^{**}$}, $P_1$\alert{$^*$} }
\NT{8.8+0.4}{-4.5}{$P_{w}$, $P_x$, $P_3$\alert{$^*$}, $P_2$\b{$^{**}$}, $P_1$\alert{$^*$} }
\NT{8.8+0.4}{-5}{$P_{w}$, $P_5$, $P_3$\alert{$^*$},$P_2$\b{$^{**}$}, $P_1$\alert{$^*$} }
\NT{8.8+0.4}{-5.5}{$P_{w}$, $P_4$, $P_3$\alert{$^*$},$P_2$\b{$^{**}$}, $P_1$\alert{$^*$} }
\NT{8.5+0.4}{-6}{$P_{w}$, $P_3$\alert{$^*$},$P_2$\b{$^{**}$}, $P_1$\alert{$^*$} }
\NT{8.8+0.4}{-6.5}{$P_{y}$, $P_z$, $P_3$\alert{$^*$},$P_2$\b{$^{**}$}, $P_1$\alert{$^*$} }
\NT{8.5+0.4}{-7}{$P_{y}$, $P_3$\alert{$^*$},$P_2$\b{$^{**}$}, $P_1$\alert{$^*$} }
\NT{8.3+0.4}{-7.5}{ $P_3$\alert{$^*$},$P_2$\b{$^{**}$}, $P_1$\alert{$^*$} }

\NT{12.1}{-7}{ \b{$**$} - \# $I$-$K$ paths $\le2$}
\NT{12.1}{-7.5}{ \alert{$*$} - zero or more paths}

\NT{10.1}{-8}{ All these combinations have Hamiltonian paths }

	\end{tikzpicture}
\end{center}  \vspace{-10pt}
\caption{An overview of the algorithm for finding Hamiltonian path in $K_{1,4}$-free split graphs}
\label{Steinerfig3}
\end{figure}
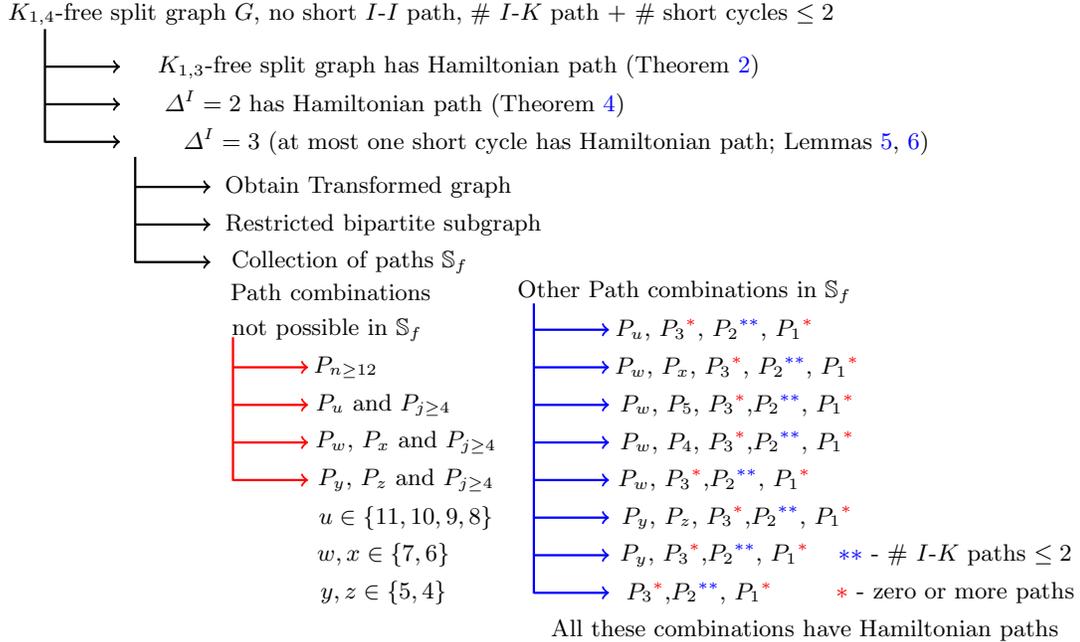

\section{Hardness Result}
\label{sec4.3}
T.Akiyama et. al. \cite{akiyama} proved the NP-completeness of Hamiltonian cycle in planar bipartite graphs with maximum degree $3$.  In \cite{dicho1}, Hamiltonian cycle problem in planar bipartite graphs with maximum degree $3$ is  reduced to Hamiltonian cycle problem in $K_{1,5}$-free split graph.  An in depth analysis of the reduction reveals that the reduced instances are split graphs with $\Delta^I\le3$.  We show a polynomial-time reduction from Hamiltonian cycle problem in split graphs with $\Delta^I\le3$ to Hamiltonian path problem in split graphs with $\Delta^I\le3$.  Further note that such graphs are sub class of \K15-free split graphs.  In the next theorem, we prove that Hamiltonian path problem in $K_{1,5}$-free split graph is NP-complete. 
\begin{theorem} \label{k15hp}
Unless P=NP, there is no polynomial-time algorithm for Hamiltonian path problem in $K_{1,5}$-free split graphs.
\end{theorem}
\begin{proof}
We reduce Hamiltonian cycle problem in split graphs with $\Delta^I\le3$ to Hamiltonian path problem in $K_{1,5}$-free split graphs as follows.  For a given instance of split graph $G$ with $\Delta^I\le3$, having partitions $K$ and $I$, we create $m$ instances of $K_{1,5}$-free split graphs $G_j, 1\le j\le m $ with partitions $K_j$ and $I_j$ where $m=|E=\{uv : u\in K, v\in I\}|$.  That is, corresponding to each edge $uv$ in $G$ such that $u\in K, v\in I$, $G_j$ is constructed as follows: 
\begin{center}
$K_j=K\cup \{z\}, I_j=I\cup \{s,t\}, E'=\{zw : w\in K\},$ \\
$E(G_j)=E(G)\cup E'\cup \{zt,zv,us\}\setminus \{uv\}.$  
\end{center}
It is easy to see that the reduction produces linear number of graph instances, each in linear time. 
Each of the instances created has one more clique vertex $z$, with $d^I(z)=2$. Further, $d^I_{G_j}(u)=d^I_G(u)$.
Therefore, all such instances are having $\Delta^I\le3$, and thus $G_j$ is a $K_{1,5}$-free split graph.
Now we claim that there exists a Hamiltonian cycle in $G$ if and only if there exists a Hamiltonian path in at least one of the reduced graphs $G_j$.  
The necessary is trivial.  For sufficiency, consider a Hamiltonian path $P$ in $G_j$, $j\in \{1,\ldots,m\}$.  Clearly, $P$ is a $(s,t)$-Hamiltonian path.
%
Path $P$ contains the edges $zt$ and $us$.  Further, depending on the position of $v$ in the path $P$, we observe the following possibilities.  $P$ is of the form $(t,z,v,\ldots,u,s)$ or $(t,z,z_1,\ldots,z_k,v,z_{k+1},\ldots,u,s)$, where $\{z_1,z_k,z_{k+1}\}\subset K$.  If $P=(t,z,v,\ldots,u,s)$, the $(v,\ldots,u,v)$ is a Hamiltonian cycle in $G$.  If $P=(t,z,z_1,\ldots,z_k,v,z_{k+1},\ldots,u,s)$, then note that $(t,z,v,z_k,\ldots,z_1,z_{k+1},$ $\ldots,u,s)$ is also a Hamiltonian path in $G_j$, and $(v,z_k,\ldots,z_1,z_{k+1},\ldots,u,v)$ is a Hamiltonian cycle in $G$.
%
%
\begin{figure}[h!]
\begin{center}
\includegraphics[scale=1]{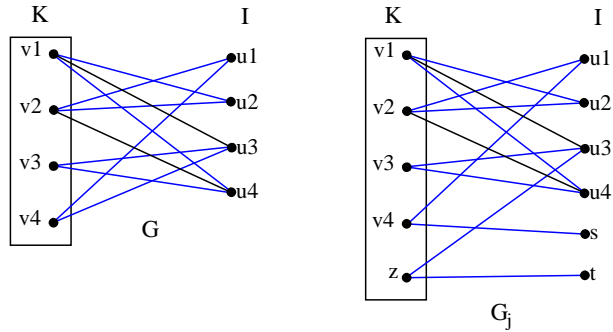}
\vspace*{-5pt}
\caption{An illustration of reduction of Hamiltonian cycle problem to Hamiltonian path problem in split graphs}
\label{eg:HC-reduction} 
\end{center}
\end{figure}
Thus this reduction establishes the fact that Hamiltonian path problem in \K15-free split graph is NP-hard. 
Therefore, we conclude that unless P=NP, there is no polynomial-time algorithm for Hamiltonian path problem in \K15-free split graphs.  This completes a proof of the theorem. $\hfill \qed$
\end{proof}
An illustration for the Hamiltonian path reduction is shown in Figure \ref{eg:HC-reduction}.  Note that the graph $G$ shown in the figure is having a Hamiltonian cycle $(v_1,u_2,v_2,u_1,$ $v_4,u_3,v_3,u_4)$.  Observe that the reduction produces $m$ new graphs.  One of the reduced graph $G_j$ is also shown which is obtained by introducing vertices $\{z,s,t\}$ and edges $\{v_4s,u_3z,zt\}$, and excluding the edge $v_4u_3$.  Here we obtain a Hamiltonian path $(t,z,u_3,v_3,u_4,v_1,u_2,v_2,$ $u_1,v_4,s)$ in $G_j$.

\section{Conclusion and Future work}
We produced a dichotomy on the Hamiltonian path problem in split graphs.  A natural extension is to study longest path problem and minimum-leaf spanning tree problem which are generalizations of Hamiltonian path problem.  Arti Pandey and B.S. Panda have showed in \cite{artipanda} that the total outer-connected dominating set is NP-complete in split graphs.  It is interesting to see that the reduction instances are $K_{1,5}$-free split graphs.  Thus, the complexity of total outer-connected dominating set in $K_{1,4}$-free split graphs remain open.  


\bibliographystyle{splncs1}
\bibliography{references}

\end{document}